\documentclass{article} %
\usepackage[final]{colm2026_conference}

\usepackage{microtype}
\usepackage[colorlinks=true, allcolors=blue]{hyperref}
\usepackage{url}
\usepackage{booktabs}
\usepackage{centernot}

\usepackage{amsmath}
\usepackage{caption}
\usepackage{graphicx}
\usepackage{subcaption}
\usepackage{natbib}
\usepackage{xspace}
\usepackage{mathtools}
\usepackage{stmaryrd}
\usepackage{dutchcal}
\usepackage{amssymb}
\usepackage{amsthm}
\usepackage{pgfplots}
\usepackage{amsmath}
\usepackage{amssymb}
\usepackage{amsthm}
\usepackage{physics}
\usepackage{hyperref}
\usepackage{enumitem}
\usepackage{thmtools}
\usepackage{cleveref}
\usepackage{xcolor}
\usepackage{tcolorbox}
\usepackage{algorithm}
\usepackage{calc}
\usepackage{algpseudocode}
\usepackage{colortbl}
\usepackage{amsthm}
\setlist[itemize]{nosep}

\usepackage{tikz}
\usetikzlibrary{arrows,backgrounds,positioning,shapes,decorations,decorations.pathreplacing,automata,fit}
\usepackage{relsize}
\tikzset{fontscale/.style = {font=\relsize{#1}}}
\usetikzlibrary{calc}
\newcounter{theorem}
\newtheorem{proposition}[theorem]{Proposition}
\newtheorem{corollary}[theorem]{Corollary}

\newtheorem{example}[theorem]{Example}
\newtheorem{defin}[theorem]{Definition}
\newtheorem{definition}[theorem]{Definition}
\newtheorem{remark}[theorem]{Remark}
\newtheorem{lemma}[theorem]{Lemma}
\newtheorem{thm}[theorem]{Theorem}

\newcommand{\macro}[1]{{#1}}

\def\im{\macro{\mathrm{Im}}}

\newcommand{\extension}[2]{\macro{\mathrm{Ext}(#2,#1)}}

\newcommand{\N}{\macro{\mathbb{N}}}%

\newcommand{\CRASP}{\ensuremath{\macro{\mathsf{C\text{-}RASP}}}\xspace}

\newcommand{\vty}[1]{\macro{\mathbf{#1}}}

\newcommand{\Z}{\ensuremath{\mathbb{Z}}\xspace}

\newcommand{\twreath}{\macro{\mathrlap{\hspace{.25ex}\cdot}{\circ}}}
\newcommand{\wrc}{\macro{\circ}}
\newcommand{\wrt}{\macro{\mathrlap{\hspace{.25ex}\cdot}{\circ}}}

\makeatletter
\newcommand{\oset}[2]{%
  {\mathop{#2}\limits^{\vbox to -.5\ex@{\kern-\tw@\ex@
   \hbox{\scriptsize $#1$}\vss}}}}
\makeatother
\newcommand{\countl}{\macro{\ensuremath{\oset{\leftharpoonup}{\#}}}}

\newcommand{\automaton}{\macro{\mathcal{A}}}

\newcommand{\obj}{\macro{\text{Obj}}}
\newcommand{\relml}{\macro{\mathrel{\triangleleft}}}

\newcommand{\types}[1]{\macro{{\mathfrak{T}_{#1}}}}
\newcommand{\type}[1]{\macro{{\mathfrak{#1}}}}
\newcommand{\units}[1]{\macro{{\mathcal{E}_{#1}}}}
\newcommand{\unit}[1]{\macro{{\mathcal{#1}}}}

\newcommand{\comp}[1]{\macro{\ensuremath{\overline{#1}}}}

\newcommand{\sub}{\macro{\star}}

\newcommand{\since}{\mathbin{\macro{\mathbf{since}}}}

\newcommand{\dyck}{\macro{\mathcal{D}}}

\newcommand{\uone}{\macro{U_1}}
\newcommand{\utwo}{\macro{U_3}}
\newcommand{\uthree}{\macro{U_2}}

\newcommand{\lang}{\macro{L}}

\newcommand{\arrow}[1]{\xrightarrow{#1}}
\renewcommand{\arrow}[1]{\macro{\rightarrow_{#1}\,}}

\newcommand{\Hom}{\macro{\text{Hom}}}

\definecolor{blu}{HTML}{1E88E5}
\definecolor{re}{HTML}{D81B60}
\definecolor{yel}{HTML}{FFC107}
\definecolor{gre}{HTML}{004D40}

\newcommand{\red}[1]{\textcolor{re}{#1}}
\newcommand{\blu}[1]{\textcolor{blu}{#1}}
\newcommand{\yel}[1]{\textcolor{yel}{#1}}

\newcommand{\sympred}[1]{\macro{#1}}

\newcommand{\bnfto}{\mathrel{::=}}

\newcommand{\str}[1]{\macro{#1}}

\newcommand{\ptl}{\macro{linear RASP program}}

\newcommand{\ptls}{\macro{linear RASP programs}}
\newcommand{\Ptls}{\macro{Linear RASP programs}}

\newcommand{\bos}{\macro{\texttt{<BOS>}}}
\newcommand{\eos}{\macro{\texttt{<EOS>}}}

\newcommand{\wwp}{\macro{\textnormal{wp}}}
\newcommand{\wwpc}{\macro{\textnormal{wpc}}}

\newcommand{\extr}[1]{\automaton{\restriction}#1}

\usepackage{lineno}

\definecolor{darkblue}{rgb}{0, 0, 0.5}
\hypersetup{colorlinks=true, citecolor=darkblue, linkcolor=darkblue, urlcolor=darkblue}

\title{Algebraic Decomposition Theory for Transformer\\ Length Generalization}

\author{Andy Yang$^{1*}$ \quad Blerta Veseli$^{2}$\thanks{AY and BV are co-first authors.
Contact: \texttt{ayang4@nd.edu},\texttt{\{blerta.veseli,mhahn\}@uni-saarland.de}} 
\quad Corentin Barloy$^{3}$ \quad Michaël Cadilhac$^{4}$ \AND Andreas Krebs$^{5}$ \quad Charles Paperman$^{6}$ \quad Howard Straubing$^{7}$ \quad Michael Hahn$^{2}$
\\[1ex]
$^{1}$University of Notre Dame, $^{2}$Saarland University, $^{3}$Ruhr University Bochum \\ $^{4}$DePaul University, $^{5}$University of Tübingen, $^{6}$University of Lille, $^{7}$Boston College  \\}

\begin{document}

\ifcolmsubmission
\linenumbers
\fi

\newpage
\maketitle
\begin{abstract}
    Transformer-based language models are known to sometimes generalize to sequences longer than seen during training, but we lack a precise characterization of which tasks admit length generalization. 
    It is not even known which \emph{regular languages} transformers length-generalize on -- and this is a foundational class of languages.
    Our contributions are to establish the first complete characterization of which regular languages transformers length-generalize on and provide a decision algorithm running in polynomial time in the size of the language's syntactic monoid. 
    These results rely on an effective characterization of the regular languages in C-RASP, a recently-established formalism that expresses which languages transformers length-generalize on. 
    This characterization is challenging because classical tools like Krohn-Rhodes decomposition theory for finite semigroups are insufficient for C-RASP.
    Firstly, the basic building blocks of Krohn-Rhodes theory -- flip-flop and simple groups -- are not expressible in C-RASP.
    Secondly, the basic building block of C-RASP (unbounded counting) is not expressible by the finite semigroups of Krohn-Rhodes theory.
    Thus, \emph{length generalization on regular languages is controlled by an algebraic property that is invisible to classical finite decomposition theory.}
    We generalize classical decomposition theory from finite semigroups to the infinite additive group on the integers, allowing us to characterize C-RASP in terms of iterated wreath products of the integers and derive a provable polynomial-time decision algorithm for regular language membership.
    Experiments across a broad test suite of regular languages confirm that our theory captures transformers' length-generalization behavior more accurately than existing classifications.\footnote{Code available at \href{https://github.com/bveseli/state-tracking-crasp}{GitHub repository}.}
\end{abstract}

\section{Introduction}
What state-tracking capabilities do transformer language models possess? 
To answer this question we study the \emph{regular languages} on which transformers can \emph{length-generalize}. 
First, regular languages give us a formal framework to describe the structure of state-tracking algorithms that transformers can implement. 
Second, probing for length-generalization gives us an empirical confirmation that the transformer learns an implementation of the underlying algorithm. 
Even though this is a fundamental question about an important architecture, we still lack a precise answer. 
For instance as shown in \cref{fig:minimal_pair}, two structurally similar state-tracking tasks (recognizing $(ab+bbaa)^*$ and $(ab+aabb)^*$) may diverge sharply in terms of length-generalizability -- and all existing theory fails to explain this discrepancy.

\begin{figure}
    \centering
    \tikzstyle{state}=[circle, draw, minimum size=0.5cm]
\begin{minipage}[b]{0.45\textwidth}
     \begin{center}
        \begin{tikzpicture}[->,>=stealth',auto,semithick, initial text=, node distance=3ex]
            \node[state, initial above, accepting] (q0) at (0,0) {};
            \node[state, below=.55 of q0] (q1) {};
            \node[state, below left=1 of q0] (q2) {};
            \node[state, below right=1 of q2] (q3) {};
            \node[state, above right=1 of q3] (q4) {};
            \path (q0) edge [bend right=25] node[left, midway] {a} (q1);
            \path (q1) edge [bend right=25] node[right, midway] {b} (q0);
            \path (q0) edge [bend right=25] node[midway, above] {b} (q2);
            \path (q2) edge [bend right=25] node[midway, below] {b} (q3);
            \path (q3) edge [bend right=25] node[midway, below] {a} (q4);
            \path (q4) edge [bend right=25] node[midway, above] {a} (q0);
        \end{tikzpicture}
        \tiny
        \begin{tikzpicture}[x=7mm, y=5mm]
            \draw[thick] (0,0) rectangle (4,4);

            \draw[opacity=0.1] (.5,0) -- (.5,4);
            \draw[opacity=0.1] (2,0) -- (2,4);
            \draw[opacity=0.1] (3.5,0) -- (3.5,4);

            \draw[opacity=0.1] (0,2) -- (4,2);
            
            \node[left] at (0,0) {0};
            \node[left] at (0,1) {25};
            \node[left] at (0,2) {50};
            \node[left] at (0,3) {75};
            \node[left] at (0,4) {100};
            
            \node[fill=green!70!brown, regular polygon, regular polygon sides=400, minimum size=1ex, scale=0.75] at (0.5,4) (p1) {};
    
            \node[fill=green!70!brown, regular polygon, regular polygon sides=400, minimum size=1ex, scale=0.75] at (2,4) (p2) {};
    
            \node[fill=green!70!brown, regular polygon, regular polygon sides=400, minimum size=1ex, scale=0.75] at (3.5,3.71) (p3) {};
    
            \draw[green!70!brown,  thick] (p1) -- (p2);
            \draw[green!70!brown,  thick] (p2) -- (p3);
    
            \node[below] at (.5,0) {[$l_{min}$,50]};
    
            \node[below] at (2,0) {[51,100]};
    
            \node[below] at (3.5,0) {[101,150]};

        \end{tikzpicture}
        \normalsize

    \end{center}   
\end{minipage}%
\begin{minipage}[b]{0.45\textwidth}
     \begin{center}
        \begin{tikzpicture}[->,>=stealth',auto,semithick, initial text=, node distance=3ex]
            \node[state, initial above, accepting] (q0) at (0,0) {};
            \node[state, below left=1 of q0] (q2) {};
            \node[state, below right=1 of q2] (q3) {};
            \node[state, above right=1 of q3] (q4) {};

            \path (q0) edge [bend right=25] node[midway,above] {a} (q2);
            \path (q2) edge [bend right=25] node[midway,below] {b} (q0);
            \path (q2) edge [bend right=25] node[below] {a} (q3);
            \path (q3) edge [bend right=25] node[below] {b} (q4);
            \path (q4) edge [bend right=25] node[above] {b} (q0);
        \end{tikzpicture}
        \tiny
        \begin{tikzpicture}[x=7mm, y=5mm]
            \draw[thick] (0,0) rectangle (4,4);
    
            \node[left] at (0,0) {0};
            \node[left] at (0,1) {25};
            \node[left] at (0,2) {50};
            \node[left] at (0,3) {75};
            \node[left] at (0,4) {100};

            \draw[opacity=0.1] (.5,0) -- (.5,4);
            \draw[opacity=0.1] (2,0) -- (2,4);
            \draw[opacity=0.1] (3.5,0) -- (3.5,4);

            \draw[opacity=0.1] (0,2) -- (4,2);
            
            \node[fill=red!70!brown, regular polygon, regular polygon sides=300, minimum size=1ex, scale=0.75] at (0.5,4) (p1) {};
    
            \node[fill=red!70!brown, regular polygon, regular polygon sides=300, minimum size=1ex, scale=0.75] at (2,3.79) (p2) {};
    
            \node[fill=red!70!brown, regular polygon, regular polygon sides=300, minimum size=1ex, scale=0.75] at (3.5,1.184) (p3) {};
    
            \draw[red!70!brown,  thick] (p1) -- (p2);
            \draw[red!70!brown,  thick] (p2) -- (p3);
    
            \node[below] at (.5,0) {[$l_{min}$,50]};
    
            \node[below] at (2,0) {[51,100]};
    
            \node[below] at (3.5,0) {[101,150]};

        \end{tikzpicture}
        \normalsize

    \end{center}   
\end{minipage}
\tikzstyle{state}=[circle, draw, minimum size=1cm]
    \caption{DFAs for $(ab+bbaa)^*$ and $(ab+aabb)^*$ and transformer length generalization on both. 
    Transformers were trained on strings of length $[l_{min},50]$, where $l_{min}$ denotes the length of the shortest valid string in the respective language, and the tested on a held-out sets of strings from bins of length up to $150$. Existing theory fails to explain why length-generalization differs between these simple regular languages with very similar structure.}
    \label{fig:minimal_pair}
\end{figure}
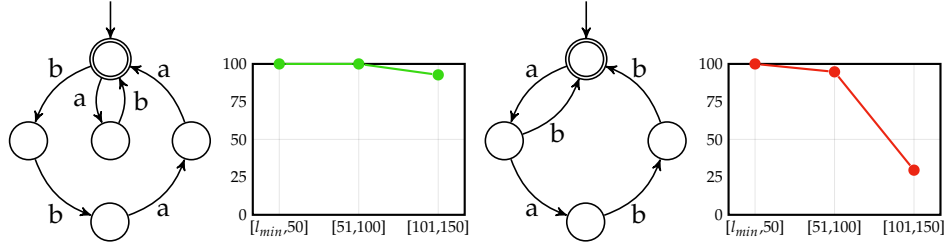

This line of inquiry has a deep history in machine learning. 
In fact, \citet{Kleene+2016+3+42} invented regular expressions specifically for the purpose of analyzing the capabilities of McCulloch-Pitts neural networks.
Regular language recognition is a fundamental task in the theory of computation that provides us a systematic method of determining the sequence-processing capabilities of a model \citep{sipser1996intro}.
The task is simple: given a sequence of inputs, each of which triggers a transition in a deterministic finite automaton (DFA), track the state of the machine after each transition.
These finite state-tracking tasks arise concretely in the context of language models, which we sketch below.

\textbf{Constrained decoding:} 
A user may want a language model to generate text in a structured format, such as JSON or code, which can pose challenges \citep{schall-de-melo-2025-hidden}. 
Even as outputs grow in length, their structured formats often follow regular constraints (e.g. matching curly braces to a fixed depth in JSON).

\textbf{Agentic workflows:} AI agents in practice follow workflows consisting of compositional sequences of actions \citep{schluntz2024buildingeffectiveagents}.
Keeping track of the state of the environment and deciding the next action to take can then explicitly be seen as simulating a DFA.  

\textbf{Natural language:} 
Morphemes in natural languages typically follow constraints that can be modeled by regular languages \citep{kaplan-kay-1994-regular}. 
Natural language semantics also involves maintaining the states of different referents over time \citep{kim-schuster-2023-entity}.

Previous work provides a theoretical foundation for our study by providing characterizations of which regular languages can be \emph{expressed} by the transformer architecture. 
\Citet{liu2023transformers,merrill2025a} showed transformers with $\log(n)$ depth could express all regular languages, while constant depth transformers could express all solvable regular languages.
The containment of $\mathsf{poly}(n)$-precision transformers within $\mathsf{TC}^0$ suggests transformers cannot express non-solvable regular languages (assuming $\mathsf{TC}^0\neq \mathsf{NC}^1$, as is common) \citep{NEURIPS2023_a48e5877,chiang2025transformers}.
\citet{10.1162/tacl_a_00306} showed that hard attention transformers recognized only languages in $\mathsf{AC}^0$, which
\citet{yang2024masked} later refined to the star-free regular languages, and \citet{jerad-etal-2025-unique} finally refined to the $\mathcal{R}$-trivial languages (when using leftmost tie breaking). 
Similarly, \citet{li2024chain} showed that finite-precision transformers recognized only star-free languages, which \citet{li2025characterizing} refined to the $\mathcal{R}$-trivial languages.

In contrast, \citet{bhattamishra-etal-2020-ability, huang2025formalframeworkunderstandinglength} showed that transformers can \emph{learn} languages both in and outside of each class discussed above, implying that \emph{existing expressivity characterizations do not account for transformer length generalization on regular languages}.
In particular, strong empirical evidence show that transformers tend to length-generalize on and only on the languages expressible in $\CRASP$, a programming language defining a subclass of the languages expressible by a transformer \citep{huang2025formalframeworkunderstandinglength,jobanputra2025born,yang2025knee,yang2024counting}.
However, prior to this work, there did not exist a complete characterization of the regular languages in $\CRASP$.

Formal language theory has a vibrant tradition of producing beautiful characterizations of regular language membership for different classes. 
For instance, \citet{MIXBARRINGTON1992478} showed that a regular language is in $\mathsf{AC}^0$ if and only if its syntactic morphism is quasi-aperiodic, thus providing a decision procedure.
\citet{10.1007/3-540-07407-4_23} showed that a language is piece-wise testable iff its syntactic monoid is $\mathcal{J}$-trivial.
On the frontier, regular language membership in $\mathsf{TC}^0$ is equivalent to the $40$ year old open problem of whether or not $\mathsf{TC}^0\neq \mathsf{NC}^1$ \citep{BARRINGTON1989150}, and regular language membership in all levels of the dot-depth hierarchy is a $50$ year old open problem \citep{pin:hal-01614357}.
Our work follows in this tradition, tackling the same question in the case of $\CRASP$. 

This is a deep and challenging theoretical question because classical tools for characterizing regular languages -- namely the Krohn-Rhodes decomposition theory of \citet{krohn1965algebraic} -- are insufficient for $\CRASP$.
On one front, the building blocks of Krohn-Rhodes theory (flip-flop units and simple groups) are not expressible in $\CRASP$ \citep{huang2025formalframeworkunderstandinglength}.
On a second front, the basic building blocks of $\CRASP$ (unbounded counting units) are not expressible by the finite semigroups of Krohn-Rhodes theory. 
Our core innovation is to develop an analogous algebraic decomposition theory for transformers using the additive group on the integers. 
We effectively characterize the regular languages in $\CRASP$, and provide a polynomial-time algorithm that decides DFA membership in $\CRASP$.
A simpler necessary (but not sufficient) criterion is proven via a profinite equation. 
Finally, we validate empirically that regular language membership in $\CRASP$ predicts transformer length-generalization better than any existing characterization.

\begin{figure}
    \centering

    \begin{subfigure}[t]{0.32\linewidth}
        \centering
        \begin{tikzpicture}
            \draw[thick, fill=orange, fill opacity=0.1] (-1.2,0) rectangle (1.2,1.3);
            \node (crasp) at (-1.9,0.75) {\textcolor{orange}{$\CRASP$}};

            \draw[thick, fill=blue, fill opacity=0.1] (-0.8,0) rectangle (0.8,0.65);
            \node (R)    at (0,0.35) {$\vty{R}$};

            \draw[thick, fill=blue, fill opacity=0.1] (-0.8,0) rectangle (0.8,1.3);
            \node (dy)   at (0,0.95) {$\wwpc(\vty{Dy})$};

            \draw[thick, fill=blue, fill opacity=0.1] (-0.8,0) rectangle (0.8,1.8);
            \node (Romg) at (0,1.55) {$\vty{R}^\omega$};

            \draw[thick, fill=blue, fill opacity=0.1] (-0.8,0) rectangle (0.8,2.3);
            \node (A)    at (0,2.05) {$\vty{A}$};

            \draw[thick, fill=blue, fill opacity=0.1] (-0.8,0) rectangle (0.8,2.9);
            \node (reg)  at (0,2.6) {\textcolor{blue}{$\vty{REG}$}};
        \end{tikzpicture}
        \caption{regular and subregular}
        \label{fig:reg}
    \end{subfigure}%
    \begin{subfigure}[t]{0.32\linewidth}
        \centering
        \begin{tikzpicture}
            \draw[thick, fill=blue, fill opacity=0.1] (-1.9,0) rectangle (1.9,1.75);
            \node (tc0) at (-1.55,1.5) {\textcolor{blue}{$\mathsf{TC}^0$}};

            \draw[thick, fill=blue, fill opacity=0.1] (-1.9,0) rectangle (1.9,0.75);
            \node (ac0) at (-1.55,0.5) {\textcolor{blue}{$\mathsf{AC}^0$}};

            \draw[thick, fill=orange, fill opacity=0.1] (-1.2,0) rectangle (1.2,1.3);
            \node (crasp) at (0,1.0) {\textcolor{orange}{$\CRASP$}};
        \end{tikzpicture}
        \caption{circuits}
        \label{fig:circuit}
    \end{subfigure}%
    \begin{subfigure}[t]{0.32\linewidth}
        \centering
        \begin{tikzpicture}
            \draw[thick, fill=blue, fill opacity=0.1] (-1.6,0) rectangle (1.6,0.33);
            \node (b0) at (0,0.165) {\footnotesize$\mathcal{B}_0$};

            \draw[thick, fill=blue, fill opacity=0.1] (-1.6,0) rectangle (1.6,0.66);
            \node (b1) at (0,0.5) {\footnotesize$\mathcal{B}_1$};

            \draw[thick, fill=blue, fill opacity=0.1] (-1.6,0) rectangle (1.6,1.1);
            \node (vd) at (0,1) {$\vdots$};

            \node (Alab) at (-1.5,1.3) {\footnotesize\textcolor{blue}{$\vty{A}$}};

            \draw[thick, fill=orange, fill opacity=0.1] (-1.2,0) rectangle (1.2,1.3);
            \node (crasp) at (0,1.7) {\footnotesize\textcolor{orange}{$\CRASP$}};
        \end{tikzpicture}
        \caption{dot depth hierarchy}
        \label{fig:dot}
    \end{subfigure}%

    \caption{$\CRASP$ is situated orthogonal to well-known classes of languages. $\CRASP$ contains all $\mathcal{R}$-trivial languages but not all star-free ones  (\cref{fig:reg}). 
    $\CRASP$ is contained in $\mathsf{TC}^0$, but incomparable to $\mathsf{AC}^0$ (\cref{fig:circuit}). 
    $\CRASP$ intersects every level of the dot depth hierarchy, while remaining incomparable to the full hierarchy (\cref{fig:dot}).}
    \label{fig:inclusions}
\end{figure}
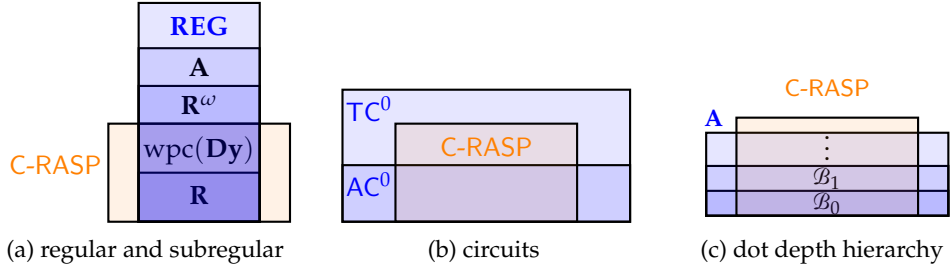

\section{Algebraic Preliminaries}

Our characterization of the regular languages that languages transformers length-generalize on (i.e. those in $\CRASP$) builds upon the algebraic theory of formal languages \citep{Pin_MFAT}.

\subsection{Basics: Algebraic Theory of Formal Languages}

We refer to \cref{app:algebra} for additional definitions.
The essential ones are presented here.

\begin{defin}[Monoid]
    A monoid $(M,\cdot,1)$ is a set $M$ with an associative binary operation and an identity element. We will just write $M$ when the operation and identity are clear.
\end{defin}

A finite monoid $M$ together with a homomorphism $\phi \colon \Sigma^* \to M$ behaves like a finite automaton: after reading string $w = w_1 \cdots w_n$, the automaton is in state $\phi(w_1) \cdots \phi(w_n)$.
Like a finite automaton, a monoid $M$ \emph{recognizes} a language $L$ if there exists an accepting subset $X\subseteq M$ and a homomorphism $\phi\colon \Sigma^*\to M$ such that $L=\phi^{-1}(X)$.
Two monoids will serve as basic units for us (the latter is the ``flip-flop'' alluded to above).

\begin{defin}
    The monoid $\uone$ is the set $\{0,1\}$ where $0\cdot 1=1\cdot 0=0\cdot 0=0$ and $1\cdot1=1$. The monoid $\uthree$ is the set $\{1,a,b\}$ where $x\cdot a=a$, $x\cdot b=b$, and $1\cdot x=x=x\cdot 1$ for any $x\in \uthree$. 
\end{defin}

As an example, consider $\lang=\Sigma^*a\Sigma^*$ and the homomorphism $\phi\colon\Sigma^*\to \uone$ where $\phi(\sigma)=0$ iff $\sigma=a$.
Then for any $w\in\Sigma^*$ we have that $\phi(w)=0\in\uone$ iff in $w\in\lang$.
In some sense, $\uone$ captures the structure of $\lang$ (i.e. it detects if $a$ ever occurs in the string). 
In fact, any monoid recognizing $\lang$ is divided by $\uone$, in the following sense:

\begin{defin}[Division]
    A monoid $M$ divides a monoid $N$ iff there is a submonoid $T$ of $N$ and homomorphism $\phi\colon T\to M$ such that $M=\phi(T)$. Division is transitive.
\end{defin}
For a language $L$, the \emph{syntactic monoid} $M(L)$ is the unique monoid which recognizes $L$ and divides all other $M$ that also recognize $L$.
In the example above, $\uone$ is the syntactic monoid of $\lang = \Sigma^*a\Sigma^*$.

We will consider classes of monoids that enjoy closure properties that will be useful for effective characterizations.

\begin{defin}[Pseudovariety]
    A \emph{pseudovariety} of monoids is a collection of monoids closed under division and finite direct products. 
\end{defin}

Some pseudovarieties relevant to our characterization are $\vty{R}$ (the $\mathcal{R}$-trivial monoids \citep{BRZOZOWSKI198032}, $\vty{A}$ (the aperiodic monoids \citep{1571135650631311872}), $\vty{REG}$ (all regular languages), and $\vty{Dy}$ (the pseudovariety generated by all bounded Dyck monoids).  

\subsection{Background: Classical Algebraic Decomposition Theory of Regular Languages}

The classical algebraic operation for composing monoids is the \emph{wreath product}. 
Here we will sometimes use $+$ to notate the monoid multiplication to decongest the notation, but we do not intend to suggest it is commutative.

\begin{definition}[Classical Wreath Product]
    The wreath product $M\wrc N$ of finite monoids $(M,+)$ and $(N,\cdot)$ is the monoid $M^N\times N$ with multiplication given by $(f_1,n_1)(f_2,n_2)=(f_1+{}^{n_1}f_2,n_1n_2)$, where the left action ${}^nf$ is given by  ${}^nf(n')=f(n'n)$. 
\end{definition}
The wreath product $M \circ N$ together with a homomorphism $\phi \colon \Sigma^* \to M \circ N$ behaves like the composition of a finite automaton and a finite transducer. Let $\phi_M$ and $\phi_N$ be such that $\phi(a) = (\phi_M(a), \phi_N(a))$ for all $a \in \Sigma$. After reading a prefix $w_1 \cdots w_t$, the automaton is in state $q_t = \phi_N(w_1) \cdots \phi_N(w_t)$, and the transducer is in state $\phi_M(w_1)(q_0) + \phi_M(w_2)(q_1) + \phi_M(w_3)(q_2) + \ldots + \phi_M(wt)(q_{t-1})$.
Wreath products are the backbone of the fundamental result in the decomposition theory of finite monoids, the Krohn-Rhodes Theorem: 

\begin{thm}[Krohn-Rhodes Theorem \citep{krohn1965algebraic}]
    Every finite monoid $M$ divides an iterated wreath product of $\uthree$ and simple groups $G$ that divide $M$.
\end{thm}

Unfortunately, Krohn-Rhodes theory is insufficient for handling the monoids involved in transformer length-generalization.
$\CRASP$ defines languages with infinite syntactic monoids (Krohn-Rhodes only applies to finite monoids), while $\uthree$ is not definable in $\CRASP$ \citep{huang2025formalframeworkunderstandinglength} (the Krohn-Rhodes flip-flop unit is useless for $\CRASP$).
The essential monoid for us will be the syntactic monoid of the bounded Dyck monoid, which can be defined in $\CRASP$. 
\begin{defin}[Bounded depth Dyck language]
    Define $\dyck_1:=(ab)^*$ and $\dyck_{k+1}:=(a\dyck_kb)^*$.
\end{defin}
In short: transformers simultaneously succeed at length-generalizing on languages beyond the scope of Krohn-Rhodes theory \citep{bhattamishra-etal-2020-ability}, and fail to length generalize on the fundamental units in scope of Krohn-Rhodes theory \citep{liu2023exposing}.
This necessitates a new decomposition theory to handle $\CRASP$.

\section{Algebraic Characterization of \(\CRASP\)}
\citet{huang2025formalframeworkunderstandinglength} showed that transformer length-generalization can be guaranteed for all languages in $\CRASP$, and strong empirical evidence suggests failure of length-generalization outside of $\CRASP$ \citep{jobanputra2025born}.
Furthermore, a form of fixed-precision transformer is equivalent to $\CRASP$ \citep{yang2025knee}.
In this section we develop an algebraic characterization of $\CRASP$, which crucially requires moving beyond Krohn-Rhodes theory to infinite monoids.
We refer to \citet{yang2024counting,huang2025formalframeworkunderstandinglength} for an exposition of $\CRASP$ and provide a formal definition in \cref{app:crasp}.

\subsection{Typed Monoids}
\citet{krebs2008typed} developed a framework for using infinite monoids to recognize languages.
We present a restriction of the aforementioned framework to the case of wreath products (a one-sided version of the block product used in previous work), which ultimately provides an exact algebraic characterization of $\CRASP$. 
The core issue here is that the wreath product of infinite monoids can generate uncountably many elements, which can be too powerful.
\begin{proposition}
    Consider the classic wreath product $\Z\wrc \Z$. Then $M(\lang)\preceq \Z\wrc \Z$ for every $\lang$.
\end{proposition}
\begin{proof}
    Without loss of generality let $\Sigma=\{0,1\}$. Consider the submonoid of $\Z\wrc\Z$ generated by the image of $\Sigma^*$ under the homomorphism $\sigma \mapsto(f_\sigma,1) $ where $f_\sigma(x)=\sigma\cdot 2^{|x|}$.
    In essence, this creates a mapping $w\mapsto (f_w,|w|)$ where $f_w(0)$ outputs the integer value of the binary number $w$.
    Thus, $\Z\wrc\Z$ can recognize arbitrary languages.
\end{proof}

This problem motivates the definition of \emph{typed monoids}, which restricts the accepting sets to be collection of sets closed under union, intersection and complement. 
\begin{definition}
    A typed monoid is a triple $(M, \types{M}, \units{M})$ where $M$ is a finitely generated monoid, $\types{M}$ is a finite Boolean algebra over $M$, and $\units{M}$ is a finite subset of $M$. Elements of $\types{M}$ are the \emph{types} and elements of $\units{M}$ are the \emph{units}. A language $L$ is recognized by $(M, \types{M}, \units{M})$ if there exists a homomorphism $h\colon \Sigma^*\to M$ such that $h(\Sigma)\subseteq\units{M}$ and $L=h^{-1}(\type{M})$ for some $\type{M}\in \types{M}$.
\end{definition}

As an example, the language $\mathsf{MAJORITY}$ (there are more $a$'s than $b$'s) is recognized by the typed monoid $(\Z,\{(-\infty,0],[1,\infty),\Z,\emptyset\},\{-1,1\})$ via the type $[1,\infty)$ and the homomorphism $a\mapsto 1$ and $b\mapsto -1$.
We will typically refer to this typed monoid as $\Z$.
Now the \emph{typed wreath} product follows the same intuition as the classical wreath product, except the computations are restricted so as not to have more distinguishing power than provided by the types.

\begin{definition}[Typed Wreath Product]
    Let \( (M, \types{M}, \units{M}) \), \( (N, \types{N}, \units{N}) \) be two typed monoids, and let \( C \subseteq N \) be a finite set. The \textit{typed wreath product} 
\[
(U, \types{U}, \units{U}) = (M, \types{M}, \units{M}) \twreath_C (N, \types{N}, \units{N})
\]
of \( (M, \types{M}, \units{M}) \) with \( (N, \types{N}, \units{N}) \) using constants $C$ is defined such that
\begin{itemize}[noitemsep]
    \item \( \units{U} \) consists of elements \( (f, n) \), where \( n \in \units{N} \), and \( f : N \rightarrow \units{M}\) is a type respecting function (see \cref{def:type_respecting}) with respect to \( (N, \types{N}, \units{N}) \) and \( C \)
    \item \( U \) is the submonoid of \( M \wrc N \) generated by \( \units{U} \)
    \item \( \types{U} \) consists of types \( \type{U}_{\type{M},\type{N}} = \{(f, n) \mid f(1_{N}) \in \type{M}, n \in \type{N} \} \), where \( \type{M} \in \types{M} \), \( \type{N} \in \types{N} \)
\end{itemize}
Multiplication is the same as in the classical wreath product.
\end{definition}

\subsection{Wreath Product Characterization }

With the notion of typed monoids in hand, we can give an algebraic characterization of $\CRASP$. 
Let $\wwpc(M,\types{M},\units{M})$ denote the \emph{wreath product closure} of a typed monoid, which closes iterated wreath products of this monoid under Boolean combinations and other basic operations.
The precise definition will be found in \cref{app:typed_monoids}.
The typed wreath product closure turns out to be the precise algebraic ``glue'' that connects integer counting to $\CRASP$ programs.
The proof is given in \cref{app:CRASP_Z}.
\begin{restatable}{thm}{CRASPZ}\label{thm:crasp_wreath}
   $L\in\CRASP \iff M(L)\in \wwpc(\Z)$
\end{restatable}

\section{Decomposition Theory for Regular Languages in $\CRASP$}

So far, we have characterized $\CRASP$ in terms of iterated wreath products of $\Z$.
We will turn this into an algebraic decision algorithm: given a regular language $\lang$, we determine if $M(\lang)$ divides a wreath product of $\Z$ (or not).
A priori, such a question is a formidable problem, due to multiple challenges: $\Z$ is infinite, and we do not even know how many $\Z$ factors are needed.
In fact, there even examples in the finite case where the problem ends up undecidable
\citep{rhodes1999undecidability}.
Our second main result will be that, using a detailed understanding of wreath products of $\Z$, this problem is decidable.

In the depth-$1$ case, it is easy to check if $M(\lang)$ divides $\Z$ (this is true e.g. for the AND language, $L = 1^*$).
How would we check if $M(\lang) \prec \Z \wrc\Z$?
If we could ``divide'' out the right factor to obtain an object ``$M(\lang) / \Z$'', we could check if $``M(\lang) / \Z'' \prec \Z$.
If this is not the case, we could ``divide'' by $\Z$ again, and iterate until we either reach division of $\Z$ (i.e., $L$ is in $\CRASP$), or a fixed point (i.e., $\lang$ is not in $\CRASP$).
This is the basic idea of our decision procedure.
There are three challenges at hand:
First, understanding how to ``divide'' by a wreath product factor; second, making this computable even though $\Z$ is infinite; third, understanding how to identify the fixed point.

\subsection{Categories}
What do you get when you ``divide'' one monoid by another?  
For this first challenge, there is a well-understood technique using categories as algebraic structures \citep{tilson1987categories}.
We informally present the main idea, but defer the precise definitions to \cref{app:derived_categories}.

In group theory, this question has a simple answer. 
With a surjective group homomorphism between groups $\phi\colon G\to H$, we can ``divide'' $G$ by $\ker\phi$ with a ``quotient'' of $H$, such that $G\preceq (\ker\phi) \wrc H$.
However, since we are dealing with monoids (and $\CRASP$ does not even contain any non-trivial finite groups), this is of no use for us. 

For monoid homomorphisms $\phi\colon M\to N$, we may not be able to ``divide'' $M$ by $\ker\phi$, as monoids' lack of inverses obstructs such a clean division. 
What is the divisor in the division $\phi\colon M(\dyck_1)\to \uthree$? 
Intuitively, the wreath product $M(\dyck_1)\preceq N\wrc \uthree$ should form a structure that contains elements $(\text{last symbol}=a)$ and $(\text{last symbol}=b)$.
However, composition in this structure will be highly restricted -- e.g. $(\text{this symbol}=b)(\text{last symbol}=a)$ is an illegal composition. 
It turns out the cleanest way to handle this is by lifting monoids and homomorphisms to higher order structures -- categories and relational morphisms
. 

A category $X$ consists of a set of \emph{objects} $\obj(X)$ and \emph{hom-sets} $X(x_1,x_2)$, which are collections of \emph{arrows} $\alpha\colon x_1\to x_2$ for each pair of objects $x_1,x_2\in X$.
Arrows can compose associatively, and each object has an identity arrow. 
A relational morphism of monoids $\phi\colon M\relml N$ is a relation where $\phi(m)\neq \emptyset$, $\phi(m_1)\phi(m_2)\subseteq \phi(m_1m_2)$, and $1_N \in \phi(1_M)$.
This is a generalization of the classical morphism, where elements may map to sets of elements. 

Now for any relational morphism of monoids $\phi\colon M\relml N$,  the correct notion of divisor turns out to be the \emph{derived category} $D_\phi$.
Here, $\obj(D_\phi)=\phi(M)$ and $D_\phi(n_1,n_2)=\{n_1\arrow{(m,n)}\mid n\in \phi(m), n_1n=n_2\}$.
Arrows compose via the rule $n_0\arrow{{(m_1,n_1)}}n_0n_1\arrow{{(m_2,n_2)}}=n_0\arrow{(m_1m_2,n_1n_2)}$, and we merge arrows that have the same behavior under composition.
    
The \emph{Derived Category Theorem} \citep{tilson1987categories} states that given a relational morphism $\phi\colon M\relml N$, and a division $D_\phi\preceq V$, we obtain a division $ M\preceq V\wrc N$.
Thus, if $\phi : M \relml \Z$, then $D_\phi$ formalizes the object ``$M/\Z$'' alluded to in the previous section.
Going forward, we work with an extension of the notions of relational morphisms and division in which
the left-hand side can also be a category.

\subsection{Decomposition of $\dyck_1$ into wreath products of $\Z$}\label{sec:dyck_derived}
\tikzstyle{state}=[rectangle, draw, minimum size=0.8cm, inner sep=0pt]
\begin{minipage}[b]{0.7\linewidth}
We explain our decision procedure by walking through the division of the syntactic monoid of $\dyck_1=(ab)^*$ into an iterated wreath product of $\Z$, via carefully selected relational morphisms into $\Z$. 
We will visualize objects of a category as rectangles (to prevent confusion with automata) and arrows as labels on edges between squares. 
Colors indicate the monoid each element comes from.
We can view $M(\dyck_1)$ as a category with a single object and arrows corresponding to monoid elements.
We take a relational morphism into $\phi_1\colon M(\dyck_1)\relml\red{\Z}$ where $\phi_1(\blu{\bot})= \red{\Z}$, $\phi_1(\blu{a})= \{\red{1}\}$,  $\phi_1(\blu{b})= \{\red{-1}\}$, $\phi_1(\blu{\epsilon})=\phi_1(\blu{ab})=\phi_1(\blu{ba})=\{\red{0}\}$.
\end{minipage}%
\begin{minipage}[b]{0.3\linewidth}
\begin{center}
    \begin{tikzpicture}[->,>=stealth',shorten >=1pt,auto,node distance=1.8cm,semithick,initial text=,initial where=left, node distance=1ex]
        \node[state] (q0) {$\blu{M}$};
        \path (q0) edge [loop above] node[align=center] {$\blu{\epsilon}$\\$\blu{a}$\\$\blu{b}$\\$\blu{ab}$\\$\blu{ba}$\\$\blu{\bot}$} (q0);
    \end{tikzpicture}
\end{center}
\end{minipage}

In the derived category $D_{\phi_1}$ there are infinitely many objects, corresponding to each integer.
We omit starting or ending objects of arrows where clear by the diagram, combine isomorphic arrows, and let low-opacity arrows and objects denote transitions to $\blu{\bot}$ elements. 
\begin{center}
    \begin{tikzpicture}[->,>=stealth',auto,semithick, node distance=12ex]
        \node[state] (q0) {$\red{0}$};
        \path (q0) edge [loop left] node[align=center, above left = 0.15 and -.3] {$\arrow{(\blu{\epsilon},\red{0})}$\\$\arrow{(\blu{ab},\red{0})}$} (q0);

        \node[state,right= of q0] (q1) {$\red{1}$};
        \path (q1) edge [loop right] node[align=center, above right = 0.15 and -.3] {$\arrow{(\blu{\epsilon},\red{0})}$\\$\arrow{(\blu{ba},\red{0})}$} (q1);

        \node[state,left= of q0, opacity=0.36] (qn1) {$\red{-1}$};

        \node[state,right= of q1, opacity=0.36] (q2) {$\red{2}$};

        \node[state,left= of qn1, opacity=0.36] (qn2) {$\red{-2}$};

        \node[left=.2 of qn2] {\large$\cdots$};
        \node[right=.2 of q2] {\large$\cdots$};

        \path (q1) edge [bend right=6] node[above, align=center] {$\red{0}\arrow{(\blu{b},\red{-1})}$} (q0);
        \path (q0) edge [bend right=6] node[below, align=center] {$\red{1}\arrow{(\blu{a},\red{1})}$} (q1);

        \path (q0.south) edge [<->, dashed, opacity=0.36, bend right=15, color=black] (q2.south);
        \path (qn1.south) edge [<->, dashed, opacity=0.36, bend right=15, color=black] (q1.south);
        \path (qn1.south) edge [<->, dashed, opacity=0.36, bend right=15, color=black] (q2.south);
        \path (qn2.south) edge [<->, dashed, opacity=0.36, bend right=15, color=black] (q0.south);
        \path (qn2.south) edge [<->, dashed, opacity=0.36, bend right=15, color=black] (q1.south);
        \path (qn2.south) edge [<->, dashed, opacity=0.36, bend right=15, color=black] (q2.south);

        \path (q1.south) edge [<->, dashed, opacity=0.36, bend right=15, color=black] (q2.south);
        \path (q0.south) edge [<->, dashed, opacity=0.36, bend right=15, color=black] (q1.south);
        \path (qn1.south) edge [<->, dashed, opacity=0.36, bend right=15, color=black] (q0.south);
        \path (qn2.south) edge [<->, dashed, opacity=0.36, bend right=15, color=black] (qn1.south);

        \path (qn2) edge [dashed, opacity=0.36, loop above] (qn2);
        \path (qn1) edge [dashed, opacity=0.36, loop above] (qn1);
        \path (q0) edge [dashed, opacity=0.36, loop above] (q0);
        \path (q1) edge [dashed, opacity=0.36, loop above] (q1);
        \path (q2) edge [dashed, opacity=0.36, loop above] (q2);

    \end{tikzpicture}
\end{center}
We have not yet arrived at a division $M(\dyck_1)\preceq \Z$ because $D_{\phi_1}$ has hom-sets containing multiple arrows \citep[Lemma~3.1, Lemma~4.1]{tilson1987categories}.
To proceed, we take another relational morphism $\psi_1\colon D_{\phi_1}\relml \yel{\Z}$ where 
$\psi_1(\red{x}\arrow{(\blu{\bot},\red{1})})=\psi_1(\red{x}\arrow{(\blu{\bot},\red{-1})})=[\yel{1},\yel{\infty})$ for all $\red{x}$; arrows not associated to $\blu{\bot}$ are instead mapped to $\{\yel{0}\}$.
From $\phi_1$, $\psi_1$ we define a relational morphism $\phi_2\colon M(\dyck_1)\relml \yel{\Z}\wrc\red{\Z}$ where we let $\phi_2(\blu{m})=\left\{\left(\yel{f}_{(\blu{m},\red{y})},\red{y}\right)\mid\red{y}\in\phi_1(\blu{m}), \forall \red{x} : \yel{f}_{(\blu{m},\red{y})}(\red{x})=\psi_1\left(\red{x} \arrow{(\blu{m},\red{y})} \right)\right\}$.
We visualize the derived category $D_{\phi_2}$ below (writing $(\yel{f_{\blu{m}}},\red{y})$ as shorthand for $\phi_2(\blu{m})$ whenever it is unique).
\begin{center}
    \begin{tikzpicture}[square/.style={
    minimum size=0.8cm},->,>=stealth',shorten >=1pt,auto,semithick,initial text=,initial where=left, node distance=12ex]
        \node[state, square] (qab) {$(\yel{f}_{\blu{ab}},\red{0})$};
        \node[state, square, right=of qab] (qa) {$(\yel{f}_{\blu{a}},\red{1})$};
        \node[state, square, right=of qa] (qb) {$(\yel{f}_{\blu{b}},\red{-1})$};
        \node[state, square, right=of qb] (qba) {$(\yel{f}_{\blu{ba}},\red{0})$};

        \node[state, square, below right=.5 and 0.5 of qa, opacity=0.36] (qbot0) {$(\yel{f}_{\blu{\bot}},\red{0})$};
        \node[state, square, right=1.4 of qbot0, opacity=0.36] (qbot1) {$(\yel{f}_{\blu{\bot}},\red{1})$};
        \node[state, square, right=1.4 of qbot1, opacity=0.36] (qbot2) {$(\yel{f}_{\blu{\bot}},\red{2})$};
        \node[state, square, left=1.4 of qbot0, opacity=0.36] (qbotn1) {$(\yel{f}_{\blu{\bot}},\red{-1})$};
        \node[state, square, left=1.4 of qbotn1, opacity=0.36] (qbotn2) {$(\yel{f}_{\blu{\bot}},\red{-2})$};
        \node[left=.2 of qbotn2] (qndots) {\large$\cdots$};
        \node[right=.2 of qbot2] (qdots) {\large$\cdots$};

        \node[state, square, above right=.5 and 0.5 of qa] (qe) {$(\yel{f}_{\blu{\epsilon}},\red{0})$};
        \path (qe) edge [loop above] node[align=center] {$\arrow{(\blu{\epsilon},(\yel{f}_{\blu{\epsilon}},\red{0}))}$} (qe);

        \path (qe) edge [bend right=12] node[align=center, above left] {$\arrow{(\blu{a},(\yel{f}_{\blu{a}},\red{1}))}$} (qa);
        \path (qe) edge [bend right=31] node[align=center, above=.2] {$\arrow{(\blu{ab},(\yel{f}_{\blu{ab}},\red{0}))}$} (qab);
        \path (qe) edge [bend left=31] node[align=center, above=.2] {$\arrow{(\blu{ba},(\yel{f}_{\blu{ba}},\red{0}))}$} (qba);
        \path (qe) edge [bend left=12] node[align=center, above right] {$\arrow{(\blu{b},(\yel{f}_{\blu{b}},\red{-1}))}$} (qb);
        
        \path (qa) edge [bend right=10] node[align=center, above] {$\arrow{(\blu{b},(\yel{f}_{\blu{b}},\red{-1}))}$} (qab);
        \path (qab) edge [bend right=10] node[align=center, below] {$\arrow{(\blu{a},(\yel{f}_{\blu{a}},\red{1}))}$} (qa);

        \path (qb) edge [bend right=10] node[align=center, below] {$\arrow{(\blu{a},(\yel{f}_{\blu{a}},\red{1}))}$} (qba);
        \path (qba) edge [bend right=10] node[align=center, above] {$\arrow{(\blu{b},(\yel{f}_{\blu{b}},\red{-1}))}$} (qb);

        \path (qab) edge [loop above] node[align=center] {$\arrow{(\blu{\epsilon},(\yel{f}_{\blu{\epsilon}},\red{0}))}$\\$\arrow{(\blu{ab},(\yel{f}_{\blu{ab}},\red{0}))}$} (qab);

        \path (qba) edge [loop above] node[align=center] {$\arrow{(\blu{\epsilon},(\yel{f}_{\blu{\epsilon}},\red{0}))}$\\$\arrow{(\blu{ba},(\yel{f}_{\blu{ba}},\red{0}))}$} (qba);

        \foreach \src in {qe, qa, qab, qba, qb} {
            \foreach \dst in {qbot0, qbot1, qbot2, qbotn1, qbotn2} {
                \path (\src) edge [dashed, opacity=0.36] node[align=center] {} (\dst);
            }
        }
        \foreach \src in {qbot0, qbot1, qbot2, qbotn1, qbotn2} {
                \path (\src) edge [dashed, loop below, opacity=0.36] node[align=center] {} (\src);
        }

        \path (qbotn2) edge [<->, dashed, opacity=0.36, bend right=15, color=black] (qbotn1);
        \path (qbotn2) edge [<->, dashed, opacity=0.36, bend right=15, color=black] (qbot0);
        \path (qbotn2) edge [<->, dashed, opacity=0.36, bend right=15, color=black] (qbot1);
        \path (qbotn2) edge [<->, dashed, opacity=0.36, bend right=15, color=black] (qbot2);
        
        \path (qbotn1) edge [<->, dashed, opacity=0.36, bend right=15, color=black] (qbot0);
        \path (qbotn1) edge [<->, dashed, opacity=0.36, bend right=15, color=black] (qbot1);
        \path (qbotn1) edge [<->, dashed, opacity=0.36, bend right=15, color=black] (qbot2);
        
        \path (qbot0) edge [<->, dashed, opacity=0.36, bend right=15, color=black] (qbot1);
        \path (qbot0) edge [<->, dashed, opacity=0.36, bend right=15, color=black] (qbot2);
        
        \path (qbot1) edge [<->, dashed, opacity=0.36, bend right=15, color=black] (qbot2);
    \end{tikzpicture}
\end{center}

We still do not have a division because there are still two arrows in a single homset of $D_{\phi_2}$, for instance $\phi_2(\blu{\epsilon})$ and $\phi_2(\blu{ab})$ act the same on $(\yel{f}_{\blu{ab}},\red{0})$.
To handle this, we define a final relational morphism $\phi_3\colon M(\dyck_1)\relml \Z\wrc \yel{\Z}\wrc\red{\Z}$ such that $\phi_3(\blu{m})= (g,\phi_2(\blu{m}))$ where $g(x)=0$ iff $x=(\yel{f}_{(\blu{\epsilon},\red{0})},\red{0})$.
We ultimately obtain the division. $\phi_3\colon M(\dyck_1)\preceq \Z\wrc\yel{\Z}\wrc\red{\Z}$.
In this example, we chose the correct relational morphisms into $\Z$ so as to obtain a division, but in principle there are infinitely many choices of morphisms.
Guiding the choice of relational morphisms to reach a termination is a major challenge we address in our decision procedure. 
\tikzstyle{state}=[circle, draw, minimum size=1cm]

\subsection{Algebraic Decision Procedure}

We now address the second challenge of maintaining computability despite the existence of infinitely many relational morphisms into $\Z$.
First, we divide $M(\lang)$ into equivalence classes via the $\mathcal{R}$ relation, a classical relation in semigroup theory which characterizes which elements are reachable by other elements via right-multiplication \citep{Pin_MFAT}.
We construct a sequence of iterated relational morphisms into $\Z$ by iterating over the $\mathcal{R}$-classes in order. 
At step $i+1$, we assume that a relational morphism $\phi_i$ covers all $\mathcal{R}$-classes up to $R_i$, and extend to a relational morphism $\phi_{i+1}$ that covers $R_{i+1}$.
This is done by iteratively computing  nontrivial relational morphisms whose values are bounded on $R_i$. 
This process terminates, because the set of such bounded relational morphisms forms a finitely generated $\Z$-module (an integer analogue of vector spaces), and every step finds a morphism linearly independent of the previous ones.

Finally we address the termination conditions by proving completeness of the algorithm.
If the algorithm terminates with success, we have constructed a division from $M(\lang)$ into a wreath product of $\Z$ (and thus $\CRASP$ by \cref{thm:crasp_wreath}).
Otherwise, if the algorithm terminates with failure, we show no such division exists. 
In this case, assume for sake of contradiction that $M(\lang)$ divides a $T$-fold iterated wreath product of $\Z$ via a relational morphism $\psi\colon M(\lang)\preceq \Z \wrc\cdots \wrc\Z$ (where $T$ is minimal).
The rightmost component defines a relational morphism $\omega : M(L) \relml \Z$.
Either $\omega$ is bounded, and it would have been chosen as part of the decomposition had it provided any useful information, or $\omega$ is unbounded which renders long strings indistinguishable by types in $\Z$ since values may run to infinity. In either case, we can remove the rightmost component of $\psi$, and obtain a division from $M$ into a $(T-1)$-fold iterated wreath product of $\Z$, contradicting the minimality of $T$. 

To formalize the above proof, we choose relational morphisms to $\Z$ which truncate values beyond a threshold $k$ (which results in relational morphisms to $\dyck_k$). This allows us to formalize the proof using finite category theory (avoiding the use of types on the category side).
A formal writeup can be found in \cref{app:algebraic_decision_procedure}.

\subsection{Necessary (but not Sufficient) Condition via Equations}

In this section use profinite equations to give an alternative characterization of the regular languages in $\CRASP$. 
These equations are a technique drawing ideas from topology to derive elegant and decidable characterizations for classes of monoids.
We defer a precise definition to \cite{pin:LIPIcs.STACS.2009.1856}. 
\begin{definition}
Define $\vty{R}^\omega$ as an equation and the corresponding variety of monoids\footnote{$\vty{R}^\omega$ alludes to the equation for $\mathcal{R}$-trivial monoids, $(xy)^\omega=(xy)^\omega x$ \citep{ALMEIDA1989129}} as $$\vty{R}^\omega:   (xy^\omega)^\omega x = (xy^\omega)^\omega.$$
\end{definition}
For our purposes, a finite monoid $M$ satisfies $\vty{R}^\omega$ if for any $m_1,m_2\in M$ we have that $(m_1m_2^\omega)^\omega=(m_1m_2^\omega)^\omega m_1$, where $m^\omega$ denotes $m^k$ such that $m^k=m^{2k}$ (i.e. the unique idempotent generated by $m$). 
We give an exact characterization of the monoids in $\vty{R}^\omega$, the proof of which is in \cref{app:R_infty}. %
\begin{restatable}{thm}{RInfinity}\label{thm:R_infty_idempotents}
     $M\in \vty{R}^\omega$ iff $M$ is aperiodic and every $\mathcal{R}$-class of $M$ contains at most one idempotent. 
\end{restatable}
Equivalently, $\vty{R}^\omega=\vty{R}\wrc\vty{G}\cap\vty{A}$, where $\vty{R}$ denotes the monoids where each $\mathcal{R}$-class is singleton, $\vty{G}$ denotes the finite groups, and $\vty{A}$ denotes the aperiodic monoids.
In our experiments below, we will draw languages from the larger class $\vty{R}\wrc\vty{G}$ rather than $\vty{R}^\omega$ for testing length-generalization. 
While both the main decision procedure and $\vty{R}^\omega$ require iterating over the $\mathcal{R}$-classes of the monoid, the former needs to iteratively compute relational morphisms, while 
the latter only needs to count idempotents. This gives a much simpler criterion that is \emph{necessary} for membership in $\CRASP$, though we can show it is \emph{not sufficient}.

\subsection{Algebraic Characterization}

The decision procedure results an an exact algebraic characterization of the regular languages in $\CRASP$ -- namely, they are wreath products of bounded-depth Dyck languages. 
We thus derive a small hierarchy within the classes of finite monoids surrounding $\CRASP$.

\begin{restatable}{thm}{DyckVSRInfty}\label{thm:dyck_vs_R_infty}
    $\CRASP\cap\vty{REG}=\wwpc(\vty{Dy})$. Hence $\vty{R}\subsetneq \CRASP\cap \vty{REG}\subsetneq \vty{R}^\omega\subsetneq \vty{A} \subsetneq \vty{REG}$
\end{restatable}
A proof is given in \cref{app:R_infty}. Furthermore, the decision procedure runs in polynomial time in the size of the monoid.  
This allows us to efficiently decide whether or not we expect a transformer to length-generalize on any given regular language. A proof is in \cref{app:algebraic_decision_procedure}.
\begin{restatable}{thm}{CraspRegPolyTime}\label{thm:algebraic_polytime}
Membership of $M\in \CRASP\cap\vty{REG}$ is decidable in $O(\mathsf{poly}(|M|))$ time.
\end{restatable}

\section{Experiments}
\label{main:experiments}

\begin{figure}[t]
\centering

\includegraphics[width=\linewidth]{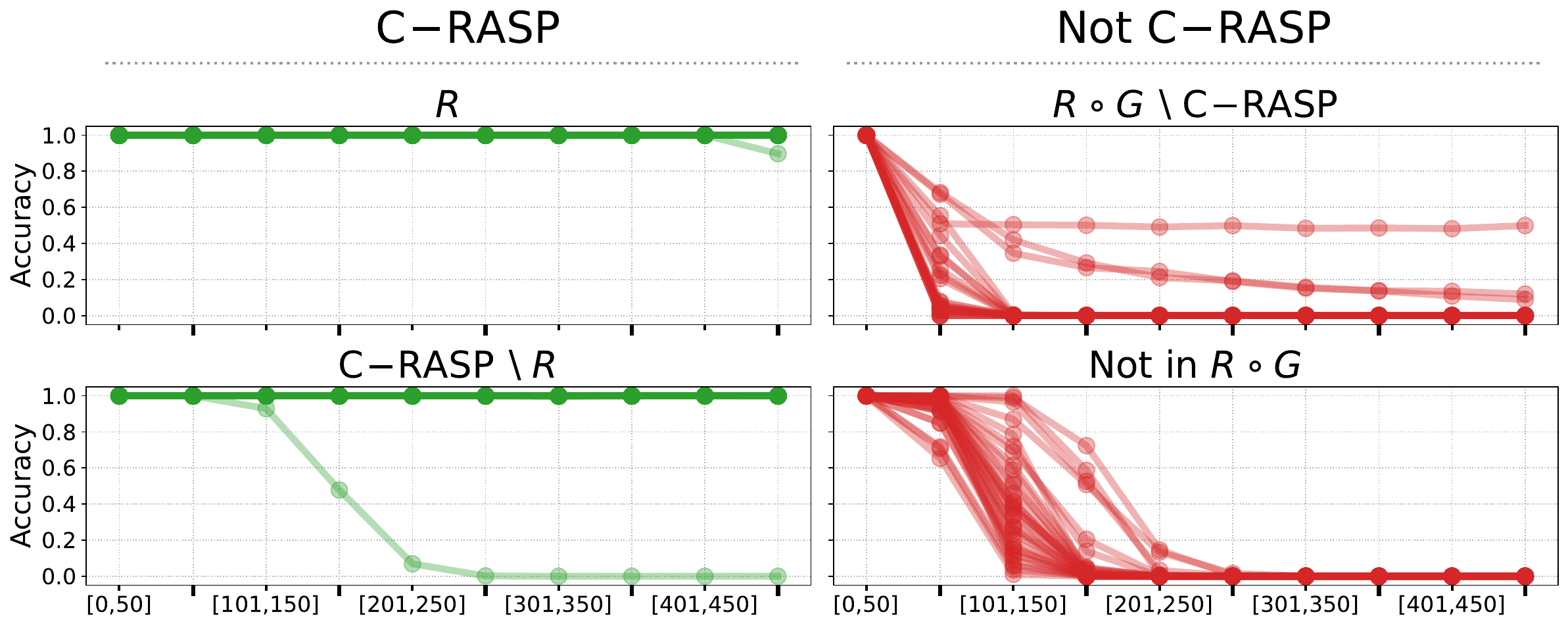}

\caption{\textbf{Length generalization on regular languages.} Models are trained on strings of lengths in $[l_{\min},50]$, and evaluated on bins $[l_{\min},50]$ (in-distribution) up to $[451,500]$ in bins of width 50. Each curve corresponds to a language. Green curves denote languages in $\CRASP$, while red curves denote languages not in $\CRASP$. Languages in $\CRASP$ maintain near-perfect accuracy well beyond the training range, whereas languages outside $\CRASP$ exhibit rapid degradation, typically failing shortly after. Each panel includes all languages in the corresponding class from Table~\ref{tab:languages} in appendix \ref{app:languages}.
}
\label{fig:length_gen}
\end{figure}

We empirically evaluate whether our characterization of regular languages in $\CRASP$ predicts transformer length-generalization.
In particular, we examine the different levels of the hierarchy shown in \cref{thm:dyck_vs_R_infty} to determine the efficacy of each class in predicting length-generalization by transformers. 
Because $\vty{R}^\omega\setminus \CRASP\cap\vty{REG}$ contains few samples, we report results based on membership in $\vty{R}\wrc\vty{G}$ (whose aperiodic fragment is $\vty{R}^\omega$).

\subsection{Languages}

We construct a diverse suite of 125 regular languages, including both representative examples from prior work \citep{li2025characterizing,huang2025formalframeworkunderstandinglength} and systematically generated ones. We provide the full table of languages in Table \ref{tab:languages}, appendix \ref{app:languages}).
To generate languages within the classes discussed in \cref{thm:dyck_vs_R_infty}, we sample regular expressions using a probabilistic context-free grammar (PCFG) for which we tune the probabilities to obtain diverse class membership. 
For each language, we determined membership in \CRASP using an automata-based version of the algebraic procedure, which we explain in \cref{app:automata-based-proof}.

\subsection{Experimental Setup}

\paragraph{Task Definition.} The task is state prediction, i.e. tracking automaton states over prefixes. Let $w = a_1 \dots a_n$, and denote by $w_{1:i} = a_1 \dots a_i$ the prefix of length $i$.
Each symbol $a_i$ triggers a transition, and the model must predict the sequence of DFA states $q_1, \dots, q_n$, where $q_i$ is the state reached after processing $w_{1:i}$.

\paragraph{Training and Evaluation.} Per formal language, we train GPT-2 models on 10,000 words sampled from the training length range $[l_{min},50]$ with $l_{min}$ being the length of the shortest valid word in the respective language. We use an 80/20 train-test split. We then evaluate generalization on test length ranges $[51,100],[101,150]\dots[451,500]$, with 1,000 words in each test set. We define successful length generalization as maintaining near in-distribution performance at lengths beyond $2\times$ the maximum training length $50$. Following \cite{huang2025formalframeworkunderstandinglength} we used AdamW with weight decay 0.01 and dropout 0.0. We performed a hyperparameter sweep over layers $\{1,2,4\}$, heads ${\{1,2,4\}}$, dimension ${\{16,64,256\}}$ and learning rates ${\{0.001,0.0001\}}$ -- training every combination in the grid with early stopping when 100\% accuracy is achieved on in-distribution test data. For those languages, where in-distribution accuracy never reached 100\%, we additionally performed a hyperparameter sweep over layers ${\{6,8,12\}}$, heads ${\{4,8\}}$, dimensions ${\{64,256\}}$ and learning rates $\{0.001, 0.0001\}$. We adopted optimization choices and hyperparameter ranges from \cite{huang2025formalframeworkunderstandinglength}. Per language, we choose the configuration that achieves the highest accuracy on the longest test length range among those whose accuracy on in-distribution length $[l_{min},50]$ is 100\%. We break remaining ties by successively considering the next-longest ranges and finally preferring the smallest model (fewest layers, then attention heads, then hidden dimension). The chosen configuration was then used for a multi-seed run, in which models were trained under different random model initializations, an approach commonly used in prior work \citep{li2025characterizing, huang2025formalframeworkunderstandinglength}. A run was considered successful if it achieved 100\% accuracy on the in-distribution test set. We evaluated random initializations sequentially, up to a maximum of 1,000 trials, and ended the search once five successful runs were found. We report the best successful seed in Figure \ref{fig:length_gen} and show the average performance across all successful seeds in \ref{app:experiments}.

\paragraph{Architectural Constraints.} To align empirical results with our theoretical setup, state prediction should depend on the evolving prefix rather than local information. 
We therefore introduce two constraints: first, we remove positional information by replacing positional embeddings with the zero function (NoPE) forcing the model to rely on the sequential order of tokens provided by the causal attention mask. 
Second, while NoPE removes absolute positional information, the model still has direct access to the most recent input symbol $a_i$ when predicting the target state $q_i$, which provides a positional signal through alignment with the prediction target. 
This can enable shortcuts in which the model predicts the state based solely on the final symbol, rather than the full prefix.
To prevent this, we insert a separator token $\&$ between symbols (see below) and require the model to predict the state $q_i$ only at these separator positions. Consequently, access to the most recent symbol is also mediated through the attention mechanism. Word symbols receive placeholder targets $\#$ and are excluded from the cross-entropy loss. Each target sequence begins with the initial DFA state, predicted at the first separator following \texttt{<bos>}. 

\newcommand{\gradcell}[1]{\cellcolor{black!#1}}
\newcommand{\diagramtext}[1]{\texttt{#1}}
\begin{center}
\footnotesize
\setlength{\tabcolsep}{2pt}
\begin{tabular}{rcccccccccccccccc}
    \hline
    \rule{0pt}{2.5ex}\texttt{Input} 
    & \gradcell{2}{\diagramtext{$\bos$}} & \gradcell{4}{\diagramtext{\&}} & \gradcell{6}{\diagramtext{a}} & \gradcell{8}{\diagramtext{\&}} 
    & \gradcell{10}{\diagramtext{b}} & \gradcell{12}{\diagramtext{\&}} & \gradcell{14}{\diagramtext{a}} & \gradcell{16}{\diagramtext{\&}} 
    & \gradcell{18}{\diagramtext{b}} & \gradcell{20}{\diagramtext{\&}} & \gradcell{22}{\diagramtext{a}} & \gradcell{24}{\diagramtext{\&}} 
    & \gradcell{26}{\diagramtext{b}} & \gradcell{28}{\diagramtext{\&}} & \gradcell{30}{\diagramtext{$\eos$}} \\

    \rule{0pt}{2.5ex}\texttt{Positional Encoding} 
    & \gradcell{2}{\diagramtext{0}} & \gradcell{4}{\diagramtext{0}} & \gradcell{6}{\diagramtext{0}} & \gradcell{8}{\diagramtext{0}} 
    & \gradcell{10}{\diagramtext{0}} & \gradcell{12}{\diagramtext{0}} & \gradcell{14}{\diagramtext{0}} & \gradcell{16}{\diagramtext{0}} 
    & \gradcell{18}{\diagramtext{0}} & \gradcell{20}{\diagramtext{0}} & \gradcell{22}{\diagramtext{0}} & \gradcell{24}{\diagramtext{0}} 
    & \gradcell{26}{\diagramtext{0}} & \gradcell{28}{\diagramtext{0}} & \gradcell{30}{\diagramtext{0}} \\

    \rule{0pt}{2.5ex}\texttt{Target} 
    & \gradcell{0}{\diagramtext{\#}} & \gradcell{4}{\diagramtext{1}} & \gradcell{0}{\diagramtext{\#}} & \gradcell{8}{\diagramtext{3}} 
    & \gradcell{0}{\diagramtext{\#}} & \gradcell{12}{\diagramtext{1}} & \gradcell{0}{\diagramtext{\#}} & \gradcell{16}{\diagramtext{3}} 
    & \gradcell{0}{\diagramtext{\#}} & \gradcell{20}{\diagramtext{1}} & \gradcell{0}{\diagramtext{\#}} & \gradcell{24}{\diagramtext{3}} 
    & \gradcell{0}{\diagramtext{\#}} & \gradcell{28}{\diagramtext{1}} & \gradcell{0}{\diagramtext{\#}} \\
    \hline
\end{tabular}
\normalsize
\end{center}

\subsection{Results}
Figure \ref{fig:length_gen} shows that $\CRASP$ membership is a strong predictor of transformer length-generalization. Languages in $\CRASP$ generalize reliably to substantially longer lengths, while languages outside $\CRASP$ fail to do so, with accuracy rapidly collapsing beyond the training length. This demonstrates that across all evaluated languages $\CRASP$ provides an accurate characterization of length generalization.

\section{Conclusion}

We have precisely characterized the regular languages on which transformers length-generalize using a decision procedure for regular language membership in $\CRASP$ (which runs in polynomial-time).
Experiments on a range of languages in and outside of $\CRASP$ demonstrate that the theory accurately predicts when transformers do and do not length-generalize.
The decision procedure is based on a novel algebraic decomposition theory for $\CRASP$ that is distinguished from the classical Krohn-Rhodes theory due to the inclusion of infinite monoids and the omission of the aperiodic flip-flop unit $\uthree$. 
A simpler criterion which is necessary (but not sufficient) for membership in $\CRASP$ is also given in the form of a profinite equation. 
Our results provide a deeper understanding of the state-tracking capabilities of transformers on the machine learning front, as well as a deeper understanding of counting, on the algebraic front.

\newpage

\section*{Acknowledgments}

We thank Dana Angluin, Michael Benedikt, David Chiang, and Will Merrill for fruitful discussion and feedback.
We thank the anonymous reviewers for their helpful comments.

Funded in part by the Deutsche Forschungsgemeinschaft (DFG, German Research Foundation) -- GRK 2853/1 “Neuroexplicit Models of Language, Vision, and Action” - project number 471607914, and the US 
National Science Foundation (grant number~2502292).
AY is supported by the US National Science Foundation Graduate Research Fellowship Program under Grant No.~2236418.
MH acknowledges support from the Deutsche Forschungsgemeinschaft (DFG, German Research Foundation) – Project number 560456343.

\bibliography{colm2026_conference}
\bibliographystyle{colm2026_conference}

\appendix
\tableofcontents
\section*{Author Contributions}
AY led paper writing, drafted the proof of Theorem \ref{thm:crasp_wreath} and the implementation of the decision procedure, contributed to the discovery and proof of Theorems \ref{thm:dyck_vs_R_infty} and \ref{thm:algebraic_polytime}, drafted Appendix \ref{app:automata-based-proof}, and drafted Sections 1–4 of the paper.
BV designed, implemented, and carried out the experiments, and drafted Section 5 of the paper.
CB, MC, AK, CP, HS contributed to the discovery and proof of Theorems \ref{thm:dyck_vs_R_infty} and \ref{thm:algebraic_polytime}, and provided input to the paper writing.
HS further contributed Theorem \ref{thm:R_infty_idempotents}.
MH contributed to the discovery and proof of Theorems \ref{thm:dyck_vs_R_infty} and \ref{thm:algebraic_polytime}, drafted the proof of Theorems \ref{thm:dyck_vs_R_infty} and \ref{thm:algebraic_polytime} in Appendix \ref{app:decision} and an early draft of Appendix \ref{app:automata-based-proof}, and contributed to paper writing.

\section{FAQ}

\begin{enumerate}
\item \emph{Q: How does the work relate and compare to \cite{liu2023transformers}?}

\citet{liu2023transformers} primarily looked at expressivity, not length generalization.
Indeed, their experiments also confirmed that transformers may not length-generalize on the particular regular languages that they were able to express and learn to a fixed length.

Our results may also apply to the expressive power of transformers under a particular fixed-precision assumption, because $\CRASP$ was shown to be equivalent to these transformers by \citet{yang2025knee}.

\item \emph{Q: How about other architectures? Like log-depth transformers, state-space models, ...?}

It is already known that log-depth architectures can simulate arbitrary automata, so a characterization of the regular languages they can express is not needed \citep{liu2023transformers, merrill2025a}.
As for length-generalization, we do not presently know what regular languages these architectures length-generalize on.

Analyzing architectures with limited recurrence, like state-space models, would be an interesting future question.
Under varying assumptions, state-space models are able to simulate flip-flops and counting \citep{sarrof2024the,alsmann2026on}.
The techniques we have developed to handle counting could be extended to handle these cases as well. 

\item \emph{Q: Why not evaluate LLMs?}

We're more interested in the architecture itself. LLM abilities depend on prompt format and are strongly impacted by what's in the training data.
There is work suggesting that the capabilities of LLMs are ultimately bounded by $\CRASP$ \citep{jobanputra2025born}, though a precise investigation in the case of regular language length-generalization is out of the scope of this work.

\item \emph{Q: How does $\CRASP$ compare to other classes, such as subregular classes, circuit classes, and the dot depth hierarchy? Could they on their own already predict transformer length generalization?}

$\CRASP$ is distinct from known classes, and much more successful at predicting length generalization than those are.
Here, we will expand on \cref{fig:inclusions}.
In fact, all existing classes do not predict length-generalization on transformers.
\begin{itemize}
    \item Regular and subregular: The regular languages in $\CRASP$ define a strict subset of the star-free languages. As we show, even $\vty{R}^\omega$ only covers star-free languages (and in fact it is a strict subset). 
    \item Circuit classes: The circuit classes most relevant to transformers are $\mathsf{AC}^0$ and $\mathsf{TC}^0$; neither of them captures $\CRASP$ or transformer length generalization well.
    $\CRASP$ is a strict subset of $\mathsf{TC}^0$, as was shown in \cite{huang2025formalframeworkunderstandinglength}.
    This extends to regular languages:   E.g., the PARITY language $(b^*ab^*ab^*)$ (checking if the number of $a$'s is even) is in $\mathsf{TC}^0$ but not in $\CRASP$.
    In general, $\CRASP$ is incomparable with $AC^0$. 
    Interestingly, on the level of regular languages, $\CRASP$ is a strict subset of $\mathsf{AC}^0$: $\CRASP$ only includes star-free languages (and all of those are in $AC^0$), but, for instance, $\{a,b\}^*b$ is in $\mathsf{AC}^0$ but not in $\CRASP$.
    \item Dot-depth hierarchy. Given the containment inside the star-free languages, one could consider  stratifications of these languages. One of the most well-studied ones is the dot-depth hierarchy. $\CRASP$ can express every bounded-depth dyck languages, and thus can touch every single level of the dot-depth hierarchy, while not covering the entire hierarchy.
\end{itemize}

\item \emph{Q: What about the role of positional encodings?}

This is an important question for future work that would result in a different characterization than the one we have derived here. 
\citet{huang2025formalframeworkunderstandinglength} showed that $\CRASP[\mathsf{periodic},\mathsf{local}]$ characterizes the languages which transformers with APE can length-generalize on. 
Obtaining a decision procedure on the algebraic side would require additional techniques from the ones used in this paper.

\item \emph{Q: Why did you not use MLRegTest for your experiments \citep{JMLR:v25:23-0518}?} 

We looked into this, but unfortunately a significant portion of the languages in the dataset contain \emph{strictly local} patterns, which are not expressible in $\CRASP$ without positional encodings \citep{huang2025formalframeworkunderstandinglength}.
Future work extending the characterization of $\CRASP$ to handle positional encodings should use this benchmark. 

\item \emph{Q: The experiments deliberately prevent residual-stream access to the last token in state prediction. Why? Could the theory handle the case where one actually wants this access?}

We argue that state tracking should be robust to the insertion of ``do-nothing'' actions not changing the state, or other unrelated extra material.
In formal language theory, such actions are referred to as \emph{neutral symbols}. 
Invariance under the inclusion of such neutral symbols is a natural property of characterizations in terms of syntactic monoids, as we have done here.
An example where this plays a role is $\{a,b\}^* b$ (where the state -- \emph{is the last symbol seen a $b$?} -- can be easily tracked based on the last symbol), which has the same syntactic monoid as the language $\{a,b,e\}^* b e^*$ resulting from adding a neutral symbol $e$, which is difficult for Transformer length generalization \citep{liu2023exposing, huang2025formalframeworkunderstandinglength}. Neither is definable in $\CRASP$.

\item \emph{Q: Why do your empirical results differ from that of \citet{li2025characterizing} and \citet{huang2025formalframeworkunderstandinglength} who both test length generalization for regular languages in $\CRASP$?}

We think addressing the empirical and theoretical results of \citet{huang2025formalframeworkunderstandinglength} and \citet{li2025characterizing} in light of our findings is an important question.
\citet{li2025characterizing} suggested that transformers consistently fail to length-generalize on languages outside of $\vty{R}$, providing evidence when trained on length $N$ and tested on length $12N$.
While \citet{huang2025formalframeworkunderstandinglength} showed length-generalization in languages in $\CRASP\setminus\vty{R}$, they only trained on length $N$ and tested up to $3N$, so the results may be incomparable.
In the present work, we specifically probed languages in $\CRASP\cap\vty{REG}\setminus\vty{R}$ and found length-generalization from length $N$ to $10N$, suggesting a more optimistic picture of length-generalization than \citet{li2025characterizing}.

We suspect the difference in our empirical observations might be attributed to the small number of languages in $\CRASP\cap\vty{REG}\setminus\vty{R}$ tested by \citet{li2025characterizing}  -- only $3$ such languages were tested.
Indeed, prior to this work there existed no sizeable dataset of languages in $\CRASP\cap\vty{REG}\setminus\vty{R}$, as there was no computable membership criterion for $\CRASP\cap\vty{REG}$.

An additional difference is that \citet{li2025characterizing} use a language \emph{classification} task where the model is trained on both positive and negative examples. 
The experiments done by \citet{bhattamishra-etal-2020-ability,huang2025formalframeworkunderstandinglength,yang2025knee}, and this work on $\CRASP$ all use a \emph{next-token prediction} setup.

\end{enumerate}

\newpage

\section{Additional Experimental Results}
\label{app:experiments}

In this section, we provide additional experimental results supporting Section \ref{main:experiments}. We first report complementary analyses on the language suite from Table \ref{tab:languages} used for the experiments in Figure \ref{fig:length_gen} shown in the main paper. Specifically, we present best seed results grouped directly by $\CRASP$ membership in Section \ref{app:crasp_membership_comp}, results across multiple random seeds in Section \ref{app:seeds}, and experiments with increased training data (from 10K to 100K) in Section \ref{app:train_size}.

We then evaluate a second suite of more complex regular languages with greater nesting depth (Table \ref{tab:languages_complex}). For these languages, we extend the training range from $[l_{min},50]$ to $[l_{min},200]$ and report the corresponding experiments in Section \ref{app:complex}. We present length generalization results across the classes $\vty{R}$, $\CRASP$, and $\vty{R} \circ \vty{G}$ in Section \ref{app:length_gen_complex}, directly compare languages in and outside of $\CRASP$ in Section \ref{app:crasp_membership_comp_complex}, and evaluate performance across multiple random seeds in Section \ref{app:seeds_complex}.

Finally, Section \ref{app:languages} provides the complete lists of regular languages used in both experimental suites, together with their membership in $\vty{R}$, $\vty{R}^\omega$, $\vty{R} \circ \vty{G}$, and $\CRASP$.

\subsection{Additional Results on the Main Language Suite}
\label{app:main_suite}

\subsubsection{Results by C-RASP Membership}
\label{app:crasp_membership_comp}

The main paper reports length generalization results separately for languages in $\vty{R}$, $\CRASP \setminus \vty{R}$, $\vty{R} \circ \vty{G} \setminus \CRASP$, and outside $\vty{R} \circ \vty{G}$ (Figure \ref{fig:length_gen}). Here, we provide an alternative view of the same experimental results by grouping languages according to their $\CRASP$ membership only. Figure \ref{fig:2-panel-main} shows a clear separation between languages, where languages in $\CRASP$ reliably generalize beyond the training length range, whereas languages outside $\CRASP$ consistently fail to do so.

\begin{figure}[H]
\centering
\includegraphics[width=\linewidth]{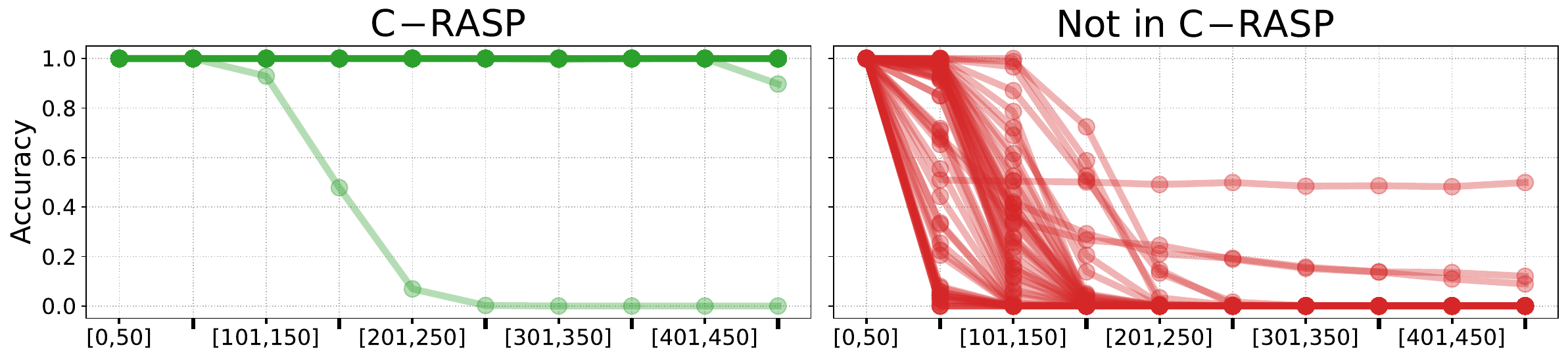}
\caption{\textbf{Length generalization by $\CRASP$ membership.}
Models are trained on strings of lengths in $[l_{min},50]$ and evaluated on ranges from $[l_{min},50]$ (in-distribution) up to $[451,500]$ in steps of 50. Each curve corresponds to a language from Table \ref{tab:languages}. %
}
\label{fig:2-panel-main}
\end{figure}

\subsubsection{Length Generalization Across Seeds and Languages}
\label{app:seeds}

The results in the main paper report the best successful seed for each language. To assess whether the observed length generalization behavior is robust across random seeds (i.e. random model initializations), we additionally aggregate results over the five best successful seeds per language, where a seed is considered successful if it achieves $100\%$ accuracy on the in-distribution test data. Figure \ref{fig:seed-main} reports these results across $\vty{R}$, $\CRASP \setminus \vty{R}$, $\vty{R} \circ \vty{G} \setminus \CRASP$, and outside $\vty{R} \circ \vty{G}$, while Figure \ref{fig:2-panel-seed-main} groups them only by $\CRASP$ membership. Results show that languages outside $\CRASP$ consistently fail to length generalize across languages and across random seeds.

\begin{figure}[H]
\centering
\includegraphics[width=\linewidth]{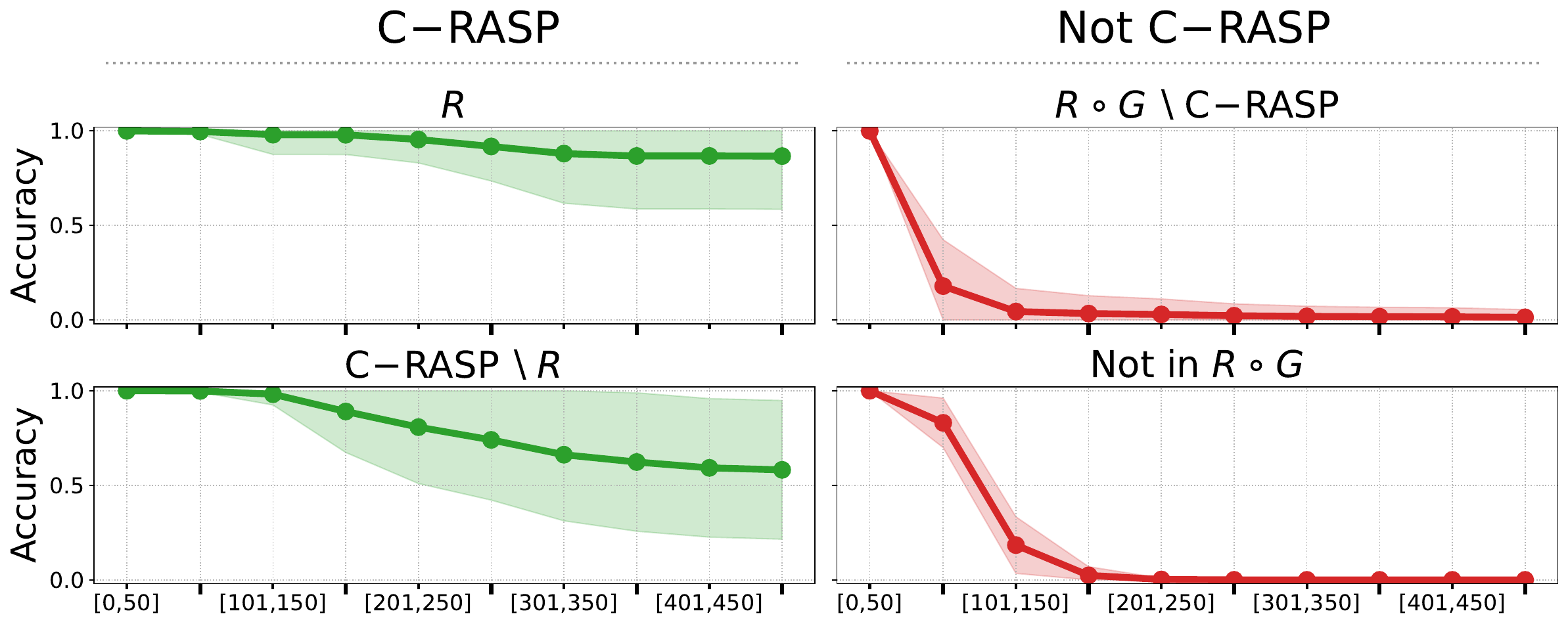}
\caption{\textbf{Length generalization on regular languages (aggregated across languages and seeds).}
Models are trained on strings of lengths in $[l_{min},50]$ and evaluated on length ranges from $[l_{min},50]$ (in-distribution) up to $[451,500]$ in steps of 50.
For each language we compute the mean accuracy across its 5 best successful seeds, where a seed is considered successful if it achieves $100\%$ accuracy on the in-distribution test data.
Solid lines show the mean of these per-language curves within each group, and shaded regions show one standard deviation across languages.
Green corresponds to languages in $\CRASP$, while red corresponds to languages not in $\CRASP$. Each panel includes all languages in the corresponding group from Table \ref{tab:languages}.}
\label{fig:seed-main}
\end{figure}

\begin{figure}[H]
\centering
\includegraphics[width=\linewidth]{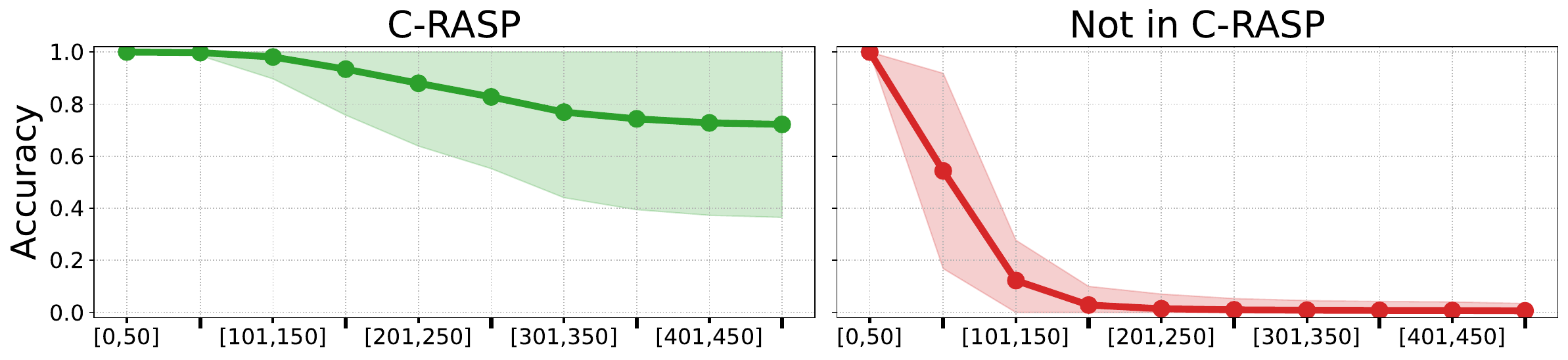}
\caption{\textbf{Length generalization by $\CRASP$ membership, aggregated across languages and seeds.}
Models are trained on strings of lengths in $[l_{min},50]$ and evaluated on length ranges from $[l_{min},50]$ (in-distribution) up to $[451,500]$ in steps of 50. For each language, we compute the mean accuracy across its 5 best successful seeds, where a seed is considered successful if it achieves $100\%$ accuracy on the in-distribution test data. Solid lines show the mean of these per-language curves for languages within and outside $\CRASP$, and shaded regions show one standard deviation across languages. In contrast to Figure \ref{fig:seed-main}, which separates languages into $\vty{R}$, $\CRASP \setminus \vty{R}$, $\vty{R} \circ \vty{G} \setminus \CRASP$, and those outside $\vty{R} \circ \vty{G}$, here we group the same languages by $\CRASP$ membership only, providing an overall view of length generalization within and outside $\CRASP$. Each panel includes all languages in the corresponding group from Table \ref{tab:languages}.}
\label{fig:2-panel-seed-main}
\end{figure}

\subsubsection{Increased Training Data Size}
\label{app:train_size}

Finally, we test how increased training data size affects length generalization within different classes. We repeat the experiments with 100K sampled words per language, using an 80/20 train-test split, compared to 10K training examples in the original experiments. We otherwise follow the same training and hyperparameter search procedure described in Section \ref{main:experiments}. Figure \ref{fig:main-longer-train} shows that increasing the amount of training data does not change trends in length generalization: languages in $\CRASP$ continue to generalize beyond the training lengths, while languages outside $\CRASP$ do not. In Figure \ref{fig:2-panel-main-longer-train} we again regroup the same results from Figure \ref{fig:main-longer-train} into languages in and outside of $\CRASP$ only.

\begin{figure}[H]
\centering
\includegraphics[width=\linewidth]{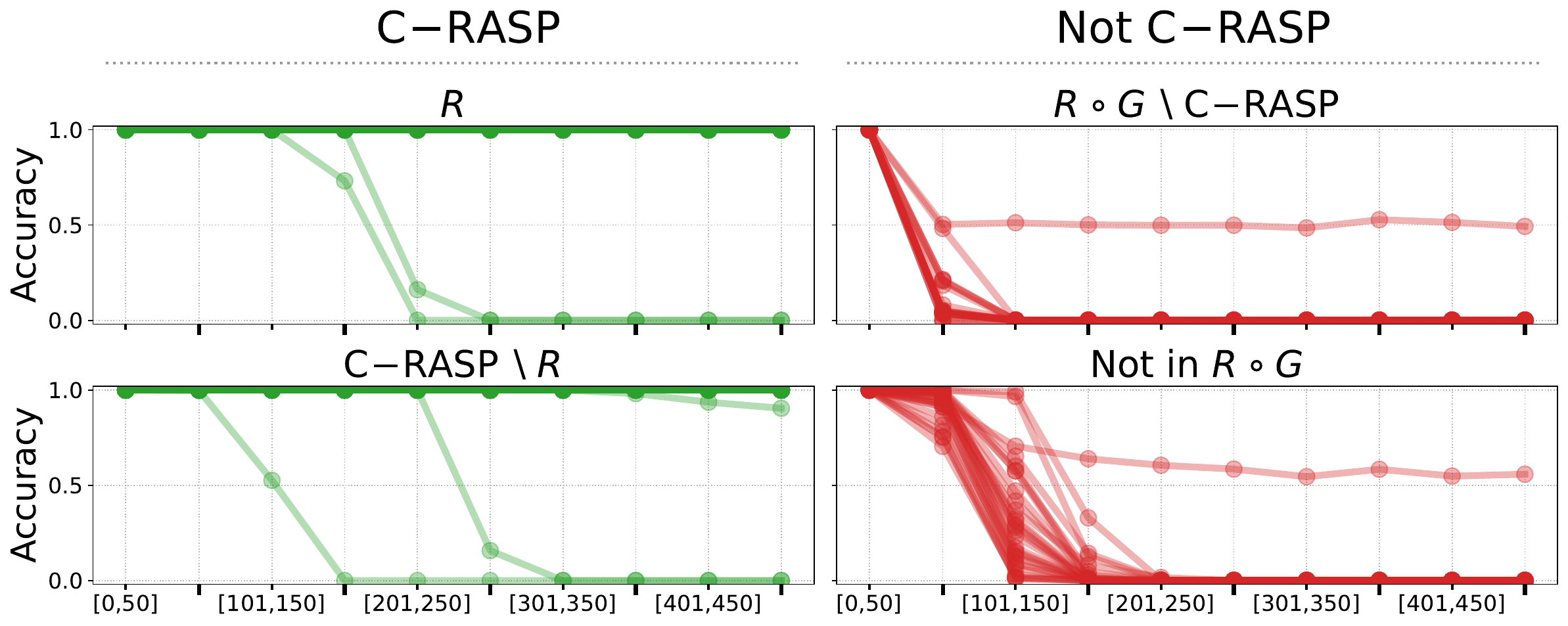}
\caption{\textbf{Length generalization on regular languages with increased training data.}
Models are trained on strings of lengths in $[l_{min},50]$ using a larger training set (100K examples instead of 10K), and evaluated on lengths from $[l_{min},50]$ (in-distribution) up to $[401,500]$ in steps of 50. Each curve corresponds to a language. We report the best seed per language. Increasing the training data does not change trends in length generalization: languages in $\CRASP$ continue to generalize, while languages not in $\CRASP$ consistently fail to generalize beyond the training range. For these experiments, we used the systematically generated subset of languages listed in Table \ref{tab:languages}, consisting of 100 languages in total.}
\label{fig:main-longer-train}
\end{figure}

\begin{figure}[H]
\centering
\includegraphics[width=\linewidth]{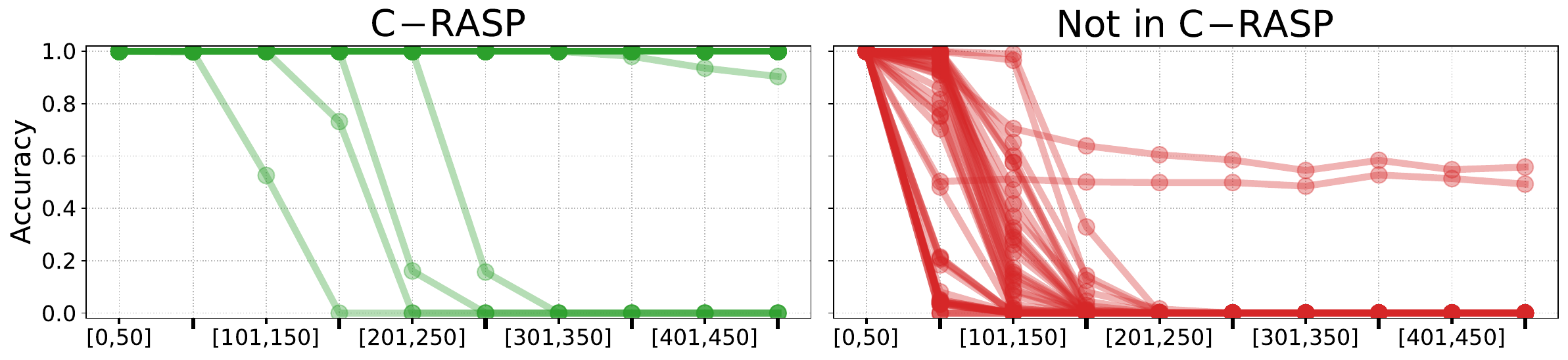}
\caption{\textbf{Length generalization by $\CRASP$ membership with increased training data.}
Models are trained on strings of lengths in $[l_{min},50]$ using a larger training set (100K examples instead of 10K) and evaluated on length ranges from $[l_{min},50]$ (in-distribution) up to $[401,500]$ in steps of width 50. Each curve corresponds to a language, and we report the best seed per language. In contrast to Figure \ref{fig:main-longer-train}, which separates languages into the four groups $\vty{R}$, $\CRASP\!\setminus\!R$, $\vty{R} \circ \vty{G}\!\setminus\!\CRASP$, and languages outside $\vty{R} \circ \vty{G}$, here we group the same languages solely by $\CRASP$ membership, providing an overall view of length generalization within and outside $\CRASP$. For these experiments, we used the systematically generated subset of languages listed in Table \ref{tab:languages}, consisting of 100 languages in total.}
\label{fig:2-panel-main-longer-train}
\end{figure}

\subsection{Experiments on More Complex Languages}
\label{app:complex}

To test whether our findings extend to more complex languages, we repeat our experiments on a systematically generated set of 50 regular languages with greater nesting depth (Table \ref{tab:languages_complex}), following the same experimental procedure described in Section \ref{main:experiments}. In addition, we increase the maximum training length from 50 to 200, allowing us to examine length generalization behavior when models are trained on longer strings.

\subsubsection{Length Generalization}
\label{app:length_gen_complex}

In Figure \ref{fig:main-complex} we show results across $\vty{R}$, $\CRASP \setminus \vty{R}$, $\vty{R} \circ \vty{G} \setminus \CRASP$, and those outside $\vty{R} \circ \vty{G}$, confirming length generalization trends seen on simpler languages and shorter training lengths. Figure \ref{fig:2-panel-main-complex} highlights this length generalization trends further by dividing languages into subsets of in and outside of $\CRASP$ providing a more summarized overview. 

\begin{figure}[H]
\centering
\includegraphics[width=\linewidth]{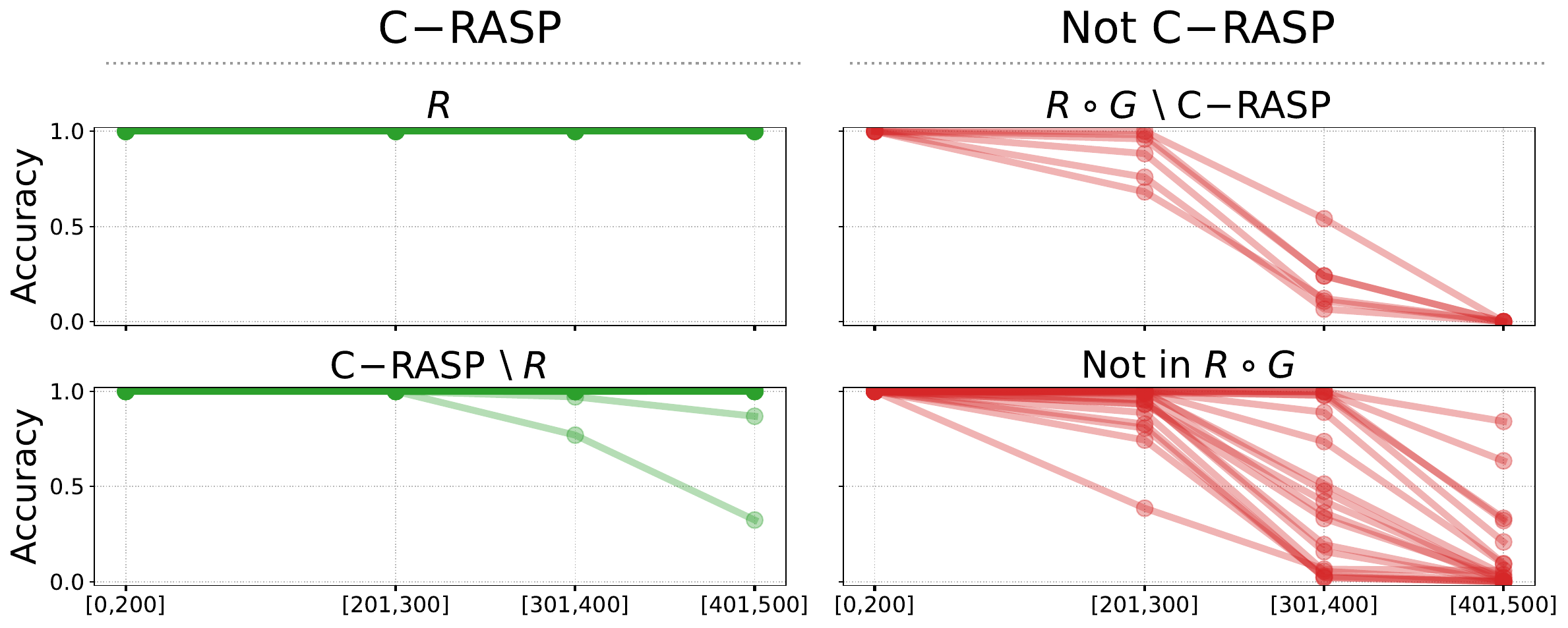}
\caption{\textbf{Length generalization on regular languages (longer training length).} Models are trained on strings of lengths in $[l_{min},200]$, and evaluated on length ranges $[l_{min},200]$ (in-distribution) up to $[401,500]$ in steps of 100. Each curve corresponds to a language. Green curves denote languages in $\CRASP$, while red curves denote languages not in $\CRASP$. Languages in $\CRASP$ maintain near-perfect accuracy well beyond the training range, whereas languages outside $\CRASP$ exhibit rapid degradation, typically failing shortly after. Languages corresponding to this plot can be viewed in Table \ref{tab:languages_complex}.}
\label{fig:main-complex}
\end{figure}

\subsubsection{Results by C-RASP Membership}
\label{app:crasp_membership_comp_complex}

\begin{figure}[H]
\centering
\includegraphics[width=\linewidth]{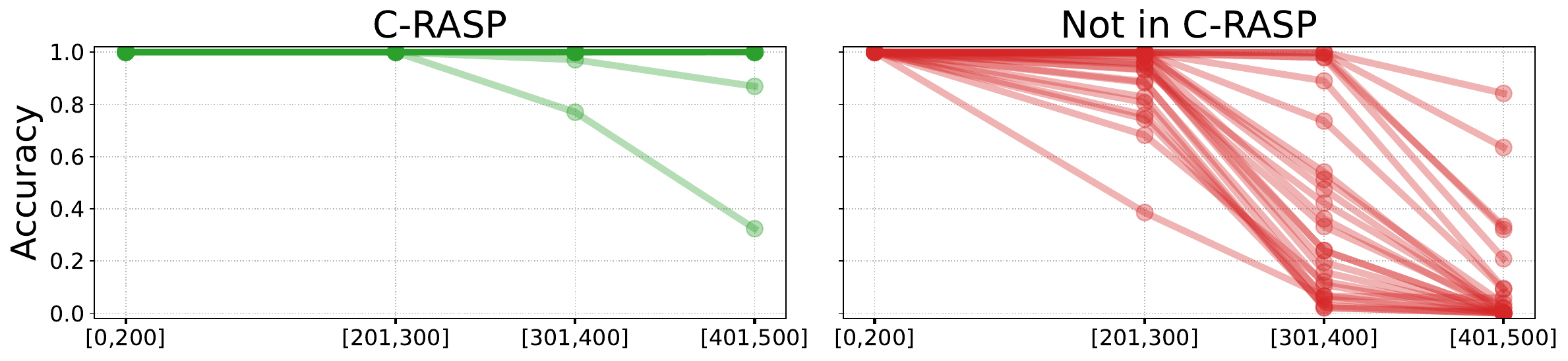}
\caption{\textbf{Length generalization by $\CRASP$ membership (longer training length).}
Models are trained on strings of lengths in $[l_{min},200]$ and evaluated on length ranges from $[l_{min},200]$ (in-distribution) up to $[401,500]$ in steps of 100. Each curve corresponds to a language. In contrast to Figure \ref{fig:main-complex}, which separates languages into the four groups $\vty{R}$, $\CRASP\!\setminus\!R$, $\vty{R} \circ \vty{G}\!\setminus\!\CRASP$, and languages outside $\vty{R} \circ \vty{G}$, here we group the same languages by $\CRASP$ membership, providing an overall view of length generalization within and outside $\CRASP$. Each panel includes all languages in the corresponding group from Table \ref{tab:languages_complex}.}
\label{fig:2-panel-main-complex}
\end{figure}

\subsubsection{Length Generalization Across Seeds and Languages}
\label{app:seeds_complex}

Similar to Section \ref{app:seeds}, we evaluate how consistent length generalization trends are across random seeds (i.e. random model initializations) and languages within a class. We compute the mean accuracy across the five best successful seeds for each language and aggregate these results across languages. Figure \ref{fig:seed-main-complex} shows that same length generalization trends persist across seeds: languages in $\CRASP$ maintain high accuracy beyond the training length range, while languages outside $\CRASP$ consistently degrade with increasing length. Figure \ref{fig:2-panel-seed-main-complex} further emphasizes this difference by grouping languages in and outside of $\CRASP$.

\begin{figure}[H]
\centering
\includegraphics[width=\linewidth]{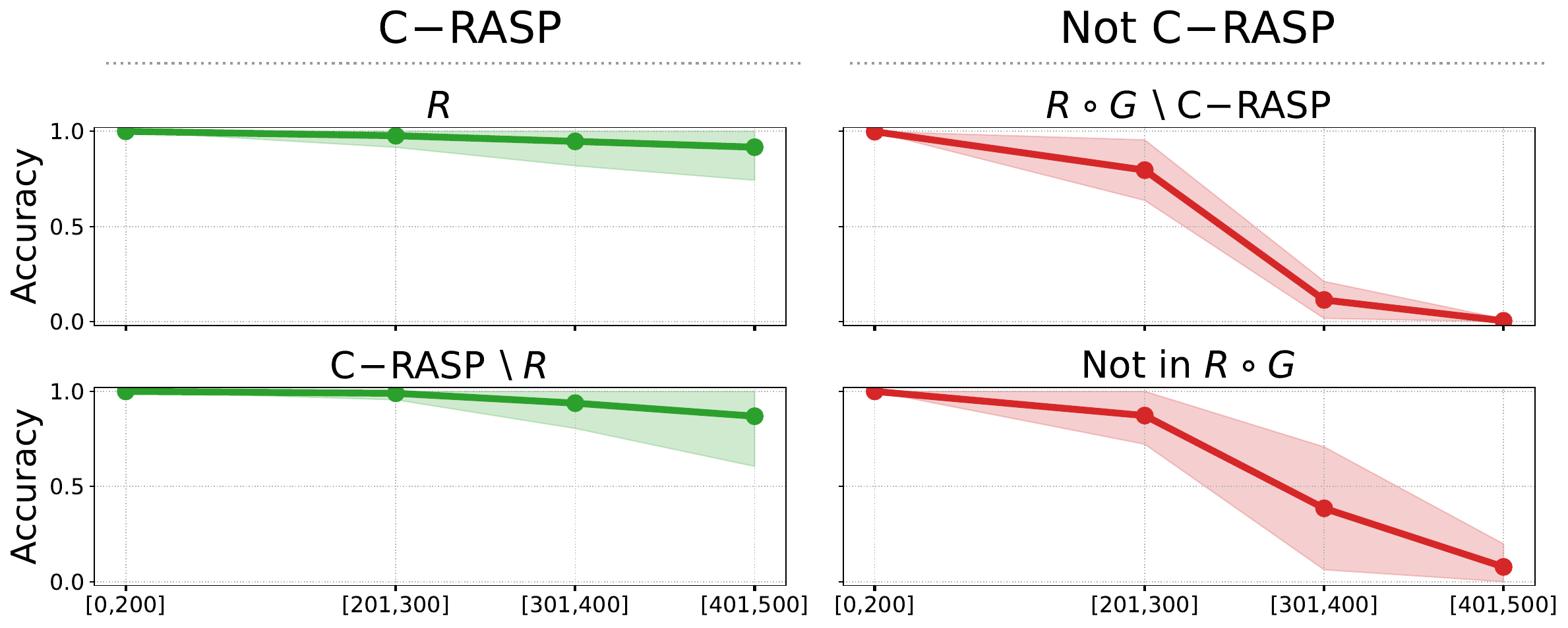}
\caption{\textbf{Length generalization on regular languages (aggregated across languages and seeds).}
Models are trained on strings of lengths in $[l_{min},200]$ and evaluated on length ranges from $[l_{min},200]$ (in-distribution) up to $[401,500]$ in steps of 100.
For each language we compute the mean accuracy across its 5 best successful seeds, where a seed is considered successful if it achieves $100\%$ accuracy on the in-distribution test data.
Solid lines show the mean of these per-language curves within each class, and shaded regions show one standard deviation across languages.
Green corresponds to languages in $\CRASP$, while red corresponds to languages not in $\CRASP$.
Results for individual best seeds are shown in Figure \ref{fig:main-complex}.
Each panel includes all languages in the corresponding class from Table \ref{tab:languages_complex}.}
\label{fig:seed-main-complex}
\end{figure}

\begin{figure}[H]
\centering
\includegraphics[width=\linewidth]{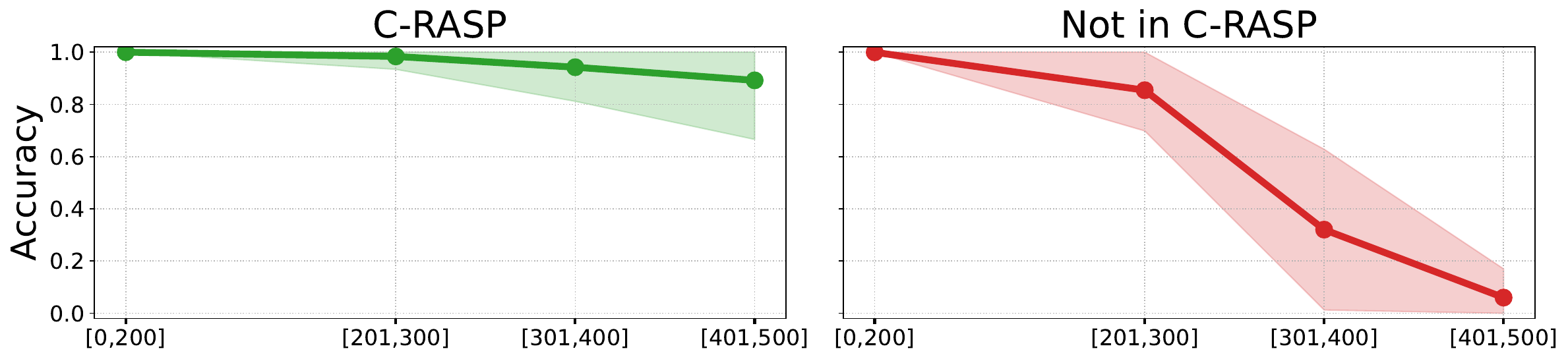}
\caption{\textbf{Length generalization by $\CRASP$ membership, aggregated across languages and seeds.}
Models are trained on strings of lengths in $[l_{min},200]$ and evaluated on length ranges from $[l_{min},200]$ (in-distribution) up to $[401,500]$ in steps of 100. For each language, we compute the mean accuracy across its 5 best successful seeds, where a seed is considered successful if it achieves $100\%$ accuracy on the in-distribution test data. Solid lines show the mean of these per-language curves for languages within and outside $\CRASP$, and shaded regions show one standard deviation across languages. In contrast to Figure \ref{fig:seed-main-complex}, which separates languages into the four groups $\vty{R}$, $\CRASP\!\setminus\!R$, $\vty{R} \circ \vty{G}\!\setminus\!\CRASP$, and languages outside $\vty{R} \circ \vty{G}$, here we group the same languages solely by $\CRASP$ membership, providing an overall view of length generalization within and outside $\CRASP$. Each panel includes all languages in the corresponding group from Table \ref{tab:languages_complex}.}
\label{fig:2-panel-seed-main-complex}
\end{figure}

\subsection{Regular Languages}
\label{app:languages}

We report the complete set of regular languages used in our experimental evaluation in the main paper and appendix. Tables \ref{tab:languages} and \ref{tab:languages_complex} provide a comprehensive overview of the dataset. For each language, we indicate membership in $\vty{R}$, $\vty{R}^\omega$, $\vty{R} \circ \vty{G}$, and $\CRASP$.

\begin{table*}[t]
\centering
\scriptsize
\setlength{\tabcolsep}{3.5pt}
\renewcommand{\arraystretch}{0.9}

\begin{tabular}{>{\raggedright\arraybackslash}p{0.1828\textwidth}>{\centering\arraybackslash}m{0.040\textwidth}>{\centering\arraybackslash}m{0.040\textwidth}>{\centering\arraybackslash}m{0.040\textwidth}>{\centering\arraybackslash}m{0.040\textwidth}|>{\raggedright\arraybackslash}p{0.318\textwidth}>{\centering\arraybackslash}m{0.040\textwidth}>{\centering\arraybackslash}m{0.040\textwidth}>{\centering\arraybackslash}m{0.040\textwidth}>{\centering\arraybackslash}m{0.040\textwidth}}
\toprule
Formal Language & $\vty{R}$ & \makebox[0pt][c]{$\CRASP$} & $\vty{R}^\omega$ & \makebox[0pt][c]{$\vty{R} \circ \vty{G}$} & Formal Language & $\vty{R}$ & \makebox[0pt][c]{$\CRASP$} & $\vty{R}^\omega$ & \makebox[0pt][c]{$\vty{R} \circ \vty{G}$} \\
\midrule
$(bbac)^*$ & False & True & True & True & $(ab)^* + (bb)^*$ & False & False & False & True \\
$(bab + b)^*$ & False & False & False & False & $(bb)^*(bb)^*$ & False & False & False & True \\
$(c + b(a)^*)^*$ & False & False & False & False & $((b)^*ac)^*$ & False & False & False & False \\
$b + (bb)^*$ & False & False & False & True & $(bc(c)^*)^*$ & False & False & False & False \\
$acabcc(c)^*$ & True & True & True & True & $(baa + a)^*$ & False & False & False & False \\
$bbcc(aa)^*$ & False & False & False & True & $(ac)^*c + ba + a$ & False & True & True & True \\
$(ab(a)^*)^*$ & False & False & False & False & $(b + a)^*aaac$ & False & False & False & False \\
$(cc)^*$ & False & False & False & True & $(bb)^*cbac$ & False & False & False & True \\
$((a)^*ac)^*$ & False & False & False & False & $((b)^*)^*(ab)^*$ & False & True & True & True \\
$((b)^*ab)^*$ & False & False & False & False & $bba(c + a)^*$ & True & True & True & True \\
$(ac(a)^*)^*$ & False & False & False & False & $cb(a)^*baa$ & True & True & True & True \\
$(cca + a)^*$ & False & False & False & False & $(b)^*baa$ & True & True & True & True \\
$cb(a)^* + abbc$ & True & True & True & True & $(aa + (b)^*)^*$ & False & False & False & True \\
$aa(a)^*c$ & True & True & True & True & $aa + ca(aa)^*$ & False & False & False & True \\
$(aa + aa)^*$ & False & False & False & True & $(a)^*(a)^*b + aac$ & True & True & True & True \\
$(a + c)^*cbbb$ & False & False & False & False & $(bb)^*ab + ac$ & False & False & False & True \\
$bc(b)^*abbb$ & True & True & True & True & $(cc)^*cccb$ & False & False & False & True \\
$(ca)^*(cb)^*$ & False & True & True & True & $(b)^*cbbcaa$ & True & True & True & True \\
$(cc)^*ccb + b$ & False & False & False & True & $((b)^*b)^*$ & True & True & True & True \\
$((a)^*)^*(ac)^*$ & False & True & True & True & $ccc(ba)^*$ & False & True & True & True \\
$bccb(aa)^*$ & False & False & False & True & $bcac(a)^* + ba$ & True & True & True & True \\
$(ba)^*(b)^*bc$ & False & True & True & True & $(b)^*cc(ca)^*$ & False & True & True & True \\
$(ac)^*ba(b)^*$ & False & True & True & True & $(cab + c)^*$ & False & False & False & False \\
$(a)^*b + acab$ & True & True & True & True & $(a + ac + a)^*$ & False & False & False & False \\
$(ca)^*(a)^*bb$ & False & True & True & True & $aba + c(bb)^*$ & False & False & False & True \\
$(b)^*(b)^*(aa)^*$ & False & False & False & True & $((a)^*ca)^*$ & False & False & False & False \\
$b + ca + b(c + a)^*$ & True & True & True & True & $(abbc)^*$ & False & True & True & True \\
$(cbb + b)^*$ & False & False & False & False & $(a + c)^*(cb)^*$ & False & False & False & False \\
$c + cab + (bc)^*$ & False & True & True & True & $(aa)^*bb(a)^*$ & False & False & False & True \\
$(c)^*b + b(b)^*$ & True & True & True & True & $a(bc)^*$ & False & True & True & True \\
$(ca(c)^*)^*$ & False & False & False & False & $(ca)^*(ab)^*$ & False & True & True & True \\
$bab(ba)^*$ & False & True & True & True & $((ac)^*)^*$ & False & True & True & True \\
$bbbc(ca)^*$ & False & True & True & True & $(ba + a + c)^*$ & False & False & False & False \\
$(abcb)^*$ & False & True & True & True & $((a)^*)^*(cc)^*$ & False & False & False & True \\
$cc + cc(bb)^*$ & False & False & False & True & $((b)^*a + c)^*$ & False & False & False & False \\
$bcba + (a)^*aa$ & True & True & True & True & $(ba)^*cb + aa$ & False & True & True & True \\
$a + c(a)^*cb + a$ & True & True & True & True & $(baaa)^*$ & False & True & True & True \\
$(cb(c)^*)^*$ & False & False & False & False & $(baa + b)^*$ & False & False & False & False \\
$((b)^*bc)^*$ & False & False & False & False & $ccc(b)^*$ & True & True & True & True \\
$(a + aac)^*$ & False & False & False & False & $aa(c)^* + acca$ & True & True & True & True \\
$(cb)^*caa + b$ & False & True & True & True & $(a)^*caaa + c$ & True & True & True & True \\
$caaba + (a)^*$ & True & True & True & True & $((a)^*bb)^*$ & False & False & False & False \\
$(cc)^*bac$ & False & False & False & True & $(bb)^*(c)^*ac$ & False & False & False & True \\
$(b + ac + c)^*$ & False & False & False & False & $(b)^*b + aa + bba$ & True & True & True & True \\
$(acba)^*$ & False & True & True & True & $(b + bc)^*$ & False & False & False & False \\
$cbac(a)^*$ & True & True & True & True & $(b)^*cc(b)^*ca$ & True & True & True & True \\
$(a)^*(a)^*(bb)^*$ & False & False & False & True & $aca(a)^*$ & True & True & True & True \\
$((a)^*ab)^*$ & False & False & False & False & $bcbb(bb)^*$ & False & False & False & True \\
$baca(bc)^*$ & False & True & True & True & $(a)^*a(b + b)^*$ & True & True & True & True \\
$(ca)^*(cc)^*$ & False & False & False & True & $b + c + ac(ac)^*$ & False & True & True & True \\
$(ab + aabb)^*$ & False & False & False & False & $(ab + bbaa)^*$ & False & True & True & True \\
$(aa)^*$ & False & False & False & True & $(a^+b^+)+$ & False & False & False & False \\
$(ab)^+ a^+$ & False & True & True & True & $(ab)^+ a^+ b^+$ & False & True & True & True \\
$(ab)^+ a^+ b^+ a^+$ & False & True & True & True & $((ab)^+ b^+)^+$ & False & False & False & False \\
$((ab)^+ b^+)^k$ & False & True & True & True & $(ab + ba)^*$ & False & True & True & True \\
$(ab)^+ b(ab)^+$ & False & True & True & True & $(ab + bba)^*$ & False & False & True & True \\
$(a+b+)^k$ & True & True & True & True & ${abe}^*be^*$ & False & False & False & False \\
$(ab)^*$ & False & True & True & True & $(a(ab)^*b)^*$ & False & True & True & True \\
$b\Sigma^*$ & True & True & True & True & $\Sigma^*b$ & False & False & False & False \\
$(\Sigma \!\setminus\! \{a,b_0\})^*a(\Sigma \!\setminus\! \{b_1\})^*$ & True & True & True & True & $(\Sigma \!\setminus\! \{a_1,b_0\})^*a_1(\Sigma \!\setminus\! \{a_2,b_1\})^*a_2(\Sigma \!\setminus\! \{b_2\})^*$ & True & True & True & True \\
$\Sigma^*a\Sigma^*$ & True & True & True & True & $\Sigma^*ab\Sigma^*$ & True & True & True & True \\
$a^*(ba^*ba^*)^*$ & False & False & False & True & $\Sigma^*a\Sigma^*b\Sigma^*$ & True & True & True & True \\
$(\Sigma \!\setminus\! {b_0})^*a(\Sigma \!\setminus\! {a,b_1})^*$ & False & False & False & False & & & & & \\
\bottomrule
\end{tabular}

\caption{Set of 125 regular languages used in the experiments in Figure \ref{fig:length_gen} and in Section \ref{app:main_suite}. For each language, we report membership in $\vty{R}$, $\vty{R}^\omega$, $\vty{R}\wrc\vty{G}$, and $\CRASP$ (where True in a column denotes membership in the column's class). For each language, we sampled 10K words for the training set (lengths $l_{min}$–50) and 1K words for each evaluation length bin.
}
\label{tab:languages}

\end{table*}

\begin{table}[t]
\centering
\scriptsize
\setlength{\tabcolsep}{4pt}
\renewcommand{\arraystretch}{0.9}

\begin{tabular}{l>{\centering\arraybackslash}m{0.035\textwidth}>{\centering\arraybackslash}m{0.035\textwidth}>{\centering\arraybackslash}m{0.035\textwidth}>{\centering\arraybackslash}m{0.035\textwidth}}
\toprule
Formal Language & $\vty{R}$ & \makebox[0pt][c]{$\CRASP$} & $\vty{R}^\omega$ & \makebox[0pt][c]{$\vty{R} \circ \vty{G}$}\\
\midrule
$(ab)^*bb(aabb)^*$ & False & True & True & True \\
$(abb)^*(a)(c+a)(a)cbb+a^*(c)$ & False & True & True & True \\
$(b^*ca+ca)(ac+b)(c+a)^*(b)^*+(bc)^*(b+c)(c)+a(c)^*caa$ & False & True & True & True \\
$(b^*bc+bb)c(a)(a)c^*aa+(bc)^*(b)^*(b)b^*aa+(a+a)(b)aba+(c)^*c^*bc+b^*+bc+a$ & False & True & True & True \\
$(cba)^*a$ & False & True & True & True \\
$(cac)^*b(b+c)^*a^*+bb+a(b)a^*ba+(a)^*c^*bb$ & False & True & True & True \\
$(cb)^*ab+(aa+a)^*a(a)bba$ & False & True & True & True \\
$(ac+aa+c)(ca+b+b)(c+a)^*ba^*ab+(ab+c)^*(b+c)^*(b)^*a^*ac+c+b^*abc+b^*cb+bc+c+b$ & False & True & True & True \\
$(c^*cc+b)(bc)^*(c+a)^*+(bc+a)^*+b$ & False & True & True & True \\
$(acc)^*+(ab+c+b)(a+b)^*+(a+b)+cabc$ & False & True & True & True \\
$(cac)^*(ca+c+c)(c)(b)^*ccc+(ac+c+a)^*a^*(c)^*ba+(a+a)^*(a)a$ & False & True & True & True \\
$(bcb)^*(b)(b)^*(c)^*a^*$ & False & True & True & True \\
$(a^*ac+ba+a+a)(ab+c)^*ba^*b^*bc$ & False & True & True & True \\
$(bcb+c+c+c)(bc+c+b)(b+a)^*(a)^*bba$ & False & False & False & False \\
$(a^*cc+bb+b)(ac)^*b^*+(c+b+b)^*(a)(b)a^*ba+(a+c)^*cc$ & False & False & False & False \\
$(bcb+cb+a+c)^*$ & False & False & False & False \\
$(ccb+ac+c+b)+(cc+b)^*(b+c)^*(c)cba+a^*$ & False & False & False & False \\
$b^*+(cb+a)^*(c+b)^*c^*c^*c+(a+a)^*(c)b^*+(b)^*$ & False & False & False & False \\
$(b^*c+ba)^*c^*+(aa+b+c)(a+c)a^*ccc$ & False & False & False & False \\
$(a^*ab+cb+a+c)^*$ & False & False & False & False \\
$(cbb+cb+a+a)(ca+b+c)^*(c+b)^*(b)^*b^*+(cb)^*(a+b)(b)+b(b)c^*ba$ & False & False & False & False \\
$(b^*ab+ab+a+c)^*(cb+b)(b)^*+c^*(b+b)^*(c)^*$ & False & False & False & False \\
$(a^*bb+ba+b+b)^*+b^*(b+a)(c)acb+c(b)b^*b+aaba+cc$ & False & False & False & False \\
$(c^*bb+bb+c)+(ab+c+c)^*ab+(b+c)^*a+ab$ & False & False & False & False \\
$(c^*+bb+c+a)(ac+a)c(c)a^*ab+(cb+b+a)(a)a^*c+(c+b)^*(b)b^*a$ & False & False & False & False \\
$(b^*ca+cb+b)^*(a+a)(b+a)a$ & False & False & False & False \\
$(baa+ab+b)^*$ & False & False & False & False \\
$(c^*+ac+a+b)^*b^*(c+b)$ & False & False & False & False \\
$(a^*cb+ca+b+a)c+a^*c^*c^*+(b+a)^*(a)^*b^*ac+(b)a^*c$ & False & False & False & False \\
$a^*+(aa+a+c)(a)(b)^*b+(b+a)^*(a)bba$ & False & False & False & False \\
$(abb+b)(ba+c+c)(c+a)+(b)(a+b)^*+(c)b^*c^*+(a)^*a^*a$ & True & True & True & True \\
$(b^*cc+bb+b+a)(ca+b)$ & True & True & True & True \\
$(b^*b+ca+b+c)(a+a+c)^*(a+c)+(aa+c)a(a)$ & True & True & True & True \\
$(a^*)^*$ & True & True & True & True \\
$(aca)^*(a+c+a)^*+(ba+b+a)^*cc^*c^*+(c)(b)^*b+b^*aab+b$ & True & True & True & True \\
$(cbc)(bc)a^*(c)^*a^*ab+a^*(a+b)(c)c^*bb+c^*(c)^*a^*aa$ & True & True & True & True \\
$b^*(ba+c)c^*c^*+c(b+b)^*bccc+(c)a^*c^*bb+a^*c^*cc$ & True & True & True & True \\
$b+b^*(c)(b)cca+ac^*baa+(b)^*ac+b^*cb+b+c+b$ & True & True & True & True \\
$(acc+c+c+b)(a)(a+a)^*(a)^*bbb+(ba+a+b)+a^*ca^*aa+(b)+acc$ & True & True & True & True \\
$a^*b^*+a(b+a)(a)b^*ac$ & True & True & True & True \\
$(a^*b)^*+(ac+b+a)b^*bc^*cb$ & False & False & False & False \\
$(a^*ac)+c^*b(c)^*+(c+b)^*+(b)aac$ & True & True & True & True \\
$a^*+(a)^*$ & True & True & True & True \\
$(a+bb)^*a^*$ & False & False & False & True \\
$(b^*bc)+(aa+c)^*$ & False & False & False & True \\
$(cc+aa+a+b)(aa+b+b)^*+b^*b+(b+b)(c)aca+(a)^*acb+cab$ & False & False & False & True \\
$(b+aa+b+c)^*$ & False & False & False & True \\
$b(aa+c)^*(b)(b)cca$ & False & False & False & True \\
$(a)+(bb)^*+a^*a^*cacb$ & False & False & False & True \\
$(b^*cc)^*$ & False & False & False & False \\
\bottomrule
\end{tabular}

\caption{Set of 50 regular languages used in the experiments in \ref{app:complex}. For each language, we report membership in  $\vty{R}$, $\CRASP$, $\vty{R}^{\omega}$ and $\vty{R}\wrc\vty{G}$  (where True in a column denotes membership in the column's class). For each language, we sampled 10K words for the training set (lengths $l_{min}-200$) and 1K words for each evaluation length bin.}
\label{tab:languages_complex}
\end{table}

\section{Algebraic Preliminaries}\label{app:algebra}
\begin{defin}[Recognition; Syntactic Monoid]
    Let $\lang\subseteq \Sigma^*$ be a language. A monoid $M$ \emph{recognizes} $\lang$ iff there is a homomorphism $h\colon \Sigma^*\to M$ and a subset $X\subseteq M$ such that $\lang=h^{-1}(X)$. The syntactic monoid $M(\lang)$ of a language $\lang$ is the minimal monoid (up to isomorphism) which recognizes $\lang$.
\end{defin}
\begin{defin}[Basic Units]
    The three basic semigroup units $\uone, \utwo,\uthree$ are given with their multiplication tables, where rows denote the first operand and columns denote the second. 
    \begin{center}
    \begin{minipage}[position]{0.33\textwidth}
        \begin{tabular}{c|cc}
            $\uone$ & 0 & 1\\
            \hline
            0     & 0 & 0\\
            1     & 0 & 1 
        \end{tabular}
    \end{minipage}%
    \begin{minipage}[position]{0.33\textwidth}
        \begin{tabular}{c|cc}
            $\utwo$ & a & b\\
            \hline
            a     & a & b\\
            b     & a & b 
        \end{tabular}
    \end{minipage}%
    \begin{minipage}[position]{0.33\textwidth}
        \begin{tabular}{c|ccc}
            $\uthree$ & 1 & a & b\\
            \hline
            1     & 1 & a & b\\
            a     & a & a & b \\
            b     & b & a & b
        \end{tabular}
    \end{minipage}
    \end{center}
\end{defin}

\begin{definition}[Basic monoid operations]
    We define basic monoid operations.
    \begin{itemize}[noitemsep]
        \item Submonoid. A monoid $M$ is a submonoid of $N$ (usually written $M\leq N$) whenever $M\subseteq N$ and $M$ is closed under the monoid operation of $N$.
        \item Direct product. The direct product $M\times N$ of monoids $M$ and $N$ has elements $(m,n)$ where $m\in M $ and $n\in N$, with the operation $(m_1,n_1)\cdot_{M\times N}(m_2,n_2)=(m_1\cdot_M m_2,n_1\cdot_N n_2)$.
        \item Homomorphism. A homomorphism of monoids $\phi\colon M\to N$ is a function such that $\phi(1_M)=1_N$ and $\phi(m_1)\phi(m_2)=\phi(m_1m_2)$.
        \item Division. A monoid $M$ divides a monoid $N$ (written $M\preceq N$) whenever there is a submonoid $X\leq N$ of $N$ and a homomorphism $\phi\colon X\to M$ such that $M=\phi(X)$
    \end{itemize}
\end{definition}

\begin{defin}[Pseudovariety]
    A pseudovariety of monoids is a class closed under submonoids, division, and finite direct products. A pseudovariety of languages is a class closed under inverse homomorphism, Boolean operations, and shifts ($a^{-1}\lang b^{-1}$).
\end{defin}
\begin{defin}
    Define $\vty{V}\wrc\vty{W}$ as the pseudovariety generated by all $V\wrc W$ where $V\in\vty{V}$ and $W\in\vty{W}$.
    Define $\wwp^1(\vty{V})=\vty{V}$, $\wwp^{k+1}(\vty{V})=\wwp^k(\vty{V})\wrc\vty{V}$, and $\wwpc(\vty{V})=\bigcup_{k>0} \wwp^k(\vty{V})$.
    We also define $\wwp^k(M)$ for monoids, taking $\wwp^1(M)$ as the pseudovariety generated by $M$.
\end{defin}

We define some pseudovarieties that we use in the paper:
\begin{defin}
    The following are standard, except for $\vty{Dy}$. 
    \begin{itemize}[noitemsep]
        \item $\vty{R}$ is the pseudovariety of $\mathcal{R}$-trivial monoids \citep{BRZOZOWSKI198032}
        \item $\vty{A}$ is the pseudovariety of all aperiodic monoids 
        \item $\vty{G}$ is the pseudovariety of all finite groups
        \item $\vty{Dy}$ is the pseudovariety generated by $M(\dyck_k)$ for all $k\in\N$
    \end{itemize}
\end{defin}

An important relation on monoid elements we will use is the $\mathcal{R}$ relation \citep{BRZOZOWSKI198032}:

\begin{defin}\label{def:R_equivalent}
    For $s, t \in M$, we say $s \preceq_{\mathcal{R}} t \Leftrightarrow sM \subseteq tM$.

    $\mathcal{R}$-classes are the equivalence classes for $s \sim_\mathcal{R} t \Leftrightarrow \left[s \preceq_{\mathcal{R}} t \wedge t \preceq_{\mathcal{R}} s\right]$.
\end{defin}

\begin{proposition}\label{prop:dyck-krohn-rhodes-cyclic}
    $M(\dyck_k)\preceq (\uone\wrc \Z_{k+3})\wrc \uone$, where $\Z_{k+3}$ is the cyclic group of order $k+3$.
\end{proposition}
\begin{proof}
    First consider the submonoid of $ (\uone\wrc \Z_{k+3})\times \uone$ generated by $((g,0),1)$, $((g,1),0)$, and $((g,k+2),0)$ where $g\in \uone^{\Z_{k+3}}$ given by $g(x)=0\iff x\in\{k+1,k+2\}$.
    It can be verified that $((g,0),1)\mapsto\epsilon$, $((g,1),0)\mapsto a$, and $((g,k+2),0)\mapsto b$ extends to a surjective homomorphism $h\colon\uone\wrc \Z_{k+3}\to M(\dyck_k)\times \uone$.
    Finally, using the fact that direct products divide wreath products, we conclude that $M(\dyck_k)\preceq (\uone\wrc \Z_{k+3})\wrc \uone$.
\end{proof}
Intuitively, computation in $\Z_{k+3}$ detects if the depth ever exceeds $k+1$, and this depth-violation is detected by $\uone$. Another $\uone$ detects non-emptiness of the string.

\section{\Ptls{}}\label{app:ptl}

In this section we abstract from $\CRASP$ to the class of \emph{\ptls{}}, which intuitively are a class of straight-line programs in which each operation at position $i$ can only depend on previously defined operations and positions $j\leq i$. 
This content follows \citet{doi:10.1137/S0097539797322772}.

\begin{definition}[\Ptls{}]
    \Ptls{} are those with the syntax
    \begin{align*}
        \phi & ::= \sigma\mid  \lnot \phi_1 \mid \phi_1\land\phi_2 \mid\mathcal{O}\langle \phi_1,\phi_2,\ldots, \phi_k\rangle
    \end{align*}
    where $\sigma\in\Sigma$ and $\mathcal{O}$ is some operator of arity $k$. For each operator there is an associated collection $K_\mathcal{O}\subseteq \Sigma^*\times (2^\N)^k$ of words and sets of positions contained in the operator. 
    Semantics are defined
    \begin{align*}
        &w,i\models \sigma &&\iff &&w_i=\sigma\\
        &w,i\models \lnot\phi &&\iff && w,i\not\models\phi\\
        &w,i\models \phi_1\land \phi_2 && \iff &&w,i\models \phi_1 \text{ and } w,i\models\phi_2\\
        &w,i\models \mathcal{O}\langle \phi_1,\phi_2,\ldots, \phi_k\rangle &&\iff && (w_{\leq i}, \{j\mid w,j\models\phi_1\},\ldots, \{j\mid w,j\models \phi_k\})\in K_\mathcal{O}
    \end{align*}
\end{definition}
This is just defining Lindstr\"om quantifiers, which can be instantiated by the typical logics as follows.
\begin{example}
    Linear Temporal Logic is the class of \Ptls{} with the binary $\since$ operator, typically written infix. 
    \begin{align*}
    K_{\since} &= \{(w,R_1,R_2) \mid \exists k\in R_2\text{ st } k<|w|  \text{ and }  [k,|w|)\subseteq R_1\}
    \end{align*}
    In this way
    \begin{align*} 
        w,i\models \phi_1\since \phi_2 &\iff (w_{\leq i}, \{j\mid w,j\models\phi_1\},\{j\mid w,j\models\phi_2\})\in K_{\since}\\
        &\iff \text{there exists $k<i$ such that $w,k\models\phi_2$ and $w,j\models\phi_1$ for all $k<j<i$}
    \end{align*}
    Similarly, $\CRASP$ is the class of \ptls{} with the operator $(\Lambda,C)=\sum_{1\leq m\leq k} \lambda_m\cdot\countl\phi_m\geq C$
    \begin{align*}
    K_{(\Lambda,C)} &= \left\{(w,R_1,R_2,\ldots, R_k) \middle| \sum_{1\leq m\leq k} \lambda_m\cdot |R_m|\geq C\right\}
    \end{align*}
    In this way 
    \begin{align*}
        w,i\models \sum_{1\leq m\leq k} \lambda_m\cdot\countl\phi_m\geq C &\iff (w_{\leq i}, \{j\mid w,j\models\phi_1\},\ldots \{j\mid w,j\models\phi_2\})\in K_{(\Lambda,C)}\\
        &\iff \sum_{1\leq m\leq k} \lambda_m\cdot|\{j\mid w,j\models\phi_m\}|\geq C
    \end{align*}
\end{example}
\begin{definition}\label{def:substitution}
   Let $\Phi$ and $\Psi$ be classes of \ptls{} formulas over alphabet $\Gamma$ and $\Sigma$, respectively. Let $G=\{\psi_\gamma\}_{\gamma\in \Gamma}$ be a family of formulas in $\Psi$. Let $\theta_G$ be a mapping of formulas given by 
    \begin{align*}
        &\theta_G(Q_\gamma) &\mapsto\quad\quad& \psi_\gamma &&\\\
        &\theta_G(\lnot\phi)&\mapsto\quad\quad&\lnot\theta_G(\phi) && \\
        &\theta_G(\phi_1\land\phi_2)& \mapsto\quad\quad&\theta_G(\phi_1)\land\theta_G(\phi_2) && \\
        &\theta_G(\mathcal{O}\langle\phi_1,\phi_2,\ldots,\phi_k\rangle) & \mapsto \quad\quad& \mathcal{O}\langle\theta_G(\phi_1),\theta_G(\phi_2),\ldots,\theta_G(\phi_k)\rangle) && 
    \end{align*}
    We write $\phi[\gamma\mapsto\psi_\gamma]$ for this substitution.
    This is called a $\Psi$ substitution of $\Phi$ formulas. We write $\Phi\sub\Psi$ for the class of all $\Psi$ substitutions of $\Phi$ formulas. 
    The semantics are as would be expected.
    By default we let this operation be right-associative.

\end{definition}

To help formalize the expressivity of \ptls{}, we define a class of languages.

\begin{definition}[End-Pointed Language]
    A pointed word is a tuple $(w,p)$ for $w\in\Sigma^*$ and $1\leq p\leq |w|$. An end-pointed language is a set of pointed words where the point denotes the end of the prefix of the string used for recognition.
    \begin{itemize}[noitemsep]
        \item For $\Phi$ a class of \ptls{}, $P(\Phi)$ is the set of pointed languages $L$ for which there exists $\phi\in\Phi$ such that $(w,p)\in L$ iff $w,p\models\phi$. 
        \item For $\vty{M}$ a class of typed monoids, $P(\vty{M})$ is the set of pointed languages $L$ for which there exists $(M,\types{M},\units{M})\in\vty{M}$, typed homomorphism $h\colon\Sigma^*\to (M,\types{M},\units{M})$, a type $\type{M}\in\types{M}$, and a finite set $C\subseteq M$ such that $(w,p)\in L$ iff $(h(w_1w_2\cdots w_{p-1}),h(w_p))\in \type{M}\times C$.
    \end{itemize}
\end{definition}

Programs are classically connected to algebraic characterizations of languages via wreath product principles. 
We define how to take two classes of programs and obtain a more complex class.

\begin{definition}[Program composition]
    Let $\Phi,\Psi$ be classes of \ptls{}. The class $\Phi\sub\Psi$ consists of all programs in $\Phi$ where atomic operations may refer to programs in $\Psi$. 
\end{definition}

For instance, $\CRASP_1\sub\CRASP_1$ is equivalent to the class of depth $2$ programs $\CRASP_2$.
We will see in \cref{app:typed_monoids} that \ptls{} can be closely connected to wreath products of monoids.

\section{Typed Monoids}\label{app:typed}
\citet{krebs2008typed} developed a framework for using infinite monoids to recognize languages.
We present a restriction of the aforementioned framework to the case of wreath products (a one-sided version of the block product used in previous work), which ultimately provides an exact algebraic characterization of $\CRASP$. 
The core issue here is that the wreath product of infinite monoids can generate uncountably many elements, which can be too powerful.
\begin{proposition}
    Consider the classic wreath product $\Z\wrc \Z$. Then $M(\lang)\preceq \Z\wrc \Z$ for every $\lang$.
\end{proposition}
\begin{proof}
    Without loss of generality let $\Sigma=\{0,1\}$. Consider the submonoid of $\Z\wrc\Z$ generated by the image of $\Sigma^*$ under the homomorphism $\sigma \mapsto(f_\sigma,1) $ where $f_\sigma(x)=\sigma\cdot 2^{|x|}$.
    In essence, this creates a mapping $w\mapsto (f_w,|w|)$ where $f_w(0)$ outputs the integer value of the binary number $w$.
    Thus, $\Z\wrc\Z$ can recognize arbitrary languages.
\end{proof}

This problem motivates the definition of \emph{typed monoids}, which restricts the accepting sets.
\subsection{Definitions}

\begin{definition}
    A typed monoid is a triple $(M, \types{M}, \units{M})$ where $M$ is a finitely generated monoid, $\types{M}$ is a finite Boolean algebra over $M$, and $\units{M}$ is a finite subset of $M$. Elements of $\types{M}$ are the \emph{types} and elements of $\units{M}$ are the \emph{units}. A language $L$ is recognized by $(M, \types{M}, \units{M})$ if there exists a homomorphism $h\colon \Sigma^*\to M$ such that $h(\Sigma)\subseteq\units{M}$ and $L=h^{-1}(\type{M})$ for some $\type{M}\in \types{M}$.
\end{definition}

As an example, the language $\mathsf{MAJORITY}$ can be recognized by the typed monoid $(\Z,\{(-\infty,0],[1,\infty), \Z, \emptyset\},\{-1,1\})$ via the type $[1,\infty)$ and the homomorphism $a\mapsto 1$ and $b\mapsto -1$.
We will typically refer to this typed monoid as $\Z$.

We define morphisms

\begin{definition}
    Let $(S, \types{S}, \units{S})$ and \( (T, \types{T}, \units{T}) \) be typed monoids. A typed monoid homomorphism $h\colon (S, \types{S}, \units{S})\to (T, \types{T}, \units{T})$ is a triple $(h_S,h_\types{S},h_\units{S})$ such that:

    \begin{itemize}[noitemsep]
        \item $h_S\colon S\to T$ is a monoid homomorphism
        \item $h_\types{S}\colon \types{S}\to\types{T}$ is a homomorphism of Boolean algebras
        \item $\forall \type{S}\in\types{S}, h_S(\type{S})=h_\types{S}(\type{S})\cap h_S(S)$
        \item $\forall \unit{s}\in\units{S}, h_S(\unit{s})=h_\units{S}(\unit{s})$
    \end{itemize}

    And due to the compatibility of $h_S,h_\types{S},h_\units{S}$ we can omit the subscripts. We say that a typed monoid $(S, \types{S}, \units{S})$ recognizes the language $L\subseteq\Sigma^*$ if there is a morphism $h \colon \Sigma^* \to S$ with $h(\Sigma)\subseteq \units{S}$ and a type $\type{S} \in \types{S} $ such that $L = h^{-1}(\type{S})$.

\end{definition}

So we want our functions to be compatible with the finite types, which motivates the following definition. 
In a sense, this requires that all elements of the same type, up to some constant shifting $C$, behave the same under the function.

\begin{definition}[Type-respecting Functions]\label{def:type_respecting}
    Let \( S \) be a set and \( (T, \types{T}, \units{T}) \) be a typed monoid and let \( C \subseteq T \) be a nonempty finite set of constants. A function \( f : T  \to S \) is called type respecting with respect to \( (T, \types{T}, \units{T}) \) and $C$ if the preimage \( f^{-1}(s) \) can be described by a finite Boolean combination of conditions of the form \( t c\in \type{T} \) where \( c \) is a constant in \( T \) (not necessarily in \( \units{T} \)) and \( \type{T} \in \types{T} \).  
\end{definition}

Intuitively, the image of $x$ under a type-respecting function depends only on the type of $xc$ for some qualified set of constants $c$.
Now the typed wreath product is similar to the untyped case, though the functions are constrained to be type-respecting functions.

\begin{definition}[Typed Wreath Product]
    Let \( (M, \types{M}, \units{M}) \), \( (N, \types{N}, \units{N}) \) be two typed monoids, \( C \subseteq N \) be a finite set. The \textit{typed wreath product} 
\[
(U, \types{U}, \units{U}) = (M, \types{M}, \units{M}) \twreath_C (N, \types{N}, \units{N})
\]
of \( (M, \types{M}, \units{M}) \) with \( (N, \types{N}, \units{N}) \) is defined such that
\begin{itemize}[noitemsep]
    \item \( \units{U} \) consists of all elements \( (f, n) \), where \( n \in \units{N} \), and \( f : N \rightarrow \units{M}\) is a type respecting function (see \cref{def:type_respecting}) with respect to \( (N, \types{N}, \units{N}) \) and \( C \)
    \item \( U \) is the submonoid of \( M \wrc N \) generated by \( \units{U} \) 
    \item \( \types{U} \) consists of all types \( \type{U}_{\type{M},\type{N}} = \{(f, n) \in U \mid f(1_{N}) \in \type{M}, n \in \type{N} \} \), where \( \type{M} \in \types{M} \), \( \type{N} \in \types{N} \)
    
\end{itemize}
\end{definition}

\begin{defin}[Typed Monoid Pseudovariety]
    A typed monoid pseudovariety is a class of typed monoids closed under
    \begin{itemize}[noitemsep]
        \item Division
        \item Shifting (changing types by inverse multiplication)
        \item Unit relaxation (swapping out units)
        \item Trivial extension (applying a congruence)
    \end{itemize}
\end{defin}

\begin{defin}[Typed Wreath Product Closure]
    For typed monoid pseudovarieties define $\vty{V}\wrc\vty{W}$ as the pseudovariety generated by all $V\wrt W$ where $V\in\vty{V}$ and $W\in\vty{W}$.
    Define $\wwp^1(\vty{V})=\vty{V}$, $\wwp^{k+1}(\vty{V})=\wwp^k(\vty{V})\wrt\vty{V}$, and $\wwpc(\vty{V})=\bigcup_{k>0} \wwp^k(\vty{V})$.
    We also define $\wwp^k((M,\types{M},\units{M}))$ for typed monoids $(M,\types{M},\units{M})$, taking $\wwp^1((M,\types{M},\units{M}))$ as the pseudovariety generated by $(M,\types{M},\units{M})$.
\end{defin}

Finally, we note the compatibility of the classical wreath product and the typed wreath product, which will become important in our algebraic decision procedure.
\begin{restatable}{lemma}{FiniteDividesTyped}\label{lemma:finite-product-divides-typed-product}
    Assume $M \preceq S, N \preceq T$ where $M, N$ are finite and $S, T$ are typed.
    Then 
    \begin{equation}
        M \wrc N \preceq S^{N} \wrt T
    \end{equation}
    where $\wrc$ is the classical wreath product and $\wrt$ is the typed wreath product. 
\end{restatable}

\begin{proof}
    Call these typed semigroups $(S,\types{S},\units{S})$ and $(T,\types{T},\units{T})$. 
    Let $h_M\colon (S',\types{S}',\units{S}')\to (M,2^M,M)$ and $h_N\colon (T',\types{T'},\units{T}')\to (N,2^N,N)$ define the divisions. 
    We will define a function $h\colon (U,\types{U},\units{U})\to M\wrc N$ where $(U,\types{U},\units{U})\leq (S^N,\types{S}^N,\units{S}^N)\wrt(T,\types{T},\units{T})$. First, let $U$ be the submonoid of $(S^N)^T\times T$ generated by $(f_d,t)$ where $f_d$ for $d\in (S')^N$ is defined such that $f_d(1_T)=d$ and $f_d(t)=[n\mapsto d(n+h_N(t))]$ for all $t\in T'$, and $f_d(t)=1_{(S')^N}$ for $t\in T\setminus T'$.
    Note that all $(f,t)$ in $U$ satisfy the following constraints:
    \begin{enumerate}[noitemsep]
        \item $t\in T'$
        \item $\im(\im(f))\subseteq S'$
        \item $f(t_1+t_2)(n)=f(t_1)(n+h_N(t_2))$ 
        \item $f(t_1)=f(t_2)$ whenever $h_N(t_1)=h_N(t_2)$ 
        \item $\im(f(t))=1_{{S'}^N}$ for $t\in T\setminus T'$
    \end{enumerate}

    As a submonoid of the wreath product, \( \types{U} \) consists of all types \( \type{U}_{\type{SN},\type{T}} = \{(f, t) \in U \mid f(1_{T}) \in \type{SN}, t \in \type{T} \} \), where \( \type{SN} \in \types{S'}^N \), \( \type{T} \in \types{T'} \) and \( 1_T \) is the neutral element of \( T \). 
    Let $\units{U}$ consist of $(f,t)$ where $t\in\units{T}'$. 
    Define a function $h\colon (U,\types{U},\units{U})\to M\wrc N$ such that $h((f,t))=(g,h_N(t))$ where $g(n)=h_M(f(1_T)(n))$.
    \begin{itemize}
        \item All functions in $U$ are type-respecting.
        This is because each of the generators $f_d$ is type-respecting -- since $h_N$ is a homomorphism on types $h_N\colon \types{T}\mapsto 2^N$, the image $f_d(t)$ is determined by a boolean combination of conditions on the type of $t$.
        \item $h$ is a surjection, because any $(g,n)\in M\wrc N$:
        \begin{itemize}
            \item There is $f\in U$ such that $g(\cdot)=h_M(f(t)(\cdot))$, because all $d\in (S')^{N}$ are represented in $U$. 
            We just pick $d$ to be compatible with $g$ and $f_d$ witnesses the preimage.
            As for which $d$ to choose, we construct $d$ where for $n_0\in N$, we pick an element $s_0\in h_M^{-1}(g(n_0))$ (which exists by the surjectivty of $h_M$) and set $d$ such that $d(n')=s_0$ for all $n'$ where $g(n')=g(n_0)$.
            \item There is $t\in T'$ such that $h_N(t')=n$ by the surjectivity of $h_N$.
        \end{itemize}
        \item $h$ is a homomorphism on elements. From $(3)$ we get that $f(t_1+t_2)(n)=f(t_1)(n+h_N(t_2))$.
        \begin{align*}
           & h((f_1,t_1)(f_2,t_2))\\
           &=h((f_1+{}^{t_1}f_2,t_1t_2)) &&\\
            &=(n\mapsto h_M((f_1+{}^{t_1}f_2)(1_T)(n)), h_N(t_1t_2)) &&\\
            &=((n\mapsto h_M(f_1(1_T)(n)))+(n\mapsto h_M({}^{t_1}f_2(1_T)(n))), h_N(t_1t_2)) &&\\
            &=((n\mapsto h_M(f_1(1_T)(n)))+(n\mapsto h_M(f_2(1_T+t_1)(n))), h_N(t_1t_2)) &&\\
            &= ((n\mapsto h_M(f_1(1_T)(n)))+(n\mapsto h_M(f_2(1_T)(n+h_N(t_1)))),h_N(t_1)h_N(t_2) &&\text{by (3)}\\
            &= ((n\mapsto h_M(f_1(1_T)(n)))+{}^{h_N(t_1)}(n\mapsto h_M(f_2(1_T)(n))),h_N(t_1)h_N(t_2)\\
            &=(n\mapsto h_M(f_1(1_T)(n)),h_N(t_1))(n\mapsto h_M(f_2(1_T)(n)),h_N(t_2))\\
            &=h((f_1,t_1))h((f_2,t_2))
        \end{align*}
        \item $h$ is a homomorphism on types. Let $\type{U}_{\type{SN},\type{T}}$ be a type of $(U,\types{U},\units{U})$.
        Here, $h(\type{SN})\in 2^{M^N}$ and $h(\type{T})\in 2^N$, because $h_M,h_N$ are homomorphisms on types by assumption\footnote{Note that \cite{krebs2008typed} defined the type of $(f,t)$ independent of $t$ (in order to ease the connection to logic), in which case the image under $h$ would not be a type in $2^{M^N}\times 2^N$. 
        By conditioning the type on $t$, the homomorphism on types goes through without need for additional direct products to enforce the types of $T$.}. 
        First, $h$ respects complements
        \begin{align*}
            h(\comp{\type{U}_{\type{SN},\type{T}}}) & = h(\{(f, t) \in U \mid f(1_{T}) \not\in \type{SN} \text{ or } t \not\in \type{T} \})\\
            &= \{h((f, t)) \in M\wrc N \mid h(f(1_{T})) \not\in h(\type{SN}) \text{ or } h(t) \not\in h(\type{T}) \}\\
            &= 2^{M^N}\times 2^N\setminus \{h((f, t)) \in M\wrc N \mid h(f(1_{T})) \in h(\type{SN}), h(t) \in h(\type{T}) \}\\
            &=\comp{h(\type{U}_{\type{SN},\type{T}})}
        \end{align*}
        And $h$ respects union
        \begin{align*}
            &h(\type{U}_{\type{SN}_1,\type{T}_1}\cup \type{U}_{\type{SN}_1,\type{T}_1}) \\
            &= h(\{(f, t) \in U \mid f(1_{T}) \in \type{SN}_1, t \in \type{T}_1 \}\cup \{(f, t) \in U \mid f(1_{T}) \in \type{SN}_2, t \in \type{T}_2 \})\\
            &= \{h((f, t)) \in U \mid f(1_{T}) \in \type{SN}_1, t \in \type{T}_1 \}\cup \{h((f, t)) \in U \mid f(1_{T}) \in \type{SN}_2, t \in \type{T}_2 \}\\
            &= h(\{(f, t) \in U \mid f(1_{T}) \in \type{SN}_1, t \in \type{T}_1 \})\cup h(\{(f, t) \in U \mid f(1_{T}) \in \type{SN}_2, t \in \type{T}_2 \})\\
            &= h(\type{U}_{\type{SN}_1,\type{T}_1})\cup h(\type{U}_{\type{SN}_1,\type{T}_1}) 
        \end{align*}
        And $h$ preserves the inf and sup
        \begin{align*}
            h(\emptyset)&=\emptyset\\
            h(U) &= \{h(f_d,t)\mid d\in (S')^N, t\in T'\}\\
            &= \{(h_M(d),h_N(t))\mid d\in (S')^N, t\in T'\}\\
            &= M^N\times N
        \end{align*}
    \end{itemize}
\end{proof}

\subsection{Proof of Typed Wreath Product Principle}\label{app:typed_monoids}

First, we show how the type of an element is computed within the wreath product. 
Here we write $\pi_1$ and $\pi_2$ as the projections from the first and second coordinates of $M^{N}\times N$
\begin{lemma}\label{lem:wreath_type_computation}
    Let $h\colon\Sigma^*\to(T,\types{T},\mathcal{T})=(M, \types{M}, \units{M}) \twreath_C (N, \types{N}, \units{N})$ for $C\subseteq N$. Let $\pi_1(\type{T})\in\types{M}$ be such that $t\in\type{T}\iff \pi_1(t)(1_N)\in\pi_1(\type{T})$.
    Then $h(w)\in\type{T}\in\types{T}$ iff

    \begin{align*}
        \sum_{1\leq i\leq |w|}\pi_1(h(w_i))\left(\prod_{1\leq j<i} \pi_2(h(w_j))\right)&\in \pi_1(\type{T})\in\types{M}\\
        \prod_{1\leq j<\leq i} \pi_2(h(w_j))&\in \pi_2(\type{T})\in\types{N}
    \end{align*}
\end{lemma}
\begin{proof}
    We walk through the computation. First, we note that $h(w)\in\type{T}$ iff $\pi_1(h(w))\in\pi_1(\type{T})$, and this type $\pi_1(\type{T})$ exists by the definition of the typed wreath product. 
    Let $h(w_i)=(f_i,n_i)$ where $f_i\in M^N$ and $n_i\in N$. Then compute

    \begin{align*}
        h(w) &= h(w_1)h(w_2)\ldots h(w_{|w|})\\
            & =(f_1,n_1)(f_2,n_2)\ldots (f_{|w|},n_{|w|})\\
            &= \left(f_1+{}^{n_1}f_2+\ldots+ {}^{\left(\prod_{1\leq j\leq |w|-1} n_j\right)}f_{|w|}, \prod_{1\leq i\leq |w|} n_i\right)\\
            &= \left(\sum_{1\leq i\leq |w|} {}^{\left(\prod_{1\leq j<i} n_j\right)}f_i, \prod_{1\leq i\leq |w|} n_i\right)
    \end{align*}

    Recall by the definition of the wreath product that ${}^{n}f(x)=f(xn)$, so ${}^nf(1_N)=f(n)$. 

    \begin{align*}
        \pi_1(h(w))(1_N) \in\pi_1(\type{T})&\iff\sum_{1\leq i\leq |w|} {}^{\left(\prod_{1\leq j<i} n_j\right)}f_i(1_N)\in\pi_1(\type{T})\\
        &\iff \sum_{1\leq i\leq |w|} f_i\left(\prod_{1\leq j<i} n_j\right)\in\pi_1(\type{T})\\
        &\iff\sum_{1\leq i\leq |w|}\pi_1(h(w_i))\left(\prod_{1\leq j<i} \pi_2(h(w_j))\right)\in \pi_1(\type{T})\\
    \end{align*}

    The computation of the second coordinate is routine as is computed in $(N,\types{N},\units{N})$
    \begin{align*}
        \pi_2(h(w)) \in\pi_2(\type{T})&\iff \prod_{1\leq j<\leq i} \pi_2(h(w_j))\in \pi_2(\type{T}).
    \end{align*}
\end{proof}

Now connection between the algebraic and logical formulations is often spelled out using statements in the form of a ``wreath product principle''.

\begin{restatable}{theorem}{TypedWreathPrinciple}\label{thm:typed_wreath_principle}
\begin{thm}[Typed Wreath Product Principle]
    Let $\Phi,\Psi$ be classes of \ptls{} and $\vty{M},\vty{N}$ be pseudovarieties of monoids such that $L(\Phi)=L(\vty{M})$ and $P(\Psi)=P(\vty{N})$.
    Then $ L(\Phi\sub\Psi)=L(\vty{M}\twreath\vty{N})$.
\end{thm}
\end{restatable}
\begin{proof}
    \quad
    \begin{itemize}
        \item Suppose $h\colon\Sigma^*\to(T,\types{T},\mathcal{T})=(M, \types{M}, \units{M}) \twreath_C (N, \types{N}, \units{N})$ for $C_N\subseteq N$.
        For each type $\type{T}_{\type{M},\type{N}}\in\types{T}$ we will construct a formula $\theta_\type{T}\in \Phi\sub\Psi$ such that $w\models\theta_\type{T} \iff h(w)\in\type{T}$.
        Observe from Lemma \ref{lem:wreath_type_computation} that $h(w)\in\type{T}$ iff 
        \begin{align*}
            \sum_{1\leq i\leq |w|}\pi_1(h(w_i))\left(\prod_{1\leq j<i} \pi_2(h(w_j))\right)\in \type{M} &&\text{ and } && \prod_{1\leq j<\leq i} \pi_2(h(w_j))\in \type{N}.
        \end{align*}

        The first coordinate can be viewed as recognition of a language by $\type{M}$ in $M$ via a homomorphism $h_1\colon (\types{N}^C)^*\to (M,\types{M},\units{M})$. By assumption obtain a $\phi\in\Phi$ that recognizes this language. 
        As for the word $w'\in (\types{N}^C)^*$ we have that

        \begin{align*}
            w'_i = (\type{N}_{c_1},\type{N}_{c_2},\ldots, \type{N}_{c_{|C|}}) \iff \bigwedge_{c\in C} \left[\left(c\prod_{1\leq j< i} \pi_2(h(w_j))\right)\in\type{N}_c\right]
        \end{align*}
        Observe that for each $c$, we use $(N,\types{N},\units{N})$ to recognize an end-pointed language accepting with type $\type{N}_c$. Then by assumption we obtain $\psi_c\in\Psi$ recognizing each prefix, and then use the substitution 
        
        $$\phi_{\type{M}}\left[(\type{N}_{c_1},\type{N}_{c_2},\ldots, \type{N}_{c_{|C|}})\mapsto \bigwedge_{c\in C} \psi_c\right]$$

        The second coordinate can be viewed as recognition of a language by $\type{N}$ in $N$. Since $\vty{N}$ is a pseudovariety, $L(\vty{N})\subseteq P(\vty{N})$.
        We obtain a formula $\psi_{\type{N}}\in\Psi$ such that recognizes the same language and define $\theta_{\type{T}}=\phi_{\type{M}}\land\psi_{\type{N}}$.
        This results in a formula $\theta_{\type{T}}\in\Phi\sub\Psi$ which recognizes the same language as $(T,\types{T},\mathcal{T})$. 
        
        \item The other direction is similar. Suppose we have a formula with substitution $\phi[\gamma\mapsto\psi_\gamma]\in\Phi\sub\Psi$. 
        Because $P(\Psi)=P(\vty{N})$, there are typed monoids $N_\gamma$ that can compute the substitution's output $\gamma$ at each position. That is, for each $\gamma$ there is a type $\type{N}_\gamma$ and homomorphism $h\colon \Sigma^*\to N_\gamma$

        \[w,i\models\psi_\gamma \iff  h(w_{<i})\in\type{N}_\gamma\]

        Using this, we can compute the substitution into a word over $\Gamma$.
        Then because $L(\Phi)=L(\vty{M})$, we obtain a typed monoid $M$ that can recognize the resulting language over $\Gamma$.
        The entire computation can thus be computed using a monoid in $\vty{M}\wrt\vty{N}$.
    \end{itemize}
\end{proof}

\section{$\CRASP$}\label{app:crasp}

Here we reiterate some of the definitions of $\CRASP$, which can be found in prior papers by \citet{yang2024counting,huang2025formalframeworkunderstandinglength}.
\subsection{Definitions}

\begin{definition}

The syntax of $\CRASP$ formulas is defined:
    \begin{align*}
        \phi & \bnfto \sympred\sigma \mid \lnot \phi_1 \mid \phi_1\land \phi_2 \mid \sum _{t\in \mathcal{T}} \alpha_t \cdot \countl[\phi_t] \sim k
    \end{align*} 
    where $\sympred\sigma\in\Sigma$, $\alpha_i,k\in\Z$ and $\mathord\sim\in\{<,\leq, =,\geq,>\}$.
    The semantics of formulas is defined as follows:
    \begin{subequations}
        \let\origiff\iff
        \renewcommand{\iff}{\hspace{\tabcolsep}\origiff\hspace{\tabcolsep}}
        \begin{alignat*}{2}
            &\str{w},i \models \sympred\sigma &\iff &\text{$\str{w}_i = \sigma$} \\
            &\str{w},i \models \lnot \phi & \iff &\str{w},i\not\models \phi\\
            &\str{w},i \models \phi_1\land \phi_2 &\iff &\text{$\str{w},i \models \phi_1$ and $\str{w},i \models \phi_2$} \\
            &\str{w},i \models \sum _{t\in \mathcal{T}} \alpha_t t \sim k &\iff &\sum_{t\in \mathcal{T}} \alpha_t \cdot |\{j\in[1,i] \mid \str{w},j \models \phi\}| \sim k.
        \end{alignat*}
    \end{subequations}
    We write $\str{w} \models \phi$ iff $\str{w}\eos, |\str{w}|+1 \models \phi$ where $\eos\not\in\Sigma$ is a special end-of-sequence symbol. and we say that $\phi$ defines the language $\lang(\phi) = \{\str{w} \mid \str{w} \models \phi\}$.
\end{definition}

In the sequel, we will use a \emph{DAG} (directed acyclic graph) representation of $\CRASP$ formulas,
where a subformula $\varphi$ may be used \emph{multiple times} in a formula. 
Such a formula can be thought of as a straight-line \emph{program}, i.e.,
a sequence $\varphi = (\varphi_i)_{i=1}^n$, where $\varphi_i$ is any $\CRASP$
definition that could refer to $\varphi_j$ with $j < i$. 

\begin{definition}
\label{def:TLC_semantics}
The syntax of $\CRASP$ is as follows:
    \begin{align*}
        \phi & \bnfto \sympred\sigma \mid t_1 < t_2 \mid \lnot \phi_1 \mid \phi_1\land \phi_2 && \sigma \in \Sigma && \text{Boolean-valued formulas}\\
        t & \bnfto \countl[\phi_1] \mid t_1 + t_2 
        \mid 1 &&&& \text{integer-valued terms}
    \end{align*} 
    The semantics of formulas is defined as follows:
    \begin{subequations}
        \let\origiff\iff
        \renewcommand{\iff}{\hspace{\tabcolsep}\origiff\hspace{\tabcolsep}}
        \begin{alignat}{2}
            &w,i \models \sympred\sigma &\iff &\text{$w_i = \sigma$} \\
            &w,i \models \lnot \phi & \iff &w,i\not\models \phi\\
            &w,i \models \phi_1\land \phi_2 &\iff &\text{$w,i \models \phi_1$ and $w,i \models \phi_2$} \\
            &w,i \models t_1 < t_2 &\iff &t_1^{w,i} < t_2^{w,i}.
        \end{alignat}
    \end{subequations}
    The semantics of terms is defined as follows:
    \begin{subequations}
        \begin{align}
            \countl[\phi]^{w,i} &= |\{j\in[1,i] \mid w,j \models \phi\}| \\
            (t_1+t_2)^{w,i} &=  t_1^{w,i}+t_2^{w,i} \label{eq:TLC_sem_plus} \\
            1^{w,i} &= 1.
        \end{align}
    \end{subequations}
    We write $w \models \phi$ if $w, |w| \models \phi$, and we say that $\phi$ defines the language $\lang(\phi) = \{w \mid w \models \phi\}$.
\end{definition}
The table below shows how this program works for the string $(())()$, which belongs to $\dyck_2$.
\begin{center} \small \normalfont
\begin{tabular}{rll|cccccc}
\toprule
predicate & definition & description & $\sympred{a}$ & $\sympred{a}$ & $\sympred{b}$ & $\sympred{b}$ & $\sympred{a}$ & $\sympred{b}$ \\
\midrule
$\sympred{a}$ & & is left paren & $\top$ & $\top$ & $\bot$ & $\bot$ & $\top$ & $\bot$ \\
$\sympred{b}$ & & is right paren & $\bot$ & $\bot$ & $\top$ & $\top$ & $\bot$ & $\top$ \\
$C_{\sympred{a}}$ &:= $\countl [\sympred{a}]$ & num of left parens & 1 & 2 & 2 & 2 & 3 & 3 \\
$C_{\sympred{b}}$ &:= $\countl [\sympred{b}]$ & num of left parens & 0 & 0 & 1 & 2 & 2 & 3 \\
$\phi_{\text{low}}$ &:= $C_{\sympred{a}}-C_{\sympred{b}} \geq 0$ & depth above $0$ & $\top$ & $\top$ & $\top$ & $\top$ & $\top$ & $\top$ \\
$\phi_{\text{up}}$ &:= $C_{\sympred{a}}-C_{\sympred{b}} \leq 2$ & depth below $2$ & $\top$ & $\top$ & $\top$ & $\top$ & $\top$ & $\top$ \\
$\phi_{\text{bounded}}$ & := $ \phi_{\text{low}} \land \phi_{\text{up}}$& depth bounded & $\top$ & $\top$ & $\top$ & $\top$ & $\top$ & $\top$ \\
$\phi_{\text{matched}}$ & :=  $\countl[\lnot\phi_{\text{bounded}}]=0$ & depth bounded everywhere & $\top$ & $\top$ & $\top$ & $\top$ & $\top$ & $\top$ \\
$\phi_{\text{balanced}}$ & :=  $C_{\sympred{a}}=C_{\sympred{b}}$ & balanced at end & $\bot$ & $\bot$ & $\bot$ & $\top$ & $\bot$ & $\top$ \\
$\phi_{\dyck_2}$ & :=  $\phi_{\text{matched}}\land\phi_{\text{balanced}}$ & acceptance & $\bot$ & $\bot$ & $\bot$ & $\top$ & $\bot$ & $\top$ \\
\bottomrule
\end{tabular}
\end{center}

The table below shows how this program works for the string $())()($, which does not belong to $\dyck_2$.
\begin{center} \small \normalfont
\begin{tabular}{rll|cccccc}
\toprule
predicate & definition & description & $\sympred{a}$ & $\sympred{b}$ & $\sympred{b}$ & $\sympred{a}$ & $\sympred{b}$ & $\sympred{a}$ \\
\midrule
$\sympred{a}$ & & is left paren & $\top$ & $\bot$ & $\bot$ & $\top$ & $\top$ & $\bot$ \\
$\sympred{b}$ & & is right paren & $\bot$ & $\top$ & $\top$ & $\bot$ & $\top$ & $\bot$ \\
$C_{\sympred{a}}$ &:= $\countl [\sympred{a}]$ & num of left parens & 1 & 1 & 1 & 2 & 2 & 3 \\
$C_{\sympred{b}}$ &:= $\countl [\sympred{b}]$ & num of left parens & 0 & 1 & 2 & 2 & 3 & 3 \\
$\phi_{\text{low}}$ &:= $C_{\sympred{a}}-C_{\sympred{b}} \geq 0$ & depth above $0$ & $\top$ & $\top$ & $\bot$ & $\top$ & $\bot$ & $\top$ \\
$\phi_{\text{up}}$ &:= $C_{\sympred{a}}-C_{\sympred{b}} \leq 2$ & depth below $2$ & $\top$ & $\top$ & $\top$ & $\top$ & $\top$ & $\top$ \\
$\phi_{\text{bounded}}$ & := $ \phi_{\text{low}} \land \phi_{\text{up}}$& depth bounded & $\top$ & $\top$ & $\bot$ & $\top$ & $\bot$ & $\top$ \\
$\phi_{\text{matched}}$ & :=  $\countl[\lnot\phi_{\text{bounded}}]=0$ & depth bounded everywhere & $\top$ & $\bot$ & $\bot$ & $\bot$ & $\bot$ & $\bot$\\
$\phi_{\text{balanced}}$ & :=  $C_{\sympred{a}}=C_{\sympred{b}}$ & balanced at end & $\bot$ & $\bot$ & $\bot$ & $\top$ & $\bot$ & $\top$ \\
$\phi_{\dyck_2}$ & :=  $\phi_{\text{matched}}\land\phi_{\text{balanced}}$ & acceptance & $\bot$ & $\bot$ & $\bot$ & $\bot$ & $\bot$ & $\bot$\\
\bottomrule
\end{tabular}
\end{center}

\subsection{$\CRASP$ as a pseudovariety of Languages}\label{app:crasp_variety}

\begin{restatable}{proposition}{CRASPVariety}\label{prop:crasp-variety}

$\CRASP$ defines a pseudovariety of languages.

\end{restatable}

\begin{proof}
We sketch the proof here. 
\begin{itemize}[noitemsep]
    \item Closed under Boolean combinations: by definition.
    \item Closed under inverse homomorphisms: follow the construction in \citet[Lemma F.2]{yang2025knee} or \citet[Lemma 34]{huang2025formalframeworkunderstandinglength}, but ignoring positional predicates and allowing arbitrary symbols in the image of the homomorphism.
    \item Closed under factors: Let $L$ be in \CRASP, and $a \in \Sigma$. To obtain $a^{-1} L$, we can detect the beginning in \CRASP and simulate the computations at a prefix $a$. To obtain $L a^{-1}$, at $\eos$ we simulate the computations at $a\eos$.
\end{itemize}
\end{proof}

Because the languages form a pseudovariety, there must exist a corresponding class of monoids \citep[Theorem 2]{behle2011eilenberg}.

\begin{remark}
  If we defined recognition without the $\eos$ symbol, the result would not be a pseudovariety, because of the special role played by the final position. For instance, a $\CRASP$ program can separate $\{a,b\}^* b$ from $\{a,b\}^* a$, but taking the inverse of a homomorphism that deletes $e$ and keeps $a,b$ unchanges leads to the two sets $\{a,b\}^* b e^*$ from $\{a,b,e\}^* a e^*$ which no $\CRASP$ program can separate.
\end{remark}

\subsection{Existing Characterizations}\label{app:crasp_reg_previously}

Previous work has explored upper and lower bounds for the regular languages of $\CRASP$ and the related logic $\widehat{\mathsf{MAJ}}_2[<]$, though none have arrived at an exact characterization. 
Here, we summarize a few of these previously known results.

\begin{example}
    We know the following language characterizations
    \begin{enumerate}[noitemsep]
        \item $\vty{R} \subset \CRASP \cap \vty{REG}$, because any $\mathcal{R}$-trivial language is definable using existential quantification to the left, which is implementable in $\CRASP$ ($\countl[\phi]\geq 1)$.
        \item $\CRASP \cap \vty{REG} \subset \vty{A}$, since any periodic regular language like $(aa)^*$ is not definable in $\CRASP$ \citep[Lemma 38]{huang2025formalframeworkunderstandinglength}.
        \item $\Sigma^*b$, is not in \CRASP \citep[Lemma 38]{huang2025formalframeworkunderstandinglength}.
        \item $\Sigma^*bb\Sigma^*$is not in \CRASP, because it is not in the larger class $\widehat{\mathsf{MAJ}}_2[<]$ by Lemma 6.11 in \cite{krebs2008typed}
        \item For any $k$, $(a^+ b^+)^k $  in \CRASP by \cite{yang2025knee}
        \item $(ab)^+ $ in \CRASP
    \end{enumerate}
\end{example}

The next section develops our exact characterization.

\subsection{Algebraic Characterization of $\CRASP$}\label{app:CRASP_Z}

\CRASPZ*
\begin{proof}
    We will show that $L(\CRASP)=L(\wwpc(\Z))$, which will be equivalent to the theorem statement.
    First, we note that $\wwpc(\Z)$ and $\wwp^2(\Z)$ are pseudovarieties of typed monoids by definition.
    Then, we establish that $P(\CRASP_1)=P(\wwp^2(\Z))$, noting that $\CRASP$ is an instance of a \ptl.
    First, by \citet[Lemma A.4]{yang2025knee} every $\CRASP$ program can be written such that it only counts over positions $j<i$, with Boolean operations testing the symbol at $i$. 
    The strict counting can be translated into a typed monoid over $\wwp^2(\Z)$ (since the equations in $\CRASP$ may not just be $\countl \phi \geq 0$, we require another wreath product to add constants into the equation), and the testing of the symbol at $i$ can be handled by constants in $\wwp^2(\Z)$. The other direction is similar -- for every type of $\Z$, there exists a $\CRASP$ program that checks if the running sum is in that type. 
    Thus by \cref{thm:typed_wreath_principle},  $L(\CRASP_1)=L(\Z\wrt\Z)$.
    From this base case we can build up to $L(\CRASP_k)=L(\wwp^{2k}({\Z}))$ by induction, and thus $L(\CRASP)=\bigcup_{k\geq0}L(\CRASP_k)=\bigcup_{k\geq 0} L(\wwpc^k({\Z}))=\wwpc(\Z)$.
\end{proof}

\section{Derived Categories}\label{app:derived_categories}

Here we provide a formal exposition of the ideas that were informally presented in the body of the paper. The notions are based on \cite{tilson1987categories} but with some notational adaptations; we refer to that paper for full formal definition, and for proofs of well-definedness.
As mentioned above, we will principally be interested in finite categories for use as algebraic objects.

\begin{definition}[Category]
    A category $X$ consists of a set of objects $\obj(X)$ and for each $c,c'\in \obj(X)$ a homset of arrows $X(c,c')$, often written as $x\colon c\to c'$ for $x\in X(c,c')$. 
    A category is endowed with the following algebraic structure:
    \begin{itemize}[noitemsep]
        \item For arrows we have an associative composition operation, where $s\colon c\to c'$ and $t\colon c'\to c''$, compose into an arrow $st\colon c\to c''$. Furthermore, for $s\colon c_1\to c_2$, $t\colon c_2\to c_3$, and $v\colon c_3\to c_4$, we have that $(st)v=s(tv)$.
        \item For objects $c$ we have an identity arrow $1_c\colon c\to c$, where $s1_c=s$ and $1_ct=t$ for $s\colon c'\to c$ and $t\colon c\to c''$.
    \end{itemize}
\end{definition}

\begin{remark}
    When the ambient category $C$ is unambiguous, we also write $Hom(x \rightarrow x')$ for $C(x,x')$.
\end{remark}

We will think of a monoid as a single-object category with the monoid elements as arrows of the category.
Between categories we can define relations, which do not necessarily have to be functions.

\begin{definition}[Category relation]
    Let $X,Y$ be categories. A category relation $f\colon X\to Y$ has 
    \begin{itemize}[noitemsep]
        \item An object relation $f\colon \obj(X)\to \obj(Y)$, thought of as a subset of $X\times Y$. 
        \item For any corresponding edge sets $X(c,c')$ and $Y(d,d')$ where $d\in cf,d'\in c'f$, an edge set relation $f\colon X(c,c')\to Y(d,d')$. 
    \end{itemize}
and satisfies the property that $\#f$, defined as follows, is a subcategory $\#f$ of $X\times Y$:
    \begin{itemize}[noitemsep]
        \item $\obj(\#f)=\{(c,d)\colon d\in cf\}$
        \item $\#f[(c,d),(c',d')]=\{(x,y) \mid x\in X(c,c'), y\in Y(d,d'), y\in xf\}$
    \end{itemize}
\end{definition}

The most important kind of relation for us will be a relational morphism, denoted $\phi\colon X\relml Y$.
In a sense, this is a generalization of a homomorphism where you can take any function on the atomic elements of the $X$ as a generator of the resulting relation $X\relml Y$ after we close under composition in $X$.

\begin{definition}[Relational Morphism]
    A \emph{relational morphism} $f : C \relml C'$ between categories $C$ and $C'$ is a category relation where the object relation is a function and each hom-set relation is fully-defined. 
    A relational morphism where the hom-set relations are injective is called a \emph{division}.
\end{definition}

\begin{remark}
It is convenient to view relational morphisms as set-valued functions. That is, if $f : C \relml C'$ and $\alpha \in C(x,x')$ is an arrow in $C$, then we write $f(\alpha)$ for the set $\{ \beta \in C'(y,y') : (\alpha, \beta) \in \#f[(x,x'),(y,y')]\}$, where $(x,y), (x',y') \in Obj(\#f)$.

The condition that $\#f$ be a category implies in particular
\begin{equation}
f(\alpha)f(\beta) \subseteq f(\alpha\beta)
\end{equation}
whenever $\alpha \in C(x,x'), \beta \in C(x',x'')$. It also implies that the image of an identity arrow at some object of $C$ always includes the identity arrow at the corresponding target object in $C'$.
\end{remark}

The notion of homomorphism and division for finite monoids is a special case of the definitions above for categories. 
Then, as discussed above, the analogue of the kernel of a group homomorphism (and thus the analogue of a ``divisor'') for monoids is the \emph{derived category}.
The worked example in \cref{sec:dyck_derived} hopefully provides some intuition on the structure of the derived category.

\begin{definition}[Derived Category]
    Let $\phi\colon M\relml N$ be a relational morphism between monoids. The \emph{derived category} $D_\phi$ of the relational morphism is defined with $\obj(D_\phi)=\phi(M)$ and $\Hom(n_1,n_2)=\{n_1\arrow{(m,n)}\colon \phi^{-1}(n_1)\to \phi^{-1}(n_2)\mid (m,n)\in \#\phi, n_1n=n_2\}$,
    where $n_1\arrow{(m,n)}$ is a function $\phi^{-1}(n_1)\to \phi^{-1}(n_1n)$ mapping any $x \in \phi^{-1}(n_1)$ to $xm \in \phi^{-1}(n_1n)$.
    Composition of arrows is given by $\left(n_0\arrow{{(m_1,n_1)}}\right) \left(n_1\arrow{{(m_2,n_2)}}\right)=n_0\arrow{(m_1m_2,n_1n_2)}$.
\end{definition}

\begin{remark}
We note that, as the arrows denote functions on subsets of $M$, it is possible for $n_1\arrow{(m,n)}$ and $n_1\arrow{(m',n)}$ to be identical even if $m \neq m'$, provided $m$ and $m'$ act identically on $\phi^{-1}(n_1)$.

Where helpful for notational clarity, we explicitly include the end object, writing $n_1\arrow{(m,n)}$ as $n_1\arrow{(m,n)}n_2$ where $n_2 = n_1n$. 
\end{remark}

For readers familiar with the corresponding construction for groups, we provide some intuition. 
In some sense, $\ker\phi$ for $\phi\colon G\to H$ records what information is lost when compressing $G$ into $H$. 
By recording what elements collapse into $1_H$, we can reconstruct how every other component of $G$ collapses into $H$ (taking advantage of the inverses in the group to form connections between elements). 
Then, taking the quotient $G/(\ker\phi)$ precisely records what information is lost by this collapse, and by enriching $H$ with this information again we reconstruct $G$ via the division $G\preceq (\ker\phi)\wrc H$.

In the case of monoids, we lack the nice closure properties of groups, and thus $\ker\phi$ cannot be used to reconstruct the collapsing behavior of every other component of $G$. 
So the corresponding structure must be enriched with additional information.
Indeed, in the derived category construction, for each element of $n\in N$ we consider subsets of $M$ which collapse into $n$ via a relational morphism, and then must record specific information about the interactions of elements within and between these subsets. 
The derived category just stores the essential amount of information in order to reconstruct $M$ via a wreath product $M\preceq V\wrc N$, formalized by the Derived Category Theorem \citep{tilson1987categories}:

\begin{thm}[Derived Category Theorem]
    \quad
    \begin{enumerate}
        \item Let $\phi\colon M\relml N$ be a relational morphism of monoids, and let $V$ be a monoid satisfying $D_\phi\preceq V$. Then there is a division of monoids $\theta\colon M\preceq V\wrc N$.
        \item Let $\theta\colon M\preceq V\wrc N$ be a division of monoids, and let $\phi=\theta\pi\colon M\relml N$ be the associated relational morphism. Then $D_\phi\preceq V^N$. 
    \end{enumerate}
\end{thm}

\section{Algebraic Decision Procedure}\label{app:decision}

\subsection{Relevant Lemmas}
\begin{lemma}
    If $M\preceq S$ and $N\preceq T$ then $M\wrc N\preceq S\wrc T$.
\end{lemma}
\begin{proof}
    This is a standard fact which we sketch out here. 
    For some $S'\leq S$ and $T'\leq T$ there exists surjections $h_M\colon S'\to M$ and $h_N\colon T'\to N$. 
    Let $G'$ be the subset of $S^T\times T$ generated by all $(f,t')$ where $t'\in T'$, $\im(f)\subseteq S'$, and $f(t)=1_S$ for $t\not\in T'$, with the wreath product action inherited from $S\wrc T$.
    Define $h\colon G\to M\wrc N$ by $h(f,t')=(f',h_N(t'))$ where $f'(n)=h_M(f'(t))$ for some $t\in h_N^{-1}(n)$ (the choice does not matter). 
    It is clear that $h$ is a surjective homomorphism via inheritance from $h_M$ and $h_N$.
\end{proof}

\subsection{Bounded-Depth Dyck Monoids}\label{sec:dyck-monoids-congruence}

To make constructions computable despite the infinity of $\Z$, we use bounded-depth Dyck monoids $\dyck_k$ as partial stand-ins for $\Z$.
The non-$\bot$ elements of the monoid $\dyck_k$ can be canonically represented as tuples $(h,h_\downarrow, h_\uparrow)$ where, for a word $w \in \{a,b\}^*$, the syntactic morphism $\eta$ maps it to
\begin{align*}
h = & \sum_{i=1}^{|w|} \phi(w_i) \\
h_\uparrow =&  \max_j \sum_{i=1}^{j} \phi(w_i) \\
h_\downarrow = & \min_j \sum_{i=1}^{j} \phi(w_i) \\
\end{align*}
under $\phi(a) = 1$, $\phi(b) = -1$, provided these numbers are all in $[-k,\dots, k]$; otherwise the word is mapped to $\bot$.
The monoid $\dyck_k$ has a natural action on the set 
\begin{equation}\label{eq:truncated-integers}
S_k := [-k, \dots, k] \cup \{\bot\}
\end{equation}
given by
\begin{equation}\label{eq:action-dyck-counts}
    j \cdot (h, h_\uparrow, h_\downarrow) = \begin{cases}
\bot & \text{if }j=\bot \\
h+j & \text{if } j+h_\downarrow \geq -k \wedge j+h_\uparrow \leq k \\
\bot & \text{else}
    \end{cases}
\end{equation}
Formally, the pair $(\dyck_k, S_k)$ is a \emph{transformation monoid}.\footnote{We note that $S_k$ itself is not a monoid, because truncated addition is not associative. The monoid $\dyck_k$ describes the monoid of bounded incrementing/decrementing operations on $S_k$. Together, they provide an algebraic model of bounded counting sufficient for our purposes.} We could develop the decidability proof in terms of wreath products of $(\dyck_k, S_k)$, which are somewhat different from wreath products of $\dyck_k$; this would have some advantages because the set $S_k$ naturally behaves as a truncated version of the infinite set $\Z$.
However, to avoid introducing more technical notions, we stay on the level of monoids, at the cost of factoring out an extra congruence out of the derived category. %
Specifically, there is a natural right-congruence on $\dyck_k$, namely $n \sim m \Leftrightarrow 0n = 0m$. 
It is a right-congruence in the sense that $0n = 0m \Rightarrow 0nr = 0mr$ for any $n,m,r \in \dyck_k$.
When considering derived categories for relational morphisms to $\dyck_k$, we will factor this right congruence out of the object set; this is a well-defined construction resulting in a category because $\sim$ is a right-congruence. The result resembles the derived category, but has a more coarse-grained object set consisting of the equivalence classes of $\sim$.

\subsection{Decidability Proof (via Finite Derived Categories)}\label{app:algebraic_decision_procedure}

The following formalizes a basic construction in the theory of derived categories; it permits going from a derived category to a wreath product decomposition:
\begin{defin}[Extension]
Let $M, S, T$ be finite monoids.
    Given relational morphisms $\phi : M \relml T$, $\phi' : D_{\phi} \relml S$, we define
     $\extension{\phi}{\phi'} : M \relml (S \wrc T)$ as
     \[\extension{\phi}{\phi'}(m) = \{(f,t) \in S^T \times T : t \in \phi(m); \forall x \in Obj(D_\phi) : f(x) \in \phi'(x \rightarrow_{(m,t)} x t)\}\]
\end{defin}

\begin{remark}
    This construction is made at the top of p. 116, case (a) of the proof of Theorem 5.2 of \cite{tilson1987categories}. There, it is also proven that it is a relational morphism. It is also proven that if $\phi'$ is a division, then so is $\extension{\phi}{\phi'}$.
\end{remark}

\begin{remark}
We write $Ext(\phi_1, \phi_2, \dots, \phi_n)$ for $\extension{\extension{\dots}{\phi_{n-1}}}{\phi_n}$.
\end{remark}

The following is shown in \cite{tilson1987categories}:
\begin{proposition}
    $\extension{\phi}{\phi'}$ is a well-defined relational morphism $M \relml S \wrc T$.
\end{proposition}

In the sequel we will write $\Z$ or $(\Z,\Z_+)$ for the typed monoid $(\Z,\Z_+,\pm1)$ and $\dyck_k$ for the syntactic monoid of the language $\dyck_k$.

\begin{thm}\label{thm:algebraic-algorithm-is-correct-wreath-products}
Given a finite monoid $M$, the procedure below correctly determines if $M$ divides an iterated typed wreath product of $(\Z,\Z_+)$.
\end{thm}

Recall the definition of $\mathcal{R}$ classes and the $\prec_{\mathcal{R}}$ order from \cref{def:R_equivalent}. We now state our decision procedure:

\begin{definition}[Decision Procedure]\label{def:decition-procedure}
Our input is a finite monoid $M$.
Let $R_1, R_2, \dots, R_\ell$ be the $\mathcal{R}$-classes of $M$, such that $R_i \prec_{\mathcal{R}} R_j \Rightarrow i > j$.
For $i=1 , 2, \dots, \ell$, we maintain relational morphisms $\phi_i$ from $M$ to an iterated wreath product of $\dyck_k$ and $\uone$.
We maintain the invariant that \[\{m\} = \phi_i^{-1}(\phi_i(m)), \forall m \in R_1, \dots, R_i\; \; \; \; \; \; \; \;  (\operatorname{Invariant}(i))\]
Let $N_0$ be the trivial monoid and $\phi_0\colon M\relml N_0$ be the trivial relational morphism.
For $i=1, \dots, \ell$, do
    \begin{enumerate}

\item %
\textcolor{darkblue}{Note: Here, we record, for any path through the category going into $R_i$, which arrow was used for passing into it. We will later refer to this annotation as the ``entry point coordinate''. It intuitively tells us via which element of $R_i$ we first entered it. Later, $E_{i,j}$ will have multiple strongly connected components (or \emph{bonded components}, the term used by Tilson) corresponding to $R_i$, indexed by the different possible entry points (Lemma \ref{lemma:structure-of-Eij}).}

Let $\delta : \obj(D_{\phi_{i-1}}) \rightarrow \{0,1\}$ be given as
\[\delta(o) = \begin{cases}
    1 & \exists m \in \phi_{i-1}^{-1}(o) : m \preceq_R R_i \\ 0 & else 
\end{cases}\]

Let $\partial_i = \{(o \rightarrow_{s,t} ot) : \delta(o) = 0, \delta(ot) = 1\}$.

Define $\phi' : D_{\phi_{i-1}} \relml B$, where $B$ is the left-zero semigroup with identity $*$ adjoined, with left-zeros indexed by the elements of $\partial_i$\footnote{$B := \partial_i \cup \{*\}$, $*b = b$, $rb=r$ whenever $r \in \partial_i$, $b \in B$.}, as the relational morphism generated by (i.e., the intersection of all relational morphisms satisfying) 
\begin{equation}
    \phi'(o \rightarrow_{(s,t)} ot) \supseteq \begin{cases}
        \partial_i & \delta(o) = \delta(ot) = 1 \\
        \{*\} & \delta(o) = \delta(ot) = 0 \\
        \{(o \rightarrow_{(s,t)} ot)\} & \delta(o) = 0, \delta(ot) = 1 \\
    \end{cases}
\end{equation}
where we note $\delta(o) = 1, \delta(ot) = 0$ cannot occur. 

By definition, $* \not\in \phi'(1 \rightarrow_{(s,o)} o)$ iff $s \preceq_R R_i$.

We note that $B$ is $\mathcal{R}$-trivial; hence it divides an iterated wreath product of $\uone$.

This is a relational morphism.

    Set $\phi_{i,1} := \extension{\phi_{i-1}}{\phi'} : M \relml N_{i,1}$ where $N_{i,1} = B \wrc N_{i-1}$.

    \item Build a sequence of relational morphisms $\phi_{i,2}, \phi_{i,3}, \dots$ as follows, over $j=1,2,3,\dots$:

    \begin{enumerate}
\item

\textcolor{darkblue}{Intuition: We simplify $D_{\phi_{i,j}}$ in two ways. First, all information from earlier annotation is removed once we have entered $R_i$. This simplifies our proof by avoiding any complications arising from interactions between old and new annotation. Second, elements of $\dyck_k$ provide extra information beyond bounded counting; here, we project its elements to the underlying set of truncated integers.}

  Write an element of \(N_{i,j}\) as
  \[
  a=(f_{i,j},\ldots,f_{i,1},f).
  \]
  where $f_{i,j} : N_{i,j-1} \rightarrow \dyck_k$.
  Define its ``retained signature'' by
  \[
  \operatorname{sig}_{i,j}(a)
  :=
  \big(
    \underbrace{0 \cdot \underbrace{f_{i,j}(1_{N_{i,j-1}})}_{\in \dyck_k}}_{\in S_k},
    \ldots,
    0 \cdot f_{i,2}(1_{N_{i,1}}),
    \underbrace{f_{i,1}(1_{N_{i-1}})}_{\in B}
  \big),
  \]
  where we view $\dyck_k$ as acting on $S_k := [-k, \dots, k] \cup \{\bot\}$ from the right (\ref{eq:action-dyck-counts}). 
We define a right congruence\footnote{In the sense that $a \rho b \Rightarrow ar \rho br$ for any $a,b,r \in N_{i,j}$.} \(\rho_{i,j}\) on \(N_{i,j}\) by setting
  \(
  a \,\rho_{i,j}\, b
  \)
  whenever (i) $a=b$, or (ii)
  \(
  f^a_{i,1}(1_{N_{i-1}})\neq *, f^b_{i,1}(1_{N_{i-1}})\neq *,
  \)
  and
  \(
  \operatorname{sig}_{i,j}(a)
  =
  \operatorname{sig}_{i,j}(b).
  \)
  We define $E_{i,j}$ as the category obtained from $D_{\phi_{i,j}}$ and the right-congruence $\rho_{i,j}$ (based on the discussion in Section \ref{sec:dyck-monoids-congruence}).

\item Consider the set $\Omega_{i,j}$ of algebraic relational morphisms $\phi' : E_{i,j} \relml \Z$ such that for each $s \in R_i$ and any $t \in \phi_{i,j}(s)$, $\phi'(1 \rightarrow_{s,t} t)$ is a finite set.
\footnote{Without loss of generality, we restrict to relational morphisms $\phi'$ where, for each object $o$ with non-$*$ entry annotation, every vlaue of $\phi'$ on an arrow $1 \rightarrow o$ is generated by factorization through a prefix ending inside $R_i$: $\phi'(\alpha) = \bigcup_{\alpha=\beta\gamma, \beta \in \pi^{-1}(R_i)} \phi'(\beta)\phi'(\gamma)$. This condition is imposed because relational morphisms may add values that are not generated by any nontrivial factorization of an arrow. The construction in Lemma \ref{lemma:extend-omega} satisfies this condition.}
\textcolor{darkblue}{Note: In the automaton-based proof, the key requirement is for the relabelings to be balanced. Here, we formulate this in terms of finite sets.}

\item %
If there is $\phi' \in \Omega_{i,j}$ such that \emph{at least one} of the following is satisfied:
\begin{enumerate}
    \item Condition A: $\phi' : E_{i,j} \relml \Z$ assigns disjoint images to two arrows in $Hom(o \rightarrow o')$ where $o, o'$ both are in $R_i$ (in the sense of Definition \ref{def:notions-start-end-in-Ri}).

  \item Condition B: 
  There exist objects
  \(x,y\in Obj(E_{i,j})\), arrows
  \(
    a\in Hom(x\rightarrow y); b\in Hom(y\rightarrow x); \ell\in Hom(x\rightarrow x)
  \)
  such that $x, y$ are inside $R_i$\footnote{In the sense of Definition \ref{def:notions-start-end-in-Ri}:
  \(
    Hom(1,x)\cap \pi^{-1}(R_i)\neq\emptyset;
    Hom(1,y)\cap \pi^{-1}(R_i)\neq\emptyset,
  \), where $\pi : E_{i,j} \relml M$ is the canonical division.},
  but no arrow in $Hom(x,y)$ ends in $R_i$\footnote{In the sense of Definition \ref{def:notions-start-end-in-Ri}: \(
    \forall \beta\in Hom(x\rightarrow y);
    Hom(1\rightarrow x)\,\beta\cap \pi^{-1}(R_i)=\emptyset,
  \), where $\pi : E_{i,j} \relml M$ is the canonical division.}, whereas $\ell$ ends in $R_i$\footnote{In the sense of Definition \ref{def:notions-start-end-in-Ri}:  
  \(
    Hom(1,x)\,\ell\cap \pi^{-1}(R_i)\neq\emptyset,
  \)}
  and
  \[
    \phi'(ab)\cap \phi'(\ell)=\emptyset .
  \]

\end{enumerate}
then we take such a $\phi'$, turn it to a relational morphism $\tilde{\phi'} : E_{i,j} \relml \dyck_k$\footnote{For sufficiently large $k$, greater than the absolute value of, for each $s \in R_i$ and any $t \in \phi_{i,j}(s)$, the entries of the finite set $\phi'(1 \rightarrow_{s,t} t)$. To convert a relational morphism to $\Z$ into one on $\dyck_k$, map 1 to a, -1 to b, etc., and close to make it a relational morphism. Choosing $k$ large enough will ensure $\bot$ is avoided on these homsets. Also, on these homsets, this resulting relational morphism is at least as finegrained as the original one, since one can get back the original one by mapping Dyck elements to their height.}, pull it back to $\widehat{\phi}' : D_{\phi_{i,j}} \relml \dyck_k$ and build $\phi_{i,j+1} := \extension{\phi_{i,j}}{\widehat{\phi}'} \colon M\relml N_{i,j+1}$, where $N_{i,j+1} := \dyck_k \wrc N_{i,j}$.

        \item Otherwise, we terminate this inner loop.

        \item 
        One possibility is that $\phi_{i,j}^{-1}(\phi_{i,j}(m)) = \{m\}$ for each $m \in R_i$. By induction, we have satisfied $\operatorname{Invariant}(i)$. Then, set $\phi_i := \phi_{i,j}\colon M\relml N_i$ and pass to the next iteration in the outer loop. 
        Else, we exit from the entire algorithm and declare \emph{failure}.
        \end{enumerate}

    \end{enumerate}
    We declare \emph{success} and return $N_\ell$ if we iterated through all $i=1, \dots, \ell$ without ever declaring failure.
\end{definition}
\textcolor{darkblue}{Note: Motivation of this algorithm: The generic strategy would be to keep generating relational morphisms to $\Z$, but that would not suffice for decidability, because (i) we wouldn't know how often to iterate, and (ii) wouldn't know which types to use. We ``guide'' the process by (i) proceeding along $\mathcal{R}$ classes, (ii) considering morphisms to well-selected finite monoids, (iii) passing to $E_{i,j}$ instead of the more complicated $D_{\phi_{i,j}}$, and (iv) focusing on the sets $\Omega_{i,j}$.}

Throughout, we write $1 \in Obj(E_{i,j})$ for the  object arising from the identity element of $N_{i,j}$.

\begin{defin}\label{def:notions-start-end-in-Ri}
For any of the categories $E_{i,j}$ constructed in the decision procedure, we define the following notions. Let $\pi : E_{i,j} \prec M$ be the canonical division. 

\begin{enumerate}
\item Let $S_i$ be the set of arrows $o \rightarrow_{s,t} ot$ such that there are arrows $1 \rightarrow_{s',o} o$ and $ot \rightarrow_{s'',t''} ott''$ such that $ss's'' \in R_i$.%

\item We say an arrow $\alpha \in Hom(o \rightarrow o')$ ``ends'' in $R_i$ if there is an arrow $\beta \in Hom(1 \rightarrow o)$ such that $\beta\alpha$ is defined and $\beta\alpha \in \pi^{-1}(R_i)$.

\item We say an arrow $\alpha \in Hom(o \rightarrow o')$ ``starts'' in $R_i$ if there is an arrow $\beta \in Hom(1 \rightarrow o)$ such that $\beta\alpha$ is defined and $\beta \in \pi^{-1}(R_i)$.

\item We say an object $o$ is ``inside $R_i$'' if $Hom(1 \rightarrow o) \cap \pi^{-1}(R_i) \neq \emptyset$. 
\end{enumerate}

\end{defin}

\subsubsection{Key Properties of the Procedure}

\begin{lemma}[Correctness]\label{lemma:algebraic-correctness}%
    If the algorithm succeeds, then $M$ divides an iterated wreath product of $(\Z, \Z_+)$, of the form $(((\dots \Z) \wrt \Z) \wrt \Z)$.
\end{lemma}

\begin{proof}

We refer to a wreath product of the form $(\dots \wrt \cdot) \wrt \cdot) \wrt \cdot) \wrt \cdot$ as \emph{left-associative}.
First, if the algorithm succeeds we obtain a relational morphism $\phi_k : M\relml N$ where $N$ divides a left-associative wreath product of $\uone$ and $\dyck_k$.
By the invariant maintained in the Decision Procedure, $\phi_k$ is injective, making it a division.
We now need to explain why $M$ also divides a product of the form $(((\dots \Z) \wrt \Z) \wrt \Z)$.
First, by we can obtain a representation in terms of $\dyck_k \preceq  ((\uone)^{2k} \wrt \Z)$  and $\uone$, in the left-associative bracketing due to the associativity of the finite wreath product at the level of pseudovarieties \citep{tilson1987categories}.
By repeated application of \cref{lem:UZ_associative}, we then obtain a left-associative wreath product where the factors are either direct products of $\uone$ or just $\Z$. 
Because direct products divide wreath products, we know that $M\preceq (S\wrt (N\times \uone))\wrt T$ implies $M\preceq ((S\wrt N)\wrt \uone)\wrt T$.
Iteratively applying this identity, we obtain a division into a left-associative wreath product of $\uone$ and $\Z$ (with much greater depth). 
Finally because $\uone\preceq\Z$, we can conclude that $M$ divides a left-associative wreath product of $\Z$.

\end{proof}

\begin{lemma}\label{lem:UZ_associative}%
    For typed monoids $T,S$, we have that $ T\wrt (\uone\wrt S)\preceq(T\wrt \uone^{|\types{S}|})\wrt S$
\end{lemma}
\begin{proof}
    We show the proof ignoring any constants in the wreath product to put less strain on notation, but it is easy to add them back in.
    Let $(f_1,(g_1,h_1))(f_2,(g_2,h_2))\cdots (f_n,(g_n,h_n))$ be a multiplication in $T\wrt (\uone\wrt S)$ where $h_i\in S$, $g_i\in\uone^S$, and $f_i\in T^{\uone^S\times S}$.
    We will write multiplication in all monoids using $\Pi$, but assume the distinction is clear. By \cref{lem:wreath_type_computation} we can compute the type of this multiplication by computing the following and checking the types of the indicated terms $\type{1}$, $\type{2}$, and $\type{3}$ 
    \begin{align*}
        \left(\underbrace{\prod_{i=1}^{n-1} f_i\left(\left(\prod_{j=1}^{i-1}g_j\left(\prod_{k=1}^{j-1} h_k\right),\prod_{j=1}^{i-1} h_j\right)\right)}_{\type{1}},\left(\underbrace{\prod_{i=1}^{n-1}g_i\left(\prod_{j=1}^{i-1} h_j\right)}_{\type{2}},\underbrace{\prod_{i=1}^{n-1} h_i}_{\type{3}}\right)\right)
    \end{align*}

    We will simulate this computation in $(T\wrt \uone^{|\types{S}|})\wrt S$.
    That is, the desired computation is given by the following injective relational morphism given below. 
    Intuitively, $\uone^{|\types{S}|}$ records which of the finitely many \emph{types} have been seen.
    Since $f$ must be type-respecting, this suffices to recreate the computation of $T\wrt(\uone\wrt S)$.
    \begin{align*}
        \phi\colon T\wrt (\uone\wrt S) &\to (T\wrt \uone^{|\types{S}|})\wrt S\\
        \alpha_{f,g} &\in (T^{\uone^{|\types{S}|}}\times \uone^{|\types{S}|})^S\\
        \gamma_{f,\zeta} &\in T^{\uone^{|\types{S}|}}\\
        \phi((f,(g,h)))&= (\alpha_{f,g},h)\\
        \alpha_{f,g}\left(\zeta\right)&=(\gamma_{f,\zeta},(\lambda_{\type{S}_1},\lambda_{\type{S}_2},\ldots, \lambda_{\type{S}_{|\types{S}|}})) & \lambda_\type{S}=\begin{cases}
            1 & \zeta\in\type{S}\\
            0 & \zeta\not\in\type{S}
        \end{cases}\\
        \gamma_{f,\zeta}((\lambda_{\type{S}_1},\lambda_{\type{S}_2},\ldots, \lambda_{\type{S}_{|\types{S}|}}))&=f\left(\beta,\zeta\right) & \Im(\beta)\subseteq\bigcap_{\type{S}\in\types{S};\lambda_\type{S}=1} \type{S}
    \end{align*}
    Note that $\gamma_{f,\zeta}$ is well-defined because $f((\beta,\zeta))$ identical for all $\beta$ whose images are contained in the same set of types, since $f$ is type-respecting.
    Now, this mapping is injective because every type-respecting function $g\in \uone^S$ is represented by some $(\lambda_{\type{S}_1},\lambda_{\type{S}_2},\ldots, \lambda_{\type{S}_{|\types{S}|}})) \in \uone^{|\types{S}|}$, and every type-respecting function $f$ is represented as some $\gamma_{f,\zeta}$.
    
    Thus the sequence $(f_1,(g_1,h_1))(f_2,(g_2,h_2))\cdots (f_n,(g_n,h_n))$ maps homomorphically to the sequence $(\alpha_{f_1,g_1},h_1)(\alpha_{f_2,g_2},h_2)\cdots (\alpha_{f_n,g_n},h_n)$ which is evaluated as 
    $$\left(\prod_{i=1}^{n-1}\alpha_{f_i,g_i}\left(\prod_{j=1}^{i-1}h_j\right),\underbrace{\prod_{i=1}^{n-1}h_i}_{\type{3}}\right)$$
    Observe that the type $\type{3}$ is already computed in the second coordinate. We will see that $\type{1}$ and $\type{2}$ are also computed in the first coordinate:
    \begin{align*}
        \prod_{i=1}^{n-1}\alpha_{f_i,g_i}\left(\prod_{j=1}^{i-1}h_j\right)&=\alpha_{f_1,g_1}\left(0\right)\alpha_{f_2,g_2}\left(h_1\right)\cdots \alpha_{f_n,g_n}\left(\prod_{j=1}^{n-1}h_j\right)\\
        &=(\gamma_{f_1,0},\Lambda_1)(\gamma_{f_2,h_1},\Lambda_2)\cdots (\gamma_{f_i,\prod_{i=1}^{n-1}h_j},\Lambda_n)\\
        &=\left(\underbrace{\left(\prod_{i=1}^{n-1}f_i\left(1,\prod_{j=1}^{i-1} h_j\right)\right)}_{\type{1}}, \underbrace{\prod_{i=1}^{n-1}g_i\left(\prod_{j=1}^{i-1} h_j\right)}_{\type{2}}\right)\\
    \end{align*}
\end{proof}

The following properties of $\omega \in \Omega_{i,j}$ are key:
\begin{lemma}[Balance on Loops]\label{lemma:looping-zero} 
Let $\omega \in \Omega_{i,j}$. %
\begin{enumerate}
    \item Let $u \in R_i$, $\rho \in M$ such that $u\rho = u$. 
    Consider arrows \[1 \rightarrow_{(u,o)} o \rightarrow_{(\rho, t)} ot\] in $E_{i,j}$. 
    Then in fact $ot=o$, and  $\omega(o \rightarrow_{(\rho, t)} o) = \{0_{\Z}\}$ where $0_{\Z}$ is the neutral element of $\Z$. %

    \item Let $u$ be such that $R_iu = R_i$. 
    
For $v \in R_i$,    consider arrows
    \[1 \rightarrow_{(v,o)} o \rightarrow_{(u,t)} ot\]
    in $E_{i,j}$.
    Then $\omega(o \rightarrow_{(u,t)} ot)$ is a singleton.
\end{enumerate}

\end{lemma}

\begin{proof}
For the first point, we have two claims to prove: that (a) $ot=o$, and that (b) the image is $\{0_{\Z}\}$. We first note that (a) holds at $j=1$: $u \in R_i$, $u\rho = u$ and 
$1 \rightarrow_{(u,o)} o \rightarrow_{(\rho, t)} ot$ in $E_{i,j}$, then in fact $ot = o$ by construction of $E_{i,1}$.
We also note that, if (b) has been shown for $1, \dots, j$, then (a) follows for $1, \dots, j+1$.
We thus need to perform the inductive step for (b). We consider the arrow $o \rightarrow_{(\rho, t)} o$. Assume $z \in \omega(o \rightarrow_{(\rho, t)} o)$ and $z \neq 0$.
    Now, for all $k \geq 0$, $u\rho^k = u$; hence, $\omega(1 \rightarrow_{(u,o)} o)$ is infinite. This is a contradiction to $\omega \in \Omega_{i,j}$.

    For the second point, from $R_iu = R_i$, obtain $\rho$ such that $vu\rho = v$. We consider arrows:
    \[1 \rightarrow_{(v,o)} o \rightarrow_{(u,t)} ot \rightarrow_{(\rho, t')} ott'\]
    By (1), in fact, $ott' = o$ and $\omega(o \rightarrow_{(u,t)} ot \rightarrow_{(\rho, t')} ott') = \{0_{\Z}\}$.
    Because $\omega$ is a relational morphism and $\Z$ is a group, this entails $\omega(o \rightarrow_{(u,t)} ot)$ must be a singleton.
\end{proof}

We deduce the following structural properties of $E_{i,j}$, which are used both for termination and for completeness of the procedure. As foreshadowed when we defined the ``entry-point coordinate'' $\phi_{i,1}$, we find that $R_i$ is reflected in $E_{i,j}$ in multiple strongly connected components (or ``bonded components''), each indexed by a different entry arrow of the ``boundary set'' $\partial_i$ we defined there:
\begin{lemma}[Structure of $E_{i,j}$]\label{lemma:structure-of-Eij} %
\quad
\begin{enumerate}
\item Let $\alpha := (o \rightarrow_{(s,t)} ot) \in \partial_i$, and let $m \in R_i$.
Then there is exactly one $u \in Obj(E_{i,j})$ such that both of the following hold: (i) there is an arrow $1 \rightarrow_{(m,u)} u$, and (ii) the $\phi_{i,1}$ component of $u$ is $\alpha$.

\item Let $\alpha = (1 \rightarrow_{(m,u)} u), \alpha' = (1 \rightarrow_{(m',u')} u')$ be arrows in $E_{i,j}$, where $m, m' \in R_i$; and assume $u, u'$ have the same $\phi_{i,1}$ component. Then there is $\beta \in Hom(u,u')$ such that $\alpha' = \alpha\beta$.

\end{enumerate}
\end{lemma}

\begin{proof}
    We show this by induction over $j$. The claims are immediate at $j=1$.
    For the first claim, we show the inductive step using Lemma \ref{lemma:looping-zero}.
For any arrow \[\alpha = (1 \rightarrow_{(m, o)} o)\] with $m \in R_i$, Lemma \ref{lemma:looping-zero}.2 entails that $\omega$ is single-valued on all arrows $\beta = (o \rightarrow_{(m', t)} ot)$ with $mm' \in R_i$; hence, inductively, after passing to $E_{i,j+1}$, only a single arrow is derived from $\beta$.

The second claim follows by choosing $n$ such that $m' = mn$; the first claim enforces the presence of an arrow covering $n$ in $Hom(u,u')$.
\end{proof}

\begin{lemma}[Termination]
    The inner loop always terminates. 
\end{lemma}

\begin{proof}

A new morphism $\phi_{i,j}$ can only be accepted because it separates two arrows in a homset $Hom(o \rightarrow o')$ where $o, o'$ are both inside $R_i$.
We need to show that this can only happen a bounded number of times.

To this end, we first note by the first point of Structure of $E_{i,j}$ that, for each entry arrow $\alpha$, we can uniquely assign a single $o$ for every $r \in R_i$, and this mapping is a surjection onto the objects inside $R_i$.

Also, with every $j$ iteration, this structure becomes finer, i.e., a homset $Hom(o \rightarrow o')$ where $o, o'$ are in $R_i$ can split into several ones that (save for arrows leaving $R_i$, where arrows carrying the same $M$-label might go into into different descendant homsets) each exactly carry some subset of the original arrows.
Thus, the total number of times this happens must be bounded.
\end{proof}

\begin{lemma}[Computability]
In each step of the inner loop, we can effectively decide if there is a $\phi'$ satisfying the requirements.
\end{lemma}

\begin{proof}
We can always construct a relational morphism of the second type if it exists. We can always apply this case when it is available.
We need to check when a relational morphism of the first type exists.
For the arrows that start and end in $R_i$, since they will be single-valued, we can check the $\Z$-module of functions defined on these arrows satisfying the balancedness condition.
For arrows in the connected components of objects in $R_i$ that however do not end $R_i$, all of them will receive a nonempty set of values by closing this under composition, for, otherwise, we would have been able to choose a relational morphism of the second type.
We can extend any solution from a single connected component to a full morphism, by the Extension Lemma (Lemma \ref{lemma:extend-omega}).

\end{proof}

\subsubsection{Establishing Completeness}
Assume that, after constructing $E_{i,j}$, the procedure terminates with failure.  We localize this failure into a smaller, focused category.

\begin{defin}
    For each object $o$ inside $R_i$, we consider the strongly connected component (or ``bonded component'' in \cite{tilson1987categories}) of $o$: %
\begin{align*}
    Obj(C_o) := &\{o' \in Obj(E_{i,j}) : Hom(o \rightarrow o') \neq \emptyset, Hom(o'\rightarrow o) \neq \emptyset   \} %
\end{align*}
As above, let $\pi$ be the canonical division $\pi\colon E_{i,j}\prec M$.
Within any homset $Hom(o' \rightarrow o'')$ in $C_o$, we merge any two arrows $\beta, \beta'$ such that, in $E_{i,j}$, $Hom(1 \rightarrow o')\beta$ and $Hom(1 \rightarrow o')\beta'$ both are not in $\pi^{-1}(\bigcup_{j \leq i} R_j)$.
By construction, $C_o$ is a well-defined category, and $C_o \prec E_{i,j} \prec M$.

\end{defin}

\begin{lemma}
Every homset $Hom(u \rightarrow u')$ in $C_o$ hosts an arrow ending in $R_i$.
\end{lemma}

\begin{proof}
    Otherwise, a morphism $\phi'$ satisfying the second condition would have been selected.
\end{proof}

\begin{lemma}
    Failure of the algorithm entails that some $C_o$ contains a non-singleton homset, i.e., $C_o$ is not trivial.
    \end{lemma}

    \begin{proof}
Failure of the algorithm implies that there is $m \in R_i$ such that $\phi_{i,j}^{-1}(\phi_{i,j}(m)) \supsetneq \{m\}$.

Now consider $m' \in M$, $m \neq m'$ such that there is some $\hat{o} \in \phi_{i,j}(m) \cap \phi_{i,j}(m')$.

 Let $o$ be the corresponding element of $E_{i,j}$.
 
This means that, in $E_{i,j}$, there are distinct arrows $1 \rightarrow_{(m,\tilde{o})} o$ and $1 \rightarrow_{(m',o)} o$.

Now we consider the $\phi_{i,1}$ coordinate of $o$ (``entry point coordinate''); this is of the form $(u \rightarrow_{(s,t)} v)$ corresponding to some arrow of $D_{\phi_{i-1}}$.

By the inductive hypothesis ``(Invariant(i))'',  $Hom(1 \rightarrow u)$ in that category had exactly one object $1 \rightarrow_{(s',u)} u$, where $s' \in R_k$ for some $k<i$.

We can thus factorize
\begin{equation}
    1 \rightarrow_{(m,\tilde{o})} o = \left(1 \rightarrow_{(s',u)} u\right) \left(u \rightarrow_{(s,t)} v\right) \left(v \rightarrow_{(w,\dots)} o\right)
\end{equation}

\begin{equation}
    1 \rightarrow_{(m',\tilde{o})} o = \left(1 \rightarrow_{(s',u)} u\right) \left(u \rightarrow_{(s,t)} v\right) \left(v \rightarrow_{(w',\dots)} o\right)
\end{equation}

 Now let $v$ be the corresponding entry object, which must be contained in $C_o$ (by Lemma \ref{lemma:structure-of-Eij}.1); then  (by Lemma \ref{lemma:structure-of-Eij}.2) there is an arrow $\alpha = 1 \rightarrow_{(s's,v)} v$ such that there are $\beta, \beta' \in Hom(v, o)$ with $1 \rightarrow_{(m,\tilde{o})} o = \alpha\beta$ and $1 \rightarrow_{(m',o)} o = \alpha\beta'$.
We thus have shown $Hom(v, o)$ to be nonsingleton.
\end{proof}

\begin{lemma}[Extension Lemma]\label{lemma:extend-omega}
Assume $\omega : C_o \relml \Z$ stays finite-image on all arrows in $C_o$ that end in $R_i$.
Then there is an algebraic relational morphism $\tilde{\omega} : E_{i,j} \relml \Z$ such that  $\tilde{\omega}(1 \rightarrow_{(s,t)} t)$ is a singleton set whenever $s \in R_i$, and $\tilde{\omega}|_{C_o} \equiv \omega$. %
\end{lemma}

\begin{proof}
There are two aspects here: achieving $\tilde{\omega}|_{C_o}\equiv \omega$ -- this is easy because $C_o$ is a bonded component -- and achieving that $\tilde{\omega}(1 \rightarrow_{(s,t)} s)$ is a finite set for $s \in R_i$ -- which requires care.

We first define
\begin{equation}
    \omega_*(\alpha) := \begin{cases} \omega(\alpha) & \alpha \in C_o \\ \{0\} & \text{else} \end{cases}
\end{equation}
which is not in general a relational morphism.
We define $\tilde{\omega}$ as the closure under composition, so that $\tilde{\omega}$ is a relational morphism.
Constructively, we can write
\begin{equation}
    \tilde{\omega}(\alpha) = \bigcup_{\gamma_1 \dots \gamma_n =\alpha} \sum_{i=1}^n \omega_*(\gamma_i)
\end{equation}
where ``$\sum$'' is to be understood in a set-valued sense as $A + B = \{a+b : a \in A, b\in B\}$.

Because $C_o$ is a strongly connected component, $\tilde{\omega}|_{C_o} \equiv \omega$.
Now consider $\alpha : 1 \rightarrow_{(s,t)} t$ for some $s \in R_i$.
One option is that $t \not\in C_o$; in this case, no path through $C_o$ can multiply out to this arrow, and $\tilde{\omega}(\alpha) = \{0\}$, which is a finite set.
The other option is that $t \in C_o$.

Assume $\omega(\alpha)$ is infinite. Then, for each $n$, there is a path of the form:
\begin{equation}
    \alpha = \gamma_1 \gamma_2 \gamma_3 \dots \gamma_{n}
\end{equation}
where $\sum_{i=2}^n \omega_*(\gamma_i)$ contains unboundedly (positive or negative) large numbers; we can WLOG choose this so that $\omega_*(\gamma_2) \neq \{0\}$ by multiplying the initial prefix out into $\gamma_1$ if needed.
By construction, $\gamma_2 \in C_o$; hence, $\gamma_1$ must end in $C_o$. But then we can set $\beta := \gamma_2 \cdot \dots \cdot \gamma_{n}$, which starts and ends in $C_o$, and also ends in $R_i$ because $\alpha$ does. 
But then $\tilde{\omega}(\beta(n))$ has unboundedly large numbers as $n\rightarrow \infty$, which is a contradiction to $\tilde{\omega}|_{C_o} \equiv \omega$.

\end{proof}

\begin{lemma}\label{lemma:injectivity-anchor-from-base}\label{lemma:injectivity-anchor}
Consider two arrows in the same homset in $C_o$:
$\beta_1 := u \rightarrow_{(s,t)} ut$ and $\beta_2 := u \rightarrow_{(s',t')} ut$.
If $\beta_1 \neq \beta_2$, then 

\begin{enumerate}
    \item there is $\lambda \in R_i \cap \pi(Hom(1 \rightarrow u))$ such that $\lambda s \neq \lambda s'$.
    \item there is $\alpha \in Hom(o \rightarrow u)$ such that $\alpha\beta_1 \neq \alpha\beta_2$.
\end{enumerate}

\end{lemma}

\begin{proof}
This follows from the derived category construction; the right-congruences applied afterwards preserve it.
\end{proof}

\begin{lemma}[Completeness]\label{lemma:completeness}
    If the procedure terminates with failure, $M$ does not divide an iterated wreath product of $(\Z,\Z_+)$.
\end{lemma}

\begin{proof}
Assume that, after constructing $E_{i,j}$, the procedure terminates with failure.

\paragraph{Setting up a division}
Choose as $C''$ one nontrivial $C_o$ (it doesn't matter which one); from now on this fixes $o$.
For the purposes of defining typed division, we view $C_o$ as rooted in $o$ (Definition \ref{def:typed-division}).
Assume $M$ divides an iterated wreath product of $(\Z,\Z_+)$, then so does $C''$ (in the sense of Definition \ref{def:typed-division}) by Lemma \ref{lem:typed-division-compose}.
Consider a division of a minimum-depth wreath product of $\Z$, $\psi : C'' \prec \underbrace{\Z \wrt \dots \wrt \Z}_{T\text{ times}}$ (here all $\Z$ are typed like $(\Z,\Z_+)$); in particular, because $C''$ is not trivial, $T \geq 1$.
Let $\omega = \pi^{(T)}\circ \psi$; by definition, $\omega : C'' \relml \Z$.
Our goal is to create a division $\mu : C'' \prec \underbrace{\Z \wrt \dots \wrt \Z}_{T-1\text{ times}}$   (again in the sense of Definition \ref{def:typed-division}).
This would be a contradiction to the minimality of $T$.
As a consequence, $C''$ cannot divide \emph{any} wreath product of $(\Z,\Z_+)$.

\textcolor{darkblue}{Note: The basic idea is to show that $\omega$ cannot have provided any useful information, and in particular gives the same values across different arrows in a given hom-set. Either it stays finite-image on all arrows going into $R_i$ (in which case $\omega$ provides as much information as an element of $\Omega_{i,j}$ -- and the algorithm having terminated indicates that no such element would have been helpful), or it is infinite-image (in which case we can pass to a limiting element $\pm \infty$, showing that $\omega$ contributes no information useful to $(\Z,\Z_+)$-based recognizers).}

\paragraph{Case 1: $\omega$ stays finite-image on all arrows ending in $R_i$ inside $C''$}
First, assume $\omega$ stays finite-image on all arrows ending in $R_i$ inside $C''$. 
We obtain $\tilde{\omega} : E_{i,j} \relml \Z$ via Lemma \ref{lemma:extend-omega}.
We observe that $\tilde{\omega} \in \Omega_{i,j}$; hence, $\tilde\omega$ must assign singleton images to any arrow in $C_o$ that ends inside $R_i$.
The algorithm having terminated means that $\tilde{\omega}$, and hence $\omega$, cannot assign distinct images to any two arrows staying inside $R_i$ that appear in a single homset in $C_o$; also, any of these needs to be mapped to a singleton set. It also means that, when considering an arrow that is in $C_o$ but ends outside of $R_i$, it might in fact be assigned some further set of values, but it must overlap with the value assigned to the arrows staying inside $R_i$. That is, in any homset, there is one single value shared across all arrows; also, there potentially is a further set of values assigned to the arrows that end outside of $R_i$.

For each homset $Hom(u \rightarrow u')$, we obtain a unique $\theta_{u,u'} \in \Z$ such that $\omega(\alpha) = \{\theta_{u,u'}\}$ whenever $\alpha \in Hom(u \rightarrow u')$ ends in $R_i$; also, when the homset contains an arrow leaving $R_i$ (we can write it as $\beta : u \rightarrow_{(\bot, \dots)} u'$), $\theta_{u,u'} \in \omega(\beta)$.
We write $\nu_u := \theta_{o,u}$; importantly, $\nu_{ut} = \nu_u \theta_{u,ut}$.

We obtain the desired division $\mu : C_o \prec \underbrace{(\Z \wrt \dots \wrt \Z)}_{T-1 \text{ times}}$ as follows.
Note that any object in the image of $\psi$ has the form $(g, \dots)$ where $g : \Z \rightarrow \underbrace{(\Z \wrt \dots \wrt \Z)}_{T-1 \text{ times}}$.
For any arrow $u \rightarrow_{(s,t)} ut$ in $C_o$, we define:
\begin{equation}
    \mu(u \rightarrow_{(s,t)} ut) := \{g(\nu_u) : (g,\theta_{u,ut}) \in \psi(u \rightarrow_{(s,t)} ut)\} %
\end{equation}
which is a nonempty subset of $\underbrace{(\Z \wrt \dots \wrt \Z)}_{T-1 \text{ times}}$.

We first show that $\mu$ is an algebraic relational morphism.
Consider arrows $\alpha : u \rightarrow_{(s,t)} ut$ and $\beta : ut \rightarrow_{(s', t')} utt'$ in $C_o$. 
Then:
\begin{align*}
    \mu(\alpha)\mu(\beta) = & \{g(\nu_u) : (g,\theta_{u,ut}) \in \psi(\alpha)\} \{g'(\nu_{ut}) : (g',\theta_{ut,utt'}) \in \psi(\beta)\} \\
    = & \{g(\nu_u)g'(\nu_{ut}) : (g,\theta_{u,ut}) \in \psi(\alpha), (g',\theta_{ut,utt'}) \in \psi(\beta)\} \\
    \subseteq & \{g(\nu_u) : (g,\theta_{u,utt'}) \in \psi(\alpha\beta)\} \\
    = & \mu(\alpha\beta)
\end{align*}
where the ``$\subseteq$'' step used the fact that $\psi$ is an algebraic relational morphism.

We, second, show that $\mu$ is a typed division.
Consider two arrows $\beta_1 := u \rightarrow_{s,t} ut$, $\beta_2 := u \rightarrow_{s',t} ut$ in $C''$; if they are distinct, we find $\alpha \in Hom(o \rightarrow u)$ such that $\alpha\beta_1 \neq \alpha\beta_2$ (Lemma \ref{lemma:injectivity-anchor}). Then there are disjoint types $T_1, T_2$ of $\underbrace{\Z \wrt \dots \wrt \Z}_{T\text{ times}}$ such that $\psi(\alpha\beta_1) \subseteq T_1$, $\psi(\alpha\beta_2) \subseteq T_2$.
That is, for any $g_s, g_{s'}$ selected for the two arrows, we have
\begin{equation}
    (g_\alpha(\cdot) g_s(\cdot \theta_{o,u}), \theta_{o,u}\theta_{o,ut}) \in \psi(\alpha\beta_1) \subseteq T_1
\end{equation}
\begin{equation}
    (g_\alpha(\cdot) g_{s'}(\cdot \theta_{o,u}), \theta_{o,u}\theta_{o,ut}) \in \psi(\alpha\beta_2) \subseteq T_2
\end{equation}
Hence, there must be disjoint types $V_1, V_2$ of $\underbrace{\Z \wrt \dots \wrt \Z}_{T-1\text{ times}}$ such that
\begin{equation}
    \mu(\alpha\beta_1) = g_\alpha(0) g_s(\theta_{o,u}) \in V_1
\end{equation}
\begin{equation}
    \mu(\alpha\beta_2) = g_\alpha(0) g_{s'}(\theta_{o,u}) \in V_2
\end{equation}
But then
\begin{equation}
    g_s(\theta_{o,u}) \in g_\alpha(0)^{-1}V_1
\end{equation}
\begin{equation}
    g_{s'}(\theta_{o,u}) \in g_\alpha(0)^{-1}V_2
\end{equation}
where $V_1 \cap V_2 = \emptyset$, the types $g_\alpha(0)^{-1}V_1$ and $g_\alpha(0)^{-1}V_2$ are also disjoint types.\footnote{Here, we take a more general definition of types that are closed under (left) quotients, in line with the topological perspective on typed monoids \citep{gehrke2017stone}. This does not change the underlying expressivity of $\wwpc(\Z)$.}
The above assumes that $T_1, T_2$ are defined simply by first-coordinate evaluation at zero; if they instead arise as Boolean combinations of multiple observations, $\mu$ would instead map into a direct product of multiple copies where we evaluate $g$ at different coordinates; in this case, the proof here would be lifted to depth-$T$ wreath products of direct products of $\Z$.

\paragraph{Case 2: $\omega$ is infinite-image on some arrow ending in $R_i$ inside $C''$}
Let $\alpha = \left(u_0 \rightarrow_{(s_0,t_0)} u_0t_0\right)$ in $C''$ be an arrow ending in $R_i$ which is associated with an infinite number of different values under $\omega$.
Without loss of generality, we may assume that $\sup\omega(\alpha) = +\infty$ (else, $\inf\omega(\alpha) = -\infty$ and we replace $+\infty$ by $-\infty$ below).

Note that any object in the image of $\psi$ has the form $(g, \dots)$ where $g : \Z \rightarrow \underbrace{(\Z \wrt \dots \wrt \Z)}_{T-1 \text{ times}}$.
Because $g$ arises from a composition of type-respecting functions, $g(\infty) := \lim_{x \rightarrow \infty} g(x)$ is well-defined and in $\Z$.
We define
\begin{equation}
    \mu(u \rightarrow_{(s,t)} ut) := \{g(\infty) : (g,\dots) \in \psi(u \rightarrow_{(s,t)} ut)\} %
\end{equation}
which is a subset of $\underbrace{(\Z \wrt \dots \wrt \Z)}_{T-1 \text{ times}}$.

This is an algebraic relational morphism $\mu : C'' \relml \underbrace{(\Z \wrt \dots \wrt \Z)}_{T-1 \text{ times}}$:  For arrows $\beta, \beta'$ such that $\beta\beta'$ is defined, we have:
\begin{align*}
    \mu(\beta)\mu(\beta') =& \{g(\infty)h(\infty) : (g,\dots)\in\psi(\beta),\ (h,\dots)\in\psi(\beta')\} \\ %
    = & \{\lim_{x \rightarrow \infty} (g(x) h(x)) : g..., h...\}  \\
    = & \{\lim_{x \rightarrow \infty} (g(x) h(x \tau_1 )) : g..., h...\}  \\
    \subseteq & \{G(\infty) : (G, \dots) \in \psi(\beta\beta')\} \\ %
    = & \mu(\beta\beta')
\end{align*}

We need to show that it is typed and injective; the proof is similar to Case 1.
Consider two distinct arrows in the same homset in $C''$:
\begin{align*}
    \beta_1 := u \rightarrow_{(s,t)} ut \; \; \; \; \; \; \; \; \; \beta_2 := u \rightarrow_{(s',t')} ut
\end{align*}
If $\beta_1 \neq \beta_2$, that means there is $\lambda \in R_i \cap \pi^{-1}(Hom(1 \rightarrow u))$ such that $\lambda s \neq \lambda s'$, by Lemma \ref{lemma:injectivity-anchor-from-base}.

We now construct $\rho_1, \rho_2 \in M$ such that
\begin{equation}
    \alpha := 1 \rightarrow_{(\rho_1, u_0)} u_0 \rightarrow_{(s_0,t_0)} u_0t_0 \rightarrow_{(\rho_2, t')} u = 1 \rightarrow_{(\rho_1s_0\rho_2, u)} u
\end{equation}
in $Hom(1 \rightarrow u)$ satisfies $\lambda \in \pi^{-1}(\alpha)$.\footnote{Choose $\rho_1 \in R_i$ just to produce the first arrow. Now we choose $\rho_2 \in M$ such that $\rho_1s_0\rho_2 = \lambda$, using that $\rho_1s_0\rho_2$ and $\lambda$ both are in $R_i$.
Now applying Lemma \ref{lemma:looping-zero} (Claim 2) inductively to each $j$, there is an arrow $u_ot_0 \rightarrow_{(\rho_2, t')} u$.
}
Hence, $\alpha\beta \neq \alpha\beta'$ in $E_{i,j}$.
In $C''$, $\alpha' := u_0 \rightarrow_{(s_0,t_0)} u_0t_0 \rightarrow_{(\rho_2, t')} u$ satisfies $\alpha'\beta \neq \alpha'\beta'$. Now, we note that $\sup \omega(\alpha') = +\infty$; this enforces that $\mu(\beta) \cap \mu(\beta') = \emptyset$.

Overall, we have obtained a division $\mu$ in either case, and, by contradiction, established that $M$ does not divide an iterated wreath product of $(\Z, \Z_+)$.
\end{proof}

\subsubsection{Deriving Main Results}\label{sec:c-rasp-wpc-dyck-main-result}

\CraspRegPolyTime*
\begin{proof}

We need to show that (i) the algorithm from Definition \ref{def:decition-procedure} correctly determines membership in $\CRASP$ (shown in Theorem \ref{thm:algebraic-algorithm-is-correct-wreath-products}), and (ii) that it runs in polynomial time in $|M|$.
For (ii), one route is via the automaton-based formulation of the algorithm in Appendix \ref{app:automata-based-proof}, with proof of time polynomial in the number of automaton states (hence the size of the syntactic monoid)
 in Corollary \ref{lem:separable-poly}.
 Another route is by noting that one can identify candidate morphisms $\phi' : E_{i,j} \relml \Z$ by restricting to the strongly-connected component $C_o$, encoding (i) respecting composition, (ii) zeros on idempotent arrows as linear constraints. The number of such linear constraints is polynomial in the size of $M$; an integral basis of the solution space can then be found in polynomial time. A solution can be extended to an element of $\Omega_{i,j}$ via the Extension Lemma \ref{lemma:extend-omega}.
\end{proof}

\begin{corollary}\label{thm:dyck-vs-crasp}
    $\CRASP\cap\vty{REG}=\wwpc(\vty{Dy})$.
\end{corollary}

\begin{proof}
    The ``$\supseteq$'' direction is implied by the proof of Lemma \ref{lemma:algebraic-correctness}.
    The ``$\subseteq$'' direction follows because  the algorithm from Definition \ref{def:decition-procedure}, when it succeeds on a monoid $M$, supplies a division $M \prec \mathcal{D}_{k_1} \wrc \dots \wrc \mathcal{D}_{k_r}$.
\end{proof}

\subsection{Defining Division between Categories and Typed Monoids}\label{sec:typed-division}

Here, we define relational morphisms and divisions between finite categories and typed monoids. The definitions are naturally typed extensions of the usual definitions for finite categories \citep{tilson1987categories}.

\begin{defin}\label{def:algebraic-relm}
An \emph{algebraic relational morphism} $\phi : X \relml M$ from a category \(X\) to a monoid \(M\) specifies, for each arrow $\alpha$ in $X$ a nonempty set $\phi(\alpha) \subseteq M$, such that $\phi(\alpha)\phi(\beta) \subseteq \phi(\alpha\beta)$ when $\alpha \in Hom(o \rightarrow o'), \beta \in Hom(o' \rightarrow o'')$, and $1_M \in \phi(1_o)$ where $1_o \in Hom(o\rightarrow o)$ is the local identity arrow.
\end{defin}
This definition matches the definition of relational morphisms in \cite{tilson1987categories} in the special case where the target is a monoid; we add the qualifier ``algebraic'' to highlight that type structure does not yet enter its definition.
We do the same for division:
\begin{defin}\label{def:algebraic-division}
An \emph{algebraic division} \(\psi : X\preceq N\) is an algebraic relational morphism $\psi : X \relml N$ where, for any two distinct arrows $\alpha, \alpha' \in Hom(o \rightarrow o')$, we have $\psi(\alpha) \cap \psi(\alpha') = \emptyset$.
\end{defin}
We expand this definition of division to the case where type structure is present:

  \begin{defin}\label{def:typed-division}
  Let \(C\) be a finite category and let \(S=(S,\types S,\units S)\) be a typed
  monoid.  A \emph{typed division}
  \(
  \psi:C\preceq S
  \)
  is an algebraic division \(\psi:C\relml S\), such that, for each
  homset \(H=\Hom_C(o \rightarrow o')\), for each $\alpha \in H$, there is a type $T \in \types{S}$ such that $\psi(\alpha) = T \cap \psi(H)$.
  \end{defin}

We first note that, in the case of finite monoids with discrete types (i.e., each subset of the monoid is a type), the notions of typed relational morphisms and typed divisions reduce to the usual definitions from \cite{tilson1987categories}. 
We can further link the definition to typed recognition,\footnote{We note that there is also a notion of division in \cite{krebs2008typed}, but here we link our relational definition directly to language recognition.} considering the setting of a finite monoid $M$ such that any language it recognizes (in the ordinary sense) via the surjective morphism $\eta$ is also recognized (in the typed sense) by $S$. Then:
\begin{lemma}\label{lem:syntactic-monoid-typed-divides-recognizer}
Let $M$ be a finite monoid, and let $S$ be a typed monoid.
    Let \(\eta:\Sigma^*\twoheadrightarrow M\) be a surjective morphism, and \(h:\Sigma^*\to S\) a morphisms.
  Assume that for every \(P\subseteq M\), there is \(T_P\in\types S\) such that
  \[
    \eta^{-1}(P)=h^{-1}(T_P).
  \]
  Then \(M\preceq S\) in the sense of Definition \ref{def:typed-division}, viewing $M$ as a single-object category.
  \end{lemma}

\begin{proof}
    Consider the relation $\psi \subseteq M \times S$ defined by $\psi :=  h \circ \eta^{-1}$. Because $\eta$ is surjective, this is an algebraic relational morphism.
    The assumption also ensures that it is an algebraic division.
    Then, for any $\alpha \in M$, there is a type $T_\alpha \in \types{S}$ such that $\eta^{-1}(\alpha) = h^{-1}(T_\alpha)$. This entails $\psi(\alpha) = h(\eta^{-1}(\alpha)) = h(h^{-1}(T_\alpha)) = T_\alpha \cap \psi(M)$.
\end{proof}
Thus, our definitions here are compatible with the relevant pre-existing definitions from prior work.

We now verify compatibility of our extended notion of division with composition.

\begin{lemma}\label{lem:typed-division-compose}
Let \(C\) be a finite category, let \(M\) be a finite monoid, and let \(S=(S,\types{S},\units{S})\) be a typed monoid. 
Assume that $M$ and $S$ satisfy the assumptions of Lemma \ref{lem:syntactic-monoid-typed-divides-recognizer}.
Assume \(C\preceq M\) (in the sense of \citet{tilson1987categories}). %
Then \(C\preceq S\) (in the sense of Definition \ref{def:typed-division}).
\end{lemma}
We will apply this to the setting where $M$ is the syntactic monoid of a regular language that is definable in C-RASP; this then provides an iterated wreath product of $\Z$ recognizing all languages recognized by the syntactic morphism of $M$.

\begin{proof}

Let \(\phi\colon C\relml M\) be an algebraic division and let \(\theta\colon M\relml S\) be a typed division. Define their composite relation \(\psi\colon C\relml S\) by
\[
\psi(\alpha)=\bigcup_{m\in \phi(\alpha)}\theta(m)
\]
for each arrow \(\alpha\) of \(C\). 
Because algebraic relational morphisms and divisions compose\cite{tilson1987categories}, \(\psi\) is again an algebraic relational morphism and a division from \(C\) to the underlying monoid of \(S\).

It remains to check that \(\psi\) is also a typed relational morphism.
Because \(M\) is a one-object category, the hypothesis that \(\theta\) is typed yields a single ambient set \(A\subseteq S\) such that for every \(m\in M\) there is a type \(\type{m}\in\types{S}\) with
\[
\theta(m)=A\cap \type{m}.
\]
Fix a homset \(H\) of \(C\), and put \(A_H^\psi := \bigcup_{\beta\in H}\psi(\beta)\). For \(\alpha\in H\),
\[
\psi(\alpha)=\bigcup_{m\in\phi(\alpha)}(A_H^\psi\cap \type{m})
=A_H^\psi\cap\Bigl(\bigcup_{m\in\phi(\alpha)}\type{m}\Bigr).
\]
Because \(M\) is finite and \(\types{S}\) is a Boolean algebra, the finite union
\[
\type{a}:=\bigcup_{m\in\phi(\alpha)}\type{m}
\]
is again a type of \(S\). Thus
\[
\psi(\alpha)=A_H^\psi\cap \type{a},
\]
so \(\psi\) is a typed relational morphism.

\end{proof}

\section{Algebraic Characterization (Necessary but not Sufficient Criterion)}\label{app:R_infty}

\RInfinity*
\begin{proof}
    We note that $\vty{R}\wrc\vty{G}\cap\vty{A}$ is exactly the aperiodic monoids with at most one idempotent in each $\mathcal{R}$-class.
    First, we show $\vty{R}^\omega\subseteq\vty{R}\wrc\vty{G}\cap\vty{A}$. 
    Suppose $M\in\vty{R}^\omega$. 
    By substituting $1$ for $y$ in $\vty{R}^\omega$ we obtain that $M$ satisfies
    \[x^\omega x=x^\omega\]
    which implies $M\in\vty{A}$ \citep{pin:LIPIcs.STACS.2009.1856}.
    Next, assume for sake of contradiction that $M\not\in \vty{R}\wrc\vty{G}$. 
    Then there is a $\mathcal{R}$-class of $M$ which contains at least $2$ idempotents \citep[Theorem 3.18]{stiffler1973extension}.
    Let $e_1\neq e_2$ be $\mathcal{R}$-equivalent idempotents, where $e_1=e_2m_2$ and $e_2=e_1m_1$. Then 
    \begin{align*}
        (e_1e_2)^\omega=(e_1e_1m_1)^\omega=(e_1m_1)^\omega=e_2^\omega=e_2
    \end{align*}
    while 
    \begin{align*}
        (e_1e_2)^\omega e_1=e_2e_1=e_2e_2m_2=e_2m_2=e_1.
    \end{align*}
    This implies $(e_1e_2)^\omega\neq (e_1e_2)^\omega e_1$, and thus $M\not\in\vty{R}^\omega$, a contradiction. Thus $M\in\vty{R}\wrc\vty{G}$.

    Now we show $\vty{R}\wrc\vty{G}\cap\vty{A}\subseteq\vty{R}^\omega$. Let $m_1,m_2\in M$. First, we note that $(m_1m_2^\omega)^\omega=(m_2^\omega m_1)^\omega$. This is because 

    \begin{align*}
        (m_1m_2^\omega)^\omega&=(m_1m_2^\omega)^{\omega+1}& \text{by aperiodicity}\\
        &=m_1(m_2^{\omega}m_1)^\omega m_2^\omega \\
        &\leq_{\mathcal{R}} m_1(m_2^\omega m_1)^\omega\\
        &\leq_\mathcal{R} (m_1m_2^\omega m_1)^\omega m_1 \\
        &\leq_{\mathcal{R}}(m_2^\omega m_1)^\omega
    \end{align*}

    and similarly we have $(m_2^\omega m_1)^\omega\leq_{\mathcal{R}} (m_1m_2^\omega)^\omega$.
    Then, because $\mathcal{R}$-equivalent idempotents are equal in $M\in\vty{R}\wrc\vty{G}\cap \vty{A}$.
    Now consider the elements$(m_1m_2^\omega)^\omega$ and $(m_1m_2^\omega)^\omega(m_2^\omega m_1)^\omega$, which are both idempotent.
    We show these are $\mathcal{R}$-equivalent.
     First $(m_1m_2^\omega)^\omega\geq_\mathcal{R} (m_1m_2^\omega)^\omega m_1 \geq_\mathcal{R} (m_1m_2^\omega)^{\omega+1}$, so by aperiodicity $(m_1m_2^\omega)^\omega \equiv_\mathcal{R} (m_1m_2^\omega)^\omega m_1$.
    Then 
    \begin{align*}
       (m_1m_2^\omega)^\omega(m_2^\omega m_1)^\omega  & = (m_1m_2^\omega)^\omega m_2^\omega(m_1m_2^\omega)^{\omega}m_1\\
       & = (m_1m_2^\omega)^\omega m_1 m_2^\omega m_2^\omega(m_1m_2^\omega)^{\omega}m_1\\
       & = (m_1m_2^\omega)^{\omega}m_1\\
    \end{align*}
    Since these are in fact both $\mathcal{R}$-equivalent idempotents, they are equal \citep{stiffler1973extension}. Thus using the above equalities we can conclude 

    $$(m_1m_2^\omega)^\omega=(m_1m_2^\omega)^\omega(m_2^\omega m_1)^\omega=(m_1m_2^\omega)^{\omega}m_1.$$
\end{proof}

\DyckVSRInfty*
\begin{proof}
The first claim is shown using the decision procedure (\cref{thm:dyck-vs-crasp}): every finite monoid in $\CRASP$ can be decomposed into wreath products of $\uone$ and $\dyck_k$ -- in this case $\uone\preceq \dyck_1$.
    In the other direction, $\dyck_k\in\CRASP$, and thus  $\wwpc(\vty{Dy})$ contains only monoids in $\CRASP$, using  \cref{thm:typed_wreath_principle}.
    We use this to show the other strict inclusions.
    \begin{itemize} 
        \item First, $\vty{R}=\wwpc{(\uone)}$. $\uone\in\vty{Dy}$. Then we have that $\CRASP\cap\vty{REG}= \wwpc{(\vty{Dy})}$, so $\vty{R}\subseteq \CRASP\cap\vty{REG}$.
        For a strict separation, $\dyck_1\in\CRASP\cap\vty{REG}$ but $\dyck_1\not\in\vty{R}$ \citep{BRZOZOWSKI198032}.
        \item We have that $\CRASP\cap\vty{REG}= \wwpc{(\vty{Dy}})$.
        Then by \cref{prop:dyck-krohn-rhodes-cyclic} we know $\dyck_k$ divides into a wreath product of $\uone$ and cyclic groups. 
        We can move all $\uone$ factors to the left, obtaining a monoid in $\vty{R}\wrc\vty{G}$ \citep{stiffler1973extension}. 
        Furthermore, since all monoids $\CRASP$ are aperiodic, this results in a monoid in $\vty{R}^\omega$, by \cref{thm:R_infty_idempotents}.
        For a strict separation, the monoid $M((ab+bba)^*)\in \vty{R}^\omega\setminus \CRASP$ (this can be computed by the decision procedure)
        \item First, $\vty{R}^\omega\subseteq \vty{A}\subsetneq \vty{REG}$ by \cref{thm:R_infty_idempotents}.
        The strict separation can be witnessed by $M((ab+aabb)^*)$, which is aperiodic but has two idempotents in a single $\mathcal{R}$-class
        \item $\vty{A}\subsetneq \vty{REG}$ is a standard fact. For instance, $M((b^*ab^*ab^*)^*)$ (i.e. $\Z/2\Z$) witnesses the strict separation.
    \end{itemize}
    
\end{proof}

\section{Automata proof of $\CRASP \cap \vty{REG}$}\label{app:automata-based-proof}
We attempt to track the states of a DFA by a $\CRASP$ program by iterating over the reachability order of the DFA. 
Whenever we encounter a nontrivial strongly connected component, we need to check if the program can be extended to cover it. 
If we succeed for all strongly connected components, the language of the DFA is recognized by a $\CRASP$ program. 
Otherwise, we show that failure at any component entails non-membership in $\CRASP$.
First, we will define a property of DFAs which determines their definability in $\CRASP$. 
Then, we will argue this property is decidable in polynomial time. 

We make the following assumptions, which simplify the presentation of the proof but do not affect the expressivity of $\CRASP$:
\begin{enumerate}
    \item All atoms occur bound within $\countl$ (unbound atoms have a constant truth value since we would evaluate them on an $\eos$ token). 
    \item No atom appears negatively (since atoms are disjoint, we didn't need negation on them anyways).
\end{enumerate}

Throughout, write $\automaton=(Q,\delta,\Sigma)$ for a semiautomaton where
$\delta$ is extended to $\Sigma^*$ as usual. 
We write $t^{q,w,i}$ for
the value of a \CRASP\ term $t$ at position $i$ of $w$ when the run
starts in state $q$, and $q,w,i\models\phi$ for the truth of a formula
at that position.
We write $t^{q,w}$ as shorthand for $t^{q,w,|w|}$.
We write $L(\automaton, q_0,F)$ to denote the language recognized by $\automaton$ with start state $q_0\in Q$ and accepting states $F\subseteq Q$. 

\begin{defin}[SCC]
    $W\subseteq Q$ is a \emph{strongly connected component} of
    $\automaton$ if it is maximal such that for all $q_1,q_2\in W$
    there is $w\in\Sigma^*$ with $\delta(q_1,w)=q_2$.
    $W$ is \emph{trivial} if $W=\{q\}$ and $\delta(q,\sigma)\neq q$ for
    every $\sigma$.
\end{defin}

We will be reasoning about individual SCCs of an automaton, so we formalize what it means to take an SCC out of an automaton.

\begin{defin}[Extraction]
    Let $W$ be an SCC of $\automaton$.  The \emph{extraction of $W$} is
    the semiautomaton $\extr{W}=(W\sqcup S_W,\delta_W,\Sigma)$ where
    \begin{align*}
        S_W&=\{\delta(p,\sigma)\mid p\in W,\ \sigma\in\Sigma\}\setminus W\\
        \delta_W(q,\sigma)&=\begin{cases}
            \delta(q,\sigma) & q\in W\\
            q & q\in S_W.
        \end{cases}
    \end{align*}
\end{defin}

In other words, to extract an SCC you isolate the states of $W$ and all other states reachable in one transition. 
These external states become sink states in the extracted semiautomaton.
If a semiautomaton $\automaton$ is the extraction of a SCC from itself, then we say $\automaton$ is strongly connected.
In the sequel, we will denote by $S_W$ the external sink states of an extracted SCC. 
We define the following helper function, which acts differently on internal states and external sink states of an SCC.

\begin{defin}
    Let $\automaton$ be a semiautomaton, $W$ a SCC, $q_0$ a state of $\automaton$, and $t$ a term of $\CRASP$. We define the following function:
    \begin{align*}
        V_{q_0,q}(t)&=\begin{cases}
            \{t^{q_0,w,i}\mid \delta_W(q_0,w_{\leq i})=q\}  & q\in W\\
            \{t^{q_0,w,i}\mid \delta_W(q_0,w_{\leq i-1})\in W,\delta_W(q_0,w_{\leq i})=q\}  & \text{otherwise}\\
        \end{cases}
    \end{align*}
\end{defin}

In other words, for states in $W$, the function $V_{q_0,q}$ take the set of all realizable counts on a path from $q_0\to q$, and for states immediately outside of $W$ we take the counts realizable by a path $q_0\to q$ which takes a single step from $W$ to $q$.
We will distinguish states in an automaton using terms which realize distinct sets of counts when in each state. 
However, only certain sets of counts can be distinguished by $\CRASP$ formulas (which can be formalized by the ``types'' as in \cref{app:typed_monoids}).

\begin{defin}[Separated]\label{def:separated}
    Let $W$ be an SCC
    .
    Distinct $q_1,q_2\in \extr{W}$ are
    \emph{separated} by a term $t$ if there is some $q_0\in W$ such that $V_{q_0,q_1}(t)\cap  V_{q_0,q_2}(t)=\emptyset$.
    We say $W$ is \emph{separable} if every pair of distinct states in $\extr{W}$ is
    separated by some $t$.
    We say $W$ is \emph{separable} if every pair of distinct states in $\extr{W}$ is
    separated by some $t$.
\end{defin}

\begin{defin}[Balanced]
    The \emph{subterms} of a term $t$ are $t$ itself alongside every term $t'$ occurring inside a subformula
    $[t'\geq C]$.
    A term $t$ is \emph{balanced on $W$} if there is an interval
    $[\alpha,\beta]$ such that $V_{q_0,q}(t)\subseteq[\alpha,\beta]$
    for all $q_0,q\in W$.
    A term $t$ is \emph{totally balanced on $W$} if all subterms $t'\in t$ are balanced on $W$. 
    A formula $\phi$ is balanced on $W$ if every term occurring in
    $\phi$ is balanced on $W$.
\end{defin}

Intuitively, if $\phi$ is not balanced on $W$ then we would see arbitrarily large or small values when checking the counters when the DFA is in states in $W$. 
An important observation is that for balanced terms, $|V(t)|=1$.

\begin{lemma}\label{lem:balanced-single-valued}
    If $t$ is balanced on $W$, then $|V_{q_0,q}(t)|=1$ for all $q_0,q\in Q$.
\end{lemma}
\begin{proof}
    Suppose otherwise that some $V_{q_0,q}(t)$ holds two values $x_1\neq x_2$.
    Since cycles in $W$ must sum to $0$ (otherwise $t$ would be unbalanced), $V_{q,q_0}(t)$ must contain $-x_1$ and $-x_2$. 
    However, this implies the existence of some cycle $q_0\to q_0$ with weight $x_1-x_2\neq 0$. 
\end{proof}

With these notions at hand, we are ready to state the main claim. 
\begin{lemma}[Definability  Criterion]\label{lem:decision-complete}
    Let $\automaton=(Q,\delta,\Sigma)$ be a semiautomaton.  
    The following are equivalent
    \begin{enumerate}
        \item $L(\automaton, q_0, F)$ is recognizable in $\CRASP$ for any $q_0\in Q$ and $F\subseteq Q$.
        \item $\extr{W}$ is separable by totally balanced terms for every SCC $W$ in $\automaton$.
    \end{enumerate}
\end{lemma}
\begin{proof}
    First, $(2)\Rightarrow (1)$ is shown in \cref{lem:decision-complete}.
    Then, $\lnot(2)\Rightarrow \lnot (1)$ is shown in \cref{lem:decision-complete}.
\end{proof}

\begin{lemma}\label{lem:decision-construction}[Construction]
    Suppose every SCC in $\automaton$ is separable by totally balacned terms 
    Fix $q_0\in Q$.
    Then for each $q\in Q$ there exists a $\CRASP$ formula $\phi_q$ such that $w\models \phi_q\iff \delta(q_0,w)=q$.
\end{lemma}
\begin{proof}
    First, for any $q\in Q$ not reachable from $q_0$ we can set $\phi_q=\bot$.
    Otherwise, we induct on the reachability order of SCC's in $\automaton$.
    Consider the first SCC $W$ containing $q_0$. 
    Because $W$ is separable by a totally balanced terms, for each $q_1,q_2\in \extr{w}$ there exists a totally balanced term $t$ that separates them. 
    We will define a formula $\phi_{q_1,\lnot q_2}$ such that strings which land on $q_1$ always satisfy $\phi_{q_1,q_2}$ while strings that land in $q_2$ do not. 
    Since the sets of counts upon landing in each state form two disjoint finite sets, we can check these counts using a $\CRASP$ formula:
    \begin{align*}
        \phi_{q_1,\lnot q_2}&:= \left(\bigvee_{c\in V_{q_0,q_1}(t)} t=c\right)\land \lnot \left(\bigvee_{c\in V_{q_0,q_2}(t)} t=c\right)
    \end{align*}

    Then, we can let $\phi_q:=\bigwedge_{q'\neq q\in \extr{W}} \phi_{q,\lnot q'}$.
    For the inductive step on an SCC $W'$, we have by assumption formulas $\phi_{q_0'}$ which detect entry into $W'$ (as these are states in $S_W$ from previous SCC's).
    By assumption, there exists $\phi_{q_0',q}$ which detects if $\automaton$ is in $q$ if we started in state $\phi_{q_0'}$
    Since, by assumption, no atom appears negatively, we can take $\phi_q$ and apply the mapping $\sigma\mapsto (\sigma \land \countl \phi_{q_0'}\geq 1)$ to transform $\phi_{q_0',q}\mapsto \hat{\phi}_{q_0',q}$. 
    Intuitively, this transformation tells $\phi_{q_0',q}$ to ignore all symbols which occurred before the first $q_0'$ (i.e. before we entered $W'$ via $q_0'$).
    Now, define the set of entry points $E_{W'}=\{q\mid \delta(q',\sigma)=q, q'\not\in W'\}$, and then we can let $\phi_{q}=\bigvee_{q_0'\in E_{W'}} \hat{\phi}_{q_0',q}$.
\end{proof}

\begin{lemma}[Contradiction]
    Let $\automaton,q_0,F$ define the minimal automaton of a language. Suppose there exists an SCC $W$ in $\automaton$ which is not separable by any totally balanced terms.
    Then $L(\automaton,q_0,F)$ is not definable in $\CRASP$.
\end{lemma}
\begin{proof}
    Let $\phi$ be any $\CRASP$ formula and $C_0$ be larger than any theshold or coefficient in any subformula of the form $\sum\chi\cdot \countl[\psi]\geq C$.
    Since $W$ is not separable by totally balanced terms, there are $q_1,q_2\in W$ such that for all totally balanced terms $t$ and entry points $\delta(q_0,x)=q_\iota$, we have that $V_{q_\iota,q_1}(t)=V_{q_\iota,q_2}(t)$.
    Let $w_1$ and $w_2$ be the minimal strings such that $\delta(q_\iota,w_1)=q_1$ and $\delta(q_\iota,w_2)=q_2$.
    Let $v\in\Sigma^*$ be any distinguishing suffix of $q_1$ and $q_2$ -- such that $\delta(q_1,v)\in F$ while $\delta(q_2,v)\not\in F$ --which always exists as $q_1,q_2$ are distinct states in the minimal automaton. 
    We will prepend very large loops around $q_\iota$, carefully chosen so that all unbalanced terms $\gg C_0$ or $\ll C_0$.
 
    Order all unbalanced terms $t_1, t_2,\ldots, t_k$.
    First, we note that for all unbalanced $t_i$ there exists some loop $u_i\in\Sigma^*$ such that $\delta(q_\iota,u_i)=q_\iota$ and the loop sum $t_i^{q_\iota,u_iu_i}-t_i^{q_\iota,u_i}\neq 0$ -- otherwise, $V_{q_\iota,q_\iota}(t_i)$ would be bounded and $t_i$ would not be unbalanced on $W$. 
    Then, there exists some constant $c$ such that the loop $u:=u_{1}^{c^k}u_2^{c^{k-1}}\cdots u_k^{c}$ has a nonzero sum for all terms (The constant is chosen such that no subsequent loop can ``undo'' the counts accumulated in previous loops, thus guaranteeing that all terms will have a nonzero sum).
    
    There exists a sufficiently large exponent $N$ such that over a long prefix $xu^N$, replacing all unbalanced terms with $\infty$ (or $-\infty$, depending on the sum over $u$) results in a formula $\psi$ whose truth value matches $\phi$ on all positions after the prefix in $xu^Nw_1$ and $xu^Nw_2$.
    Intuitively, any formula $[t\geq C]$ converges to $\top$ or $\bot$ after enough iterations of $u$, and the remaining iterations are used to drown out the finite prefix before convergence. 
    Now, $\psi$ had all unbalanced subformulas replaced with constants, so it is totally balanced. 
    Thus, $t^{q_0,xu^Nw_1}=t^{q_0,xu^Nw_2}$, %
    and appending the suffix $v$ does not change this, so $\phi$ will have the same truth value on $xu^Nw_1v$ as $xu^Nw_2v$ -- therefore $\phi$ cannot define $L(\automaton,q_0,F)$.
\end{proof}

Now we will show that these conditions are decidable in polynomial time.  
We say that two terms are equivalent up to separation of states in the following sense:
Each term defines an equivalence relation $\equiv_{t,q_0}$ over states, where $q_1\approx_{t,q_0} q_2$ whenever $V_{q_0,q_1}(t)=V_{q_0,q_2}(t)$. %

\begin{lemma}\label{lem:poly_balanced_one}[Depth-$1$ Balanced]
    Let $\automaton=(Q,\delta, \Sigma)$ be an SCC. 
    The set of all depth-$1$ $\CRASP$ terms up to equivalent separability of states can be computed in $O(\mathsf{poly}(|Q|,|\Sigma|)$ time.
\end{lemma}
\begin{proof}
    Let $q_0\in Q$.
    For every term $\sum \chi_p\countl[\phi_p]$ occurring in a balanced depth-$1$ we must have that $ \left(\sum \chi_\sigma\countl[\sigma]\right)^{w,|w|}=0$ over any loop $w$, i.e. when $\delta(q_0,w)=q_0$.
    Here, we use the fact that balanced terms are single-valued on each state of an SCC (\cref{lem:balanced-single-valued}).
    First, we assert that such a balanced term exists iff there exists a function $E\colon \Sigma \to \Z$ and $V\colon Q\to \Z$ such that the following condition is satisfied:
    \begin{equation}\label{eq:condition_loop}
        \text{If $\delta(q_1,\sigma)=q_2$ then $V(q_2)=V(q_1)+E(\sigma)$}
    \end{equation}
    If such a function exists, we can obtain a term that sums to $0$ on loops via $\sum E(\sigma)\countl[\sigma]$. 
    If such a term $\sum \chi_\sigma\countl[\sigma]$ exists, we can obtain the function by setting $E(\sigma)=\chi_\sigma$ and $V(q)=\sum_{1\leq i \leq |w|} E(w_i)$ for any $w$ such that $\delta(q_0,w)=q$.
    The only other balanced terms are linear combinations of these terms,
    but these would all be equivalent with respect to separability of states.
   Deciding the existence of such a function can be done in $O(\mathsf{poly}(|Q|,|\Sigma|)$ time. 
    There are $O(|Q|^2|\Sigma|)$ many transitions $\delta(q_1,\sigma)=q_2$ so we obtain $O(|Q|^2\cdot|\Sigma|)$ many constraints of the form $V(q_2)=V(q_1)+E(\sigma)$. 
    These give $O(|Q|^2|\Sigma|)$-many linear constraints over $O(|Q|^2\cdot|\Sigma|+|Q|)$-many variables. 
    This can be solved in $O(\mathsf{poly}(|Q|,|\Sigma|)$ time \citep{doi:10.1137/0208040}.
\end{proof}

\begin{lemma}\label{lem:separable-poly}
    Whether an SCC $W$ is separable by totally balanced terms is decidable in
    $O(\mathsf{poly}(|Q|,|\Sigma|))$ time.
\end{lemma}
\begin{proof}
    We will iteratively refine an equivalence relation $\equiv$ over states, based upon the set of $V_{q_0,q}(t)$ over all terms. 
    At the end, if all states are in their own equivalence class, the SCC is separable by a balanced formula.
    \begin{enumerate}
        \item Initialize $\equiv_0:=W\times W$ by setting all elements to the same class. 
        Initialize a set of $\CRASP$ formulas $\Phi_0:=\emptyset$.
        Fix a start state $q_0$ (since a balanced term is balanced given any start state, the exact choice is immaterial).
        Then perform the loop $(2)-(4)$:
        \item Generate the set $B_k$ of all terms that count over formulas in $\Phi_{k-1}$, up to equivalent separability of states using \cref{lem:poly_balanced_one}.
        \item Denote the indicator for $K\subseteq W$ induced by $t$ as the formula $$\psi_K:= \bigvee_{q\in K}t=V_{q_0,q}.$$
        We generate the collection of indicators of separable state sets by all $t\in B_k$ as the collection $\Psi_k:=\{\psi_K\mid V_{q_0,q_1}(t)=V_{q_0,q_2}(t)\text{ for all $q_1,q_2\in K$}\}$.
        \item Now there are two cases
        \begin{itemize}
            \item If $\Psi_k\not\subseteq \Phi_{k-1}$, then define $\Phi_k$ to refine $\Phi_{k-1}$ via intersection with $\Psi_k$.
        \[\Phi_k:=\{\psi_{K_1}\cap \psi_{K_2}\mid \psi_{K_1}\in \Phi_{k-1}, \psi_{K_2}\in \Psi_k\}.\]
        Next, set $\equiv_k:=\{(q_1,q_2)\mid \text{$q_1,q_2\in K$ for some $\phi_K\in \Phi_{k}$}\}.$
        Then, return to step $(2)$ for iteration $k+1$.
        \item If $\Psi_k\subseteq \Phi_{k-1}$, then we have reached a fixed point. 
        Set $\Phi:=\Phi_k$ and $\equiv:=\equiv_k$ and move to step $(5)$.
        \end{itemize}
        \item If $[q]_{\equiv}$ is a singleton for all $q$, output TRUE and return $\Phi$. Otherwise, output FALSE. 
    \end{enumerate}

    Correctness: Since each $\phi_K\in \Phi$ is balanced, all its subterms are balanced and thus if $[q]_{\equiv}$ are singletons we can obtain totally balanced terms which separate each pair of states. 
    If some $[q]_{\equiv}$ is not a singleton, we know there are two states that no totally balanced term can separate. 
    Thus, the algorithm is correct. 

    Polynomial running time: We crucially maintain that $\equiv_k$ is monotonoically refined at each step, and $\Phi$ keeps only the formulas that define equivalence classes via $\equiv_k$, and thus there are only $O(|Q|^2)$ many of them at each iteration.
    The first step $(1)$ runs in $O(|Q|^2)$ time. 
    Step $(2)$ takes $O(\mathsf{poly}(|Q|,|\Sigma|)$ time by \cref{lem:poly_balanced_one}.
    Step $(3)$ runs in $O(\mathsf{poly}(|Q|,|\Sigma|)$ time by generating in $O(|Q|^2)$ time the equivalence class where $V_{q_0,q_1}(t)=V_{q_0,q_2}(t)$ for each of the $O(\mathsf{poly}(|Q|,|\Sigma|)$ terms $t\in B$.
    Step $(4)$ runs in $O(\mathsf{poly}(|Q|,|\Sigma|)$ time, since the inclusion check is over polynomially sized sets and the intersection only needs to be done via pairs in $\Phi_{k-1}\times \Psi_k$, since each formula identifies a disjoint set of states $K$. 
    Step $(5)$ also just needs to check if the equivalence relation has size $|Q|$, achievable in $O(|Q|)$ time.
    Finally, we only need to iterate the $(2-4)$ loop $O(|Q|^2)$ many times, as the equivalence relation $\equiv_k$ is strictly refined at each step until a fixed-point is reached.
    Thus, this algorithm runs in polynomial time. 
\end{proof}

\begin{thm}
    Checking whether or not every SCC in $\automaton$ is separable by a totally balanced terms is decidable in time polynomial in the size of $\automaton$.
\end{thm}
\begin{proof}
    We iterate over $O(\mathsf{poly}(|Q|,|\Sigma|)$ many SCCs (e.g. using Tarjan's algorithm \citep{doi:10.1137/0201010}), and only need to perform a $O(\mathsf{poly}(|Q|,|\Sigma|)$ time query on each one, as in \cref{lem:separable-poly}.
\end{proof}

\section{Implementation}

A sketch of the algorithm used in the python implementation we have provided.
For simplicity the implementation uses a looser version the algorithm which enumerates all loops in a strongly conected component (of which there may be exponentially many), though one could modify it to strictly be a polynomial time algorithm, as proven above. 
For each simple loop in a strongly connected component, the helper function \textsc{BalancedLabelsBasis} computes the basis of all morphisms into $\Z$ given the constraint that all loops must be ``balanced'' (i.e. have the images of the symbols sum to $0$).
By computing the nullspace of this basis, we find all possible morphisms that sum to $0$ on loops, and iteratively relable the symbols according to these morphisms. 
At the end, the helper function \textsc{Separated} checks if all states have differing sets of outgoing transition labels (hence the states can be distinguished only using iterated counting -- and thus being expressible in $\CRASP$).

\begin{algorithm}
    \caption{$\CRASP$ Membership}
    \begin{algorithmic}[1]
        \Procedure{DecideMembership}{$\automaton$}
            \State $\text{SCCs} \gets\textsc{StronglyConnectedComponents}(\automaton)$
            \ForAll{$SCC \in \text{SCCs}$}
                \State $B_{\text{prev}} \gets \emptyset$
                \State $B_{\text{curr}} \gets \textsc{BalancedLabelsBasis}(SCC)$
                \If{$\textsc{Nullspace}(B_{\text{cur}})\neq\emptyset$}
                    
                    \While{$B_{\text{curr}}\not\in\textsc{Span}(B_{\text{prev}})$}
                        \State $labels[\sigma_i] \gets [v_i]_{v\in B_{\text{curr}}}$
                        \ForAll{$loop\in \textsc{LoopsAround}(q_0,SCC)$}
    
                            \ForAll{$0\leq i\leq \textsc{Len}(loop)$}
                                \State $(q_i,\sigma_i, q_{i+1}) \gets loop[i]$
                                \State $SCC.\delta.\textsc{Replace}\left((q_i,\sigma, q_{i+1}),\left(q_i,\left[\sigma_i,\sum_{1\leq j\leq i}labels[\sigma_j]\right],q_{i+1}\right)\right)$
                            \EndFor
                            \State $B_{\text{prev}} \gets B_{\text{curr}}$
                            \State $B_{\text{curr}} \gets \textsc{BalancedLabelsBasis}(SCC)$
                        \EndFor
                    \EndWhile
                \EndIf
                
                \If{$\lnot\textsc{Separated}(SCC)$}
                    \State \Return $\bot$
                \EndIf
            \EndFor
            \State \Return $\top$
        \EndProcedure

    \end{algorithmic}
\end{algorithm}

\end{document}